\documentclass[11pt,onecolumn]{article}
\usepackage[letterpaper, margin=1in]{geometry}
\usepackage[T1]{fontenc}
\usepackage{times}
\usepackage{euscript}
\usepackage{graphicx}
\usepackage[utf8]{inputenc}
\usepackage{wrapfig}
\usepackage{xspace}
\usepackage{float}
\usepackage{marginnote}
\usepackage{enumerate}
\usepackage{enumitem}
\usepackage{epstopdf}
\usepackage[dvipsnames]{xcolor}
\usepackage[mathlines]{lineno}
\usepackage{algorithm}
\usepackage{algorithmic}
\usepackage{hyperref}

\usepackage{amsmath,amsfonts,amssymb,amsthm, bm}
\usepackage{thm-restate}
\usepackage{yfonts}
\usepackage{textcomp}

\def\1{\bm{1}}

\DeclareMathAlphabet{\mathsfit}{\encodingdefault}{\sfdefault}{m}{sl}
\SetMathAlphabet{\mathsfit}{bold}{\encodingdefault}{\sfdefault}{bx}{n}

\def\mcB{{\mathcal{B}}}
\def\mcC{{\mathcal{C}}}

\def\mcG{{\mathcal{G}}}
\def\mcH{{\mathcal{H}}}

\def\mcP{{\mathcal{P}}}
\def\mcQ{{\mathcal{Q}}}

\def\mbE{{\mathbb{E}}}

\def\mbG{{\mathbb{G}}}

\def\mbR{{\mathbb{R}}}

\def\mbZ{{\mathbb{Z}}}

\newcommand{\netcost}{\phi}

\newcommand{\freeofcell}[2]{{#1}_{#2}^F}

\newcommand{\cell}{\Box}

\newcommand{\excess}{\mathrm{exc}}

\newcommand{\distancetocell}[2]{\mathrm{d}(#1, #2)}
\newcommand{\distance}[2]{\mathrm{d}(#1, #2)}

\newcommand{\margin}[2]{\mathcal{S}_{#1}^{#2}}
\newcommand{\innermargin}[2]{\mathcal{S}^{\downarrow}_{#1}(#2)}
\newcommand{\outermargin}[2]{\mathcal{S}^{\uparrow}_{#1}(#2)}

\newcommand{\innermarginnum}[2]{n_{\innermargin{\cell}{\delta}}}
\newcommand{\outermarginnum}[2]{n_{\outermargin{\cell}{\delta}}}

\newcommand{\ignore}[1]{}

\newcommand{\servers}{\eta}
\newcommand{\requests}{\sigma}

\newcommand{\reverse}{\textsc{ReverseHungarianSearch}}
\newcommand{\prob}[1]{\mathrm{Pr}\left[#1\right]}

\newcommand{\plan}{\Sigma}
\newcommand{\divider}[1]{\Gamma_{#1}}

\newcommand{\pq}{\mathrm{PQ}}

\newcommand{\current}{\mcC}
\newcommand{\ymax}{\varphi}
\newcommand{\diam}{\mathrm{D}}

\newtheorem{lemma}{Lemma}[section]
\newtheorem{theorem}[lemma]{Theorem}

\newtheorem{corollary}[lemma]{Corollary}

\newcommand{\seqcost}{\text{\textcent}}

\usepackage{etoolbox} 

\newcommand*\linenomathpatch[1]{%
  \cspreto{#1}{\linenomath}%
  \cspreto{#1*}{\linenomath}%
  \csappto{end#1}{\endlinenomath}%
  \csappto{end#1*}{\endlinenomath}%
}

\linenomathpatch{equation}
\linenomathpatch{gather}
\linenomathpatch{multline}
\linenomathpatch{align}
\linenomathpatch{alignat}
\linenomathpatch{flalign}

\usepackage[square,numbers]{natbib}

\newcommand\myparagraph[1]{\paragraph{#1}}

\title{\Large \bf Geometric Bipartite Matching Based Exact Algorithms for Server Problems}
\author{Sharath Raghvendra\\{\small North Carolina State University}
\and Pouyan Shirzadian\\{\small Virginia Tech}
\and Rachita Sowle\\{\small Virginia Commonwealth University}}

\date{}

\begin{document}

\maketitle

\begin{abstract}
For any given metric space, obtaining an offline optimal solution to the classical $k$-server problem can be reduced to solving a minimum-cost partial bipartite matching between two point sets $A$ and $B$ within that metric space. 

For $d$-dimensional $\ell_p$ metric space, we present an $\tilde{O}(\min\{nk, n^{2-\frac{1}{2d+1}}\log \Delta\}\cdot \Phi(n))$ time algorithm for solving this instance of minimum-cost partial bipartite matching; here, $\Delta$ represents the spread of the point set, and $\Phi(n)$ is the query/update time of a $d$-dimensional dynamic weighted nearest neighbor data structure. Our algorithm improves upon prior algorithms that require at least $\Omega(nk\Phi(n))$ time. The design of minimum-cost (partial) bipartite matching algorithms that make sub-quadratic queries to a weighted nearest-neighbor data structure, even for bounded spread instances, is a major open problem in computational geometry. We resolve this problem at least for the instances that are generated by the offline version of the $k$-server problem. 

Our algorithm employs a hierarchical partitioning approach, dividing the points of $A\cup B$ into rectangles. It maintains a minimum-cost partial matching where any point $b \in B$ is either matched to a point $a\in A$ or to the boundary of the rectangle it is located in. The algorithm involves iteratively merging pairs of rectangles by erasing the shared boundary between them and recomputing the minimum-cost partial matching. This continues until all boundaries are erased and we obtain the desired minimum-cost partial matching of $A$ and $B$. We exploit geometry in our analysis to show that each point participates in only $\tilde{O}(n^{1-\frac{1}{2d+1}}\log \Delta)$ number of augmenting paths, leading to a total execution time of $\tilde{O}(n^{2-\frac{1}{2d+1}}\Phi(n)\log \Delta)$.

We also show that, for the $\ell_1$ norm and $d$ dimensions, any algorithm that can solve instances of the offline $n$-server problem with an exponential spread in $T(n)$ time can be used to compute minimum-cost bipartite matching in a complete graph defined on two $(d-1)$-dimensional point sets under the $\ell_1$ norm within $T(n)$ time. This suggests that removing spread from the execution time of our algorithm may be difficult as it immediately results in a sub-quadratic algorithm for bipartite matching under the $\ell_1$ norm.  
\end{abstract}
\section{Introduction}\label{sec:introduction}

This paper considers two classical optimization problems: the \emph{offline $k$-server} problem and the \emph{minimum-cost bipartite matching} problem in geometric settings. 

\myparagraph{Offline $k$-server problem and its variant.} Consider a sequence of requests $\varsigma = \langle r_1,\ldots, r_m \rangle$ in a metric space equipped with the cost function $\distance{\cdot}{\cdot}$. The cost of a single server \emph{servicing} the requests in $\varsigma$ is the sum of the distances between every consecutive pair of points in the sequence, i.e.,  $\seqcost(\varsigma)=\sum_{i=1}^{m-1} \distance{r_i}{r_{i+1}}$. Given a sequence $\requests = \langle r_1,\ldots r_n\rangle$ of $n$ requests and an integer $1 \le k \le n$, the \emph{$k$-sequence partitioning problem} (or simply the $k$-SP problem) requires partitioning the requests in $\requests$ into $k$ subsequences $\varsigma_1, \ldots, \varsigma_k$ so that $\sum_{i=1}^{k}\seqcost(\varsigma_i)$ is minimized. The optimal solution for the $k$-sequence partitioning problem is the cheapest way for $k$ servers to service all of the requests in $\requests$. 
We also consider a variant: the \emph{$k$-sequence partitioning with initial locations problem} (or the $k$-SPI problem). Here,
 in addition to the requests $\requests$, we are also given the initial locations $\eta=\langle s_1,\ldots, s_k\rangle$ for the $k$ servers. The objective of this problem is to partition $\requests' = \eta \requests$ into $k$ subsequences $\varsigma_1,\ldots, \varsigma_k$ so that the initial location $s_j$ of server $j$ appears in the first element of the subsequence $\varsigma_j$ and the cost $\sum_{i=1}^{k}\seqcost(\varsigma_k)$ is minimized. 
The optimal solution to the $k$-SPI problem is the offline optimal solution to the well-known $k$-server problem. 

We assume that the requests in $\requests$ and the initial locations in $\eta$ are scaled and translated so that they are contained inside the unit hypercube $[0,1]^d$. Such scaling and translation do not impact the optimal solutions for the $k$-SP and $k$-SPI problems. We define the diameter $\mathrm{Diam}(\requests)$ to be the largest distance between any two request locations in $\requests$ and the \emph{closest pair distance}, denoted by $\mathrm{CP}(\requests)$, to be the smallest non-zero distance between any two requests in $\requests$. The \emph{spread}, denoted by $\Delta$, is the ratio of the diameter to the closest-pair distance, i.e., $\Delta=\mathrm{Diam}(\requests)/\mathrm{CP}(\requests)$. 

\myparagraph{Minimum-cost bipartite matching.} Consider a weighted bipartite graph $G(A\cup B, E\subseteq A\times B)$, where each edge between $u$ and $v$ is assigned a real-valued cost. Let $n:=\min\{|A|, |B|\}$. A \emph{matching} of size $t\le n$, or simply a \emph{$t$-matching}, is a set of $t$ edges in $G$ that are vertex-disjoint. The \emph{cost} of any matching $M$ is the sum of the costs of its edges. Given a parameter $t \le n$, the \emph{minimum-cost bipartite $t$-matching} problem seeks to find a $t$-matching with a minimum cost. 
When $|A|=|B|=n$, a matching of size $n$ is called a \emph{perfect matching}. When $A$ and $B$ are $d$-dimensional point sets and the cost of any edge between $a \in A$ and $b \in B$ is the $\ell_p$ distance $\|a-b\|_p$, the problem of finding the minimum-cost bipartite $t$-matching is also called the (partial) \emph{geometric bipartite matching problem}.


\myparagraph{Relating the two problems.}
\citet{chrobak1991new} established a reduction from the minimum-cost bipartite $t$-matching problem to the $k$-SPI problem.
\begin{lemma}\cite[Theorem 11]{chrobak1991new}
\label{lem:spimatch}
Any algorithm that computes an optimal solution to the $n$-SPI in an arbitrary metric space in $T(n)$ time can also find, in $T(n)+O(n^2)$ time, a minimum-cost perfect matching in any complete bipartite graph with real-valued costs.
\end{lemma}

We strengthen the connection between the two problems by showing a reduction in the reverse direction, i.e., we reduce the $k$-SP (resp. $k$-SPI) problem to the minimum-cost bipartite $t$-matching problem.
Given an input sequence $\requests$ of requests to the $k$-SP problem, we construct a bipartite graph $\mathcal{G}_{\requests}$ with a vertex set $A\cup B$ and a set of edges $\mathcal{E}$ as follows. 
{\it Vertex Set:} For each request $r_i$, we create a vertex $b_i$ (resp. $a_i$) in $B$ (resp. $A$) and designate it as the \emph{entry} (resp. \emph{exit}) gate for request $r_i$. 
{\it Edge Set:} The exit gate $a_i$ of request $r_i$ is connected to the entry gate $b_j$ of every subsequent request $r_j$ with $j > i$ with an edge. The cost of this edge is $\distance{a_i}{b_j}=\|r_i-r_j\|_p$. 
It is easy to see that a minimum-cost $(n-k)$-matching $M$ in $\mathcal{G}_\requests$ can be used to find an optimal solution to the $k$-SP problem. See Figure~\ref{fig:problem}. 

For the $k$-SPI problem, for each server $j$, we add a vertex $a^j$ at the initial location $s_j$ to $A$ and connect $a^j$ to the entry gate $b_i$ of every request $r_i$ in $\requests$. The cost of this edge is $\distance{a^j}{b_i}=\|s_j-b_i\|_p$. A minimum-cost $n$-matching in this graph can be converted to an optimal solution for the $k$-SPI problem. The formal proof of correctness for both these reductions is presented in Appendix~\ref{sec:background}.

A different reduction from $k$-SP and $k$-SPI problems to the minimum-cost flow problem has been presented in previous works~\cite{chrobak1991new, rudec2009fast, rudec2013new}, leading to the development of $O(n^2k)$ time algorithm. However, unlike our reduction, their approach generates instances that include edges whose costs are $-\infty$ and fails to maintain the metric properties of costs.
\begin{figure}
    \centering
    \includegraphics[width=0.9\textwidth]{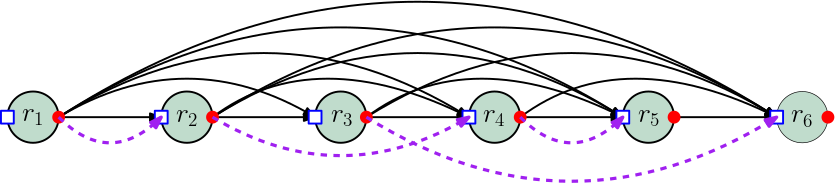}
    \caption{The graph $\mcG_\requests$ constructed for $\requests=\langle r_1, r_2, \ldots, r_6\rangle$. The vertex set $A$ (red disks) and $B$ (blue squares) represent the exit and entry gates of each request, and the purple dashed lines show an $(n-2)$-partial matching on $\mcG_\requests$ representing a $2$-partitioning $\langle r_1, r_2, r_4, r_5\rangle$ and $\langle r_3, r_6\rangle$.}
    \label{fig:problem}
\end{figure}


\begin{lemma}\label{lemma:matching-reduction} 
An optimal solution for an instance $\requests$ (resp. $\requests'$) of $k$-SP (resp. $k$-SPI)  problem can be found by computing a minimum-cost matching of size $n-k$ (resp. $n$) in $\mathcal{G}_{\requests}$ (resp. $\mathcal{G}_{\requests'}$).
\end{lemma}


Finally, we extend the reduction~\citet{chrobak1991new} to geometric settings and provide a reduction from the geometric minimum-cost matching problem under the $\ell_1$-norm to the geometric version of the $n$-SPI problem. 
This reduction, however, creates an instance of the $n$-SPI problem with a spread of $3^n$ (See Appendix~\ref{sec:reduction} for details). Lemma~\ref{lemma:kSPI_reduction} follows directly from this reduction.  



\begin{lemma}\label{lemma:kSPI_reduction}
Any algorithm that can solve the $d$-dimensional $n$-SPI problem under the $\ell_1$ costs in $T(n)$ time can also be used to solve an instance of minimum-cost bipartite matching under the $\ell_1$ costs on a complete graph in $T(n)$ time.
\end{lemma}


\myparagraph{Related work.}
For a graph with $m$ edges and $n$ vertices, the classical Hungarian algorithm computes a minimum-cost $t$-matching for all values of $0\le t\le n$~\cite{kuhn1955hungarian, phatak2022computing}. The algorithm begins with an empty matching $M$, and in each iteration $i$, updates a minimum-cost $(i-1)$-matching to a minimum-cost $i$-matching in $O(m+n\log n)$ time by finding a \emph{minimum net-cost augmenting path}, i.e., an augmenting path that increases the matching cost by the smallest value. The overall execution time of the Hungarian algorithm is $O(nm+n^2\log n)$, or $O(n^3)$ when $m = \Theta(n^2)$.  
Despite substantial efforts, this remains the most efficient algorithm for the problem. Notable exceptions include specialized cases, such as graphs with small vertex separators~\cite{lipton1980applications} as well as graphs where the edge weights are integers~\cite{chen2022maximum,gt_sjc89,gt_jacm91}. 

In geometric settings, Vaidya showed that each iteration of the Hungarian algorithm can be implemented in $\tilde{O}(n\Phi(n))$ time, where $\Phi(n)$ represents the query/update time of a dynamic weighted nearest neighbor (DWNN) data structure with respect to the edge costs. Thus, the minimum-cost $t$-matching can be computed in $O(n^2 \Phi(n))$ time, which is sub-cubic in $n$ provided $\Phi(n)$ is sub-linear.  For instance, for the $\ell_1$ norm and $d$ dimensions, $\Phi(n)=O(\log^d n)$~\cite{vaidya1989geometry} and for the $\ell_2$ norm and $2$ dimensions, $\Phi(n)=\log^{O(1)}n$~\cite{kaplan2020dynamic}. In these cases, the Hungarian algorithm can be implemented in near-quadratic time.  

The design of algorithms that compute a minimum-cost matching with a sub-quadratic number of queries to a DWNN data structure remains an important open problem in computational geometry. There are three notable exceptions to this. First, for points with integer coordinates and the $\ell_1$ and $\ell_\infty$ norms, the edge costs are integers. Using this fact,~\citet{ra_soda12} adapted an existing cost-scaling-based graph algorithm~\cite{gt_sjc89} and presented an algorithm that executes in $\tilde{O}(n^{3/2}\Phi(n))$ time. Second,~\citet{s_socg13} extended this result to two-dimensional point sets with integer coordinates and the $\ell_2$ costs. Their result, however, relies on the points being planar as well as the edge costs being the square root of integers and does not extend to $d$-dimensional points with real-valued coordinates.
Third, \citet{gattani2023robust} presented a divide-and-conquer algorithm (GRS algorithm) to compute a minimum-cost perfect matching. The worst-case execution time of their algorithm is $\tilde{O}(n^2\Phi(n)\log\Delta)$; here, $\Delta$ is the spread of the point sets $A$ and $B$.
For \emph{stochastic} points sampled from an unknown distribution $\mu$ in two dimensions, the algorithm finds the minimum-cost matching in $\tilde{O}(n^{7/4}\Phi(n)\log \Delta))$ time in expectation. 

\begin{figure}
    \centering
    \includegraphics[width=0.32\linewidth]{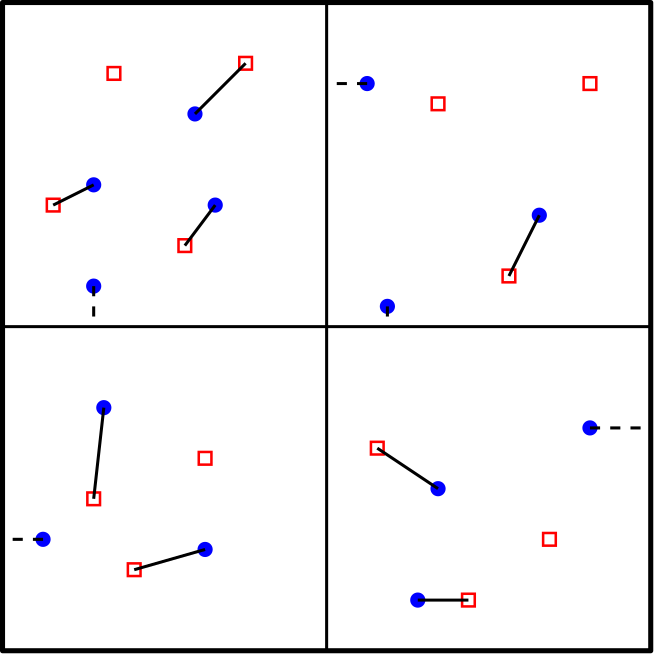} \hspace{0.5em}\includegraphics[width=0.32\linewidth]{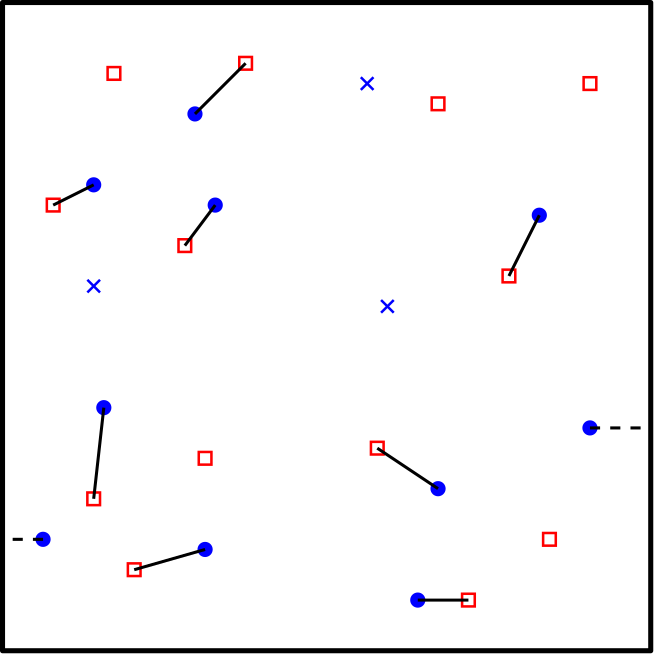} \hspace{0.5em}\includegraphics[width=0.32\linewidth]{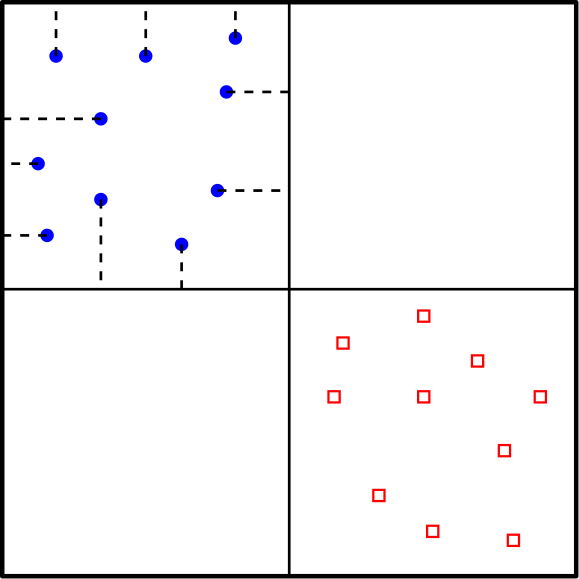}
    \caption{Illustration of the GRS algorithm.}
    \label{fig:GRS}
\end{figure}

The GRS algorithm uses a (randomly-shifted) quadtree to divide the problem into smaller sub-problems. Within each square $\cell$ of the quadtree, their algorithm recursively computes a \emph{minimum-cost extended matching}, where each point of $B$ inside $\cell$ can either match to another point of $A$ inside $\cell$ or to the boundary of $\cell$ (Figure~\ref{fig:GRS}(left)). The algorithm combines minimum-cost extended matchings of the child squares by  \emph{erasing} their common boundary, freeing all points of $B$ that matched to this boundary (Figure~\ref{fig:GRS}(middle)), and iteratively matching the freed points. 

Interestingly, \citet{gattani2023robust} showed that when $A$ and $B$ are samples from the same distribution $\mu$, most points of $B$ will match to a close-by point in $A$, leaving only $\tilde{O}(n^{3/4})$ many points in $B$ matching to the boundary. Thus, erasing the common boundary only creates $\tilde{O}(n^{3/4})$ free points, each of which can be matched in $O(n\Phi(n))$ time, leading to an overall execution time of $\tilde{O}(n^{7/4}\Phi(n))$. However, this efficiency does not extend to arbitrary point sets. For instance, suppose all points of $A$ are in a child $\cell_1$ and all points of $B$ are in another child $\cell_2$. In this case, all edges of the minimum-cost matching cross the boundaries of $\cell_1$ and $\cell_2$. The minimum-cost extended matching at $\cell_2$, therefore, will match every point of $B$ to the boundary of $\cell_2$ (See Figure~\ref{fig:GRS}(right)). Erasing the common boundary between $\cell_1$ and $\cell_2$ creates $n$ free points, causing the conquer step to take $\Omega(n^2\Phi(n))$ time. Furthermore, unlike the Hungarian algorithm, the GRS algorithm does not guarantee the optimality of intermediate matchings. Therefore, it cannot be used to produce minimum-cost $t$-matchings.

\myparagraph{Our results.} 
The optimal solutions to the $k$-SP (resp. $k$-SPI) problems can be computed in $\tilde{O}(nk\Phi(n))$ time by non-trivially adapting the Hungarian algorithm to find a minimum-cost $(n-k)$-matching (resp. $n$-matching) in $\mcG_{\requests}$ (resp. $\mcG_{\requests'}$). 
However, it is worth noting that this algorithm still makes quadratic queries to the DWNN data structure when $k = \Theta(n)$. The main contribution of this paper is the design of a novel algorithm that, for any $k$, computes the optimal solution to the $k$-SP and $k$-SPI problems while making only a sub-quadratic number of queries to the DWNN data structure. 



\begin{theorem}\label{thm:kSP}
Given any sequence $\requests$ (resp. $\requests'=\servers\requests$) of $n$ requests (resp. $n$ requests and $k$ initial server locations) in $2$ dimensions with a spread of $\Delta$, and a value $1 \le k \le n$, there exists a deterministic algorithm that computes the optimal solution for the instance of $k$-SP (resp. $k$-SPI) problem under the $\ell_p$ norm in $\tilde{O}(\min\{nk, n^{1.8}\log \Delta\}\cdot\Phi(n))$ time.
\end{theorem}

Our algorithm also extends to higher dimensions, and for any dimension $d\ge 2$, it computes optimal solutions to the $k$-SP and $k$-SPI problems in $\tilde{O}(\min\{nk, n^{2-\frac{1}{2d+1}}\log \Delta\}\cdot\Phi(n))$ time (Theorem~\ref{thm:kSP-d} in Appendix~\ref{sec:appendix-high-d-ksp}). 



Developing algorithms for geometric bipartite matching that perform sub-quadratic queries to a DWNN data structure is a challenging task. Nevertheless, it follows from Theorem~\ref{thm:kSP} that the bipartite matching instances generated by $k$-SP and $k$-SPI problems (with bounded spread) can be solved by making sub-quadratic queries to a DWNN data structure. Solving instances with unbounded spread with sub-quadratic queries to DWNN, at least for the $k$-SPI problem, remains challenging. This is because, as established in Lemma~\ref{lemma:kSPI_reduction}, any such  
algorithm can find the $d$-dimensional minimum-cost bipartite matching algorithm under $\ell_1$ costs in sub-quadratic time.  


Our algorithm uses a hierarchical partitioning tree, whose nodes (referred to as cells) are axis-parallel rectangles with an aspect ratio of at most $3$. Each cell is divided into two smaller rectangles, forming its children. At any point during the execution of the algorithm, it maintains a set of \emph{current} cells $\mcC$ that partition the input points. It computes a minimum-cost extended $(n-k)$-matching where points of $B$ are allowed to match to the boundaries of these cells by repeatedly identifying a minimum net-cost augmenting path and augmenting the matching along this path.  Once a minimum-cost extended $(n-k)$-matching is computed, the algorithm removes a pair of sibling rectangles from the set of current cells and instead makes their parent cell current. By doing so, the common boundary between the siblings is erased, creating additional free points, which are then matched again by finding the minimum net-cost augmenting paths. When all the boundaries are erased, the algorithm terminates with the desired minimum-cost $(n-k)$-matching. 

There are three major hurdles in proving the efficiency of our algorithm. The first challenge is in finding a minimum net-cost augmenting path, which typically requires a search on the entire graph. However, for extended matchings, we show that the minimum net-cost augmenting path is fully contained inside one of the current cells.  This significantly improves efficiency, as the search can be limited to individual current cells, and the overall minimum can then be determined by selecting the path with the smallest net-cost among all current cells.
The second challenge is in bounding the time to merge two cells. Similar to the worst-case example for the GRS algorithm, the $k$-SP (or $k$-SPI) may have current cells with $\Theta(n)$ points, where the optimal solution matched each point of $B$ to a point of $A$ that is very far (see Figure~\ref{fig:ksp} (left)). Despite this, we show that any minimum-cost extended matching has only $O(n^{0.8})$ points that are matched to their common boundary, helping us in bounding the time required to merge cells. To establish this, we critically use the fact that every sub-problem has a hidden low-cost high-cardinality matching $M$ (with only sub-linearly many free points, see Figure~\ref{fig:ksp} (right)). 

The final challenge in the analysis is the following. The $k$ free points associated with the extended $(n-k)$-matching $M$ maintained by our algorithm may differ from the $k$ free points in the minimum-cost $(n-k)$ matching $M^*$. For instance, a point $b \in B$ may be free in $M^*$ but matched to the boundary in $M$, thereby leaving some other point $b'$ unmatched. Our algorithm corrects them via alternating and augmenting paths, each of which can take $\Theta(n)$ time to find. Since there can be $k$ such corrections, a na\"ive analysis leads to an upper bound of $O(nk)$. Surprisingly, by exploiting geometry, we show that the number of such corrections inside any current cell cannot exceed $O(n^{0.8})$. Using this observation, we bound the overall execution time of our algorithm by $\tilde{O}(n^{1.8}\Phi(n)\log \Delta)$.

\begin{figure}
    \centering
    \includegraphics[width=0.45\textwidth]{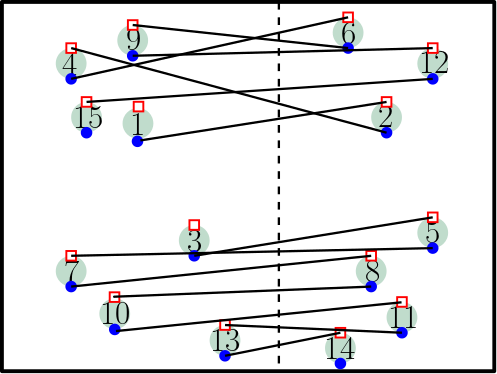}
    \hspace{2em}
    \includegraphics[width=0.45\textwidth]{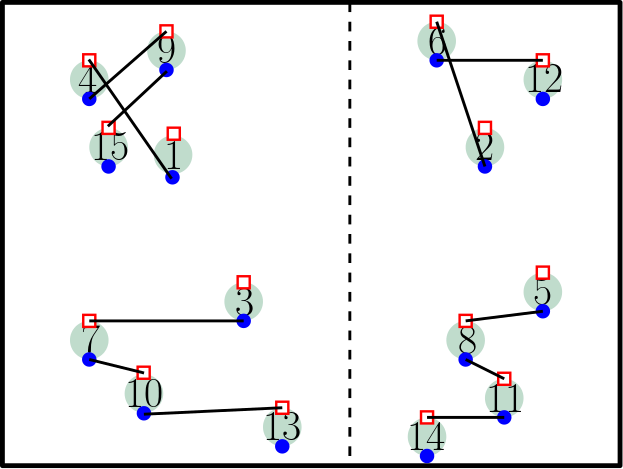}
    \caption{(left) A sub-problem of the $k$-SP problem, where the optimal solution has a high cost, and (right) there exists a low-cost high-cardinality matching inside the sub-problems.}
    \label{fig:ksp}
\end{figure}




Our proof techniques also improve the analysis of the GRS algorithm for the case where the cost between points $a \in A$ and $b \in B$ is given by $\|a-b\|_2^q$ for any $q \ge 2$. Previous work established that the GRS algorithm computes the minimum-cost perfect matching between two-dimensional stochastic points in $\tilde{O}(n^{2-\frac{1}{2(q+1)}}\Phi(n))$ time~\cite{gattani2023robust}. Notably, as $q$ approaches $\infty$, this analysis implies a quadratic number of queries to the DWNN data structure.
\citet{raghvendra2024new} recently observed that, for any $q\ge1$, there is a low-cost high-cardinality matching with only sub-linearly many free points between stochastic point sets. By incorporating this observation with the novel analysis techniques in our paper, we show that the GRS algorithm makes only a sub-quadratic number of queries to the DWNN data structure, regardless of the value of $q$. 
  

\begin{restatable}{theorem}{GRSAnalysis}\label{lemma:GRS-analysis}
    Suppose $U$ is a set of $2n$ points inside the unit square and $A$ is a subset chosen uniformly at random from all subsets of size $n$. Let $B=U\setminus A$.  Then, there exists an algorithm that computes the minimum-cost perfect matching on the complete bipartite graph on $A$ and $B$ under $\ell_2^q$ costs in $\tilde{O}(n^{7/4}\Phi(n)\log \Delta)$ expected time.
\end{restatable}

Note that the setting considered by \citet{gattani2023robust} in the analysis of the GRS algorithm, where the two sets $A$ and $B$ are i.i.d samples from the same distribution, is a special case of randomly partitioning a set of $2n$ points into two sets $A$ and $B$ of $n$ points, considered in Theorem~\ref{lemma:GRS-analysis}.

\paragraph{Organization.} In Section~\ref{sec:geo-primal-dual}, we establish our primal-dual framework for the minimum-cost bipartite matching problem and use it in Section~\ref{sec:k-seq} to present our sub-quadratic algorithms for the $k$-SP and $k$-SPI problems. We present our fast implementation of the Hungarian algorithm in Section~\ref{sec:nk-algod} and the improved analysis of the GRS algorithm in Section~\ref{sec:randomly-colored}. We conclude the paper by discussing some open questions in Section~\ref{sec:conclusion}.

\section{Geometric Primal-Dual Framework}\label{sec:geo-primal-dual}

Let $\sigma$ be an input to the $k$-SP problem, where the distance between two locations $a$ and $b$ is given by $\distance{a}{b}=\|a-b\|_p$. In this section, we introduce a primal-dual framework based on hierarchical partitioning to compute a minimum-cost $(n-k)$-matching in $\mcG_\sigma=(A\cup B, E\subset A\times B)$. We begin by describing the hierarchical partitioning scheme.


\subsection{Hierarchical Partitioning}\label{sec:hierarchical}
Using $\lambda:=9n^{-1/5}$, we construct a hierarchical partitioning $\mcH$ recursively. Each node of $\mcH$ is an axis-parallel rectangle, referred to as a \emph{cell}. The root node, $\cell^*:=[-3n, 3n]^2$, contains all points in $A\cup B$. For each node $\cell$, let $A_\cell$ and $B_\cell$ be the points of $A$ and $B$ inside $\cell$ and let $n_{\cell}=|A_\cell \cup B_{\cell}|$. If $n_{\cell}\le2$ (i.e., $\cell$ is empty or contains the entry and exit gates of a single request), then $\cell$ is marked as a leaf node. Otherwise, we partition $\cell$ into two smaller rectangles as follows. Let $\ell_\cell$ be the larger of the length and width of rectangle $\cell$. Without loss of generality, assume that $\ell_\cell$ is the width of $\cell$ and let $x_{\min}$ be the $x$-coordinate of the bottom-left corner of $\cell$.   For any value $\hat{x}\in [x_{\min}+\frac{\ell_\cell}{3}, x_{\min}+\frac{2\ell_\cell}{3}]$, define $\Lambda(\hat{x}):=\{u\in A_\cell\cup B_\cell: |u_x-\hat{x}|\le \ell_\cell\lambda\}$; here, $u_x$ denotes the $x$ coordinate of the point $u$. Let $x^*:=\arg\min_{\hat{x}\in [x_{\min}+\frac{\ell_\cell}{3}, x_{\min}+\frac{2\ell_\cell}{3}]}|\Lambda(\hat{x})|$. 
We partition $\cell$ into two smaller rectangles by using a vertical line defined by $x=x^*$ and add them as the children of $\cell$ to $\mcH$. We refer to the segment partitioning $\cell$ into its two children as its \emph{divider} and denote it by $\divider\cell$. See Figure~\ref{fig:hierarchical}. For any cell $\cell$, the four sides of its rectangle are defined by the dividers of its ancestor or the boundaries of the root square.
This completes the construction of $\mcH$. Note that the height of the tree is $O(\log n\Delta)$. A simple sweep-line algorithm can compute $x^*$ in $O(n_\cell\log n_\cell)$ time. Using this procedure, we construct $\mcH$ in $\tilde{O}(n \log (n\Delta))$ time.
\begin{figure}
    \centering
    \includegraphics[width=0.6\linewidth]{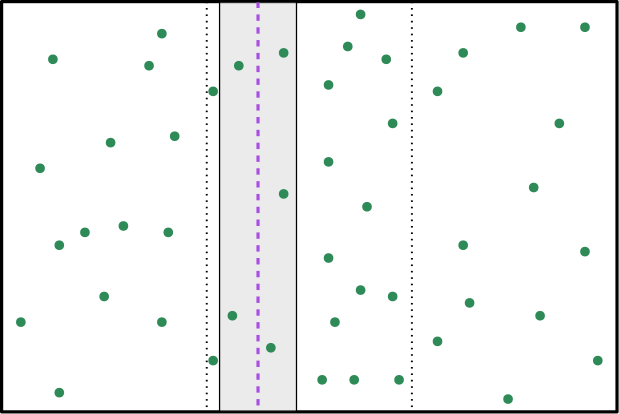}
    
    \caption{We Partition a rectangle into two children by picking a divider (the purple dashed vertical line) with the minimum number of points close to it (gray shaded area) within the middle-third of its longer side (the part between the two vertical dotted lines).}
    \label{fig:hierarchical}
\end{figure}

\begin{restatable}{lemma}{margin}\label{lemma:margin}
    For each cell $\cell$ of $\mcH$, the ratio of the largest to the smallest side of $\cell$ is at most $3$. Furthermore, the number of points of $A_\cell\cup B_\cell$ with a distance smaller than $\ell_\cell\lambda$ to the divider $\divider\cell$ is $O(n_\cell\lambda)$.
\end{restatable}

Recollect that a matching $M$ in $\mcG_\sigma$ is a subset of vertex-disjoint edges. We refer to a point $b\in B$ as \emph{unmatched} in $M$ if it does not have an edge of $M$ incident on it and \emph{matched} otherwise. 
The following structural property of the $k$-SP problem will be critical in bounding the efficiency of our algorithm.

\begin{restatable}{lemma}{geometricAnalysis}\label{lemma:geometric-matching}
    For any cell $\cell$ of $\mcH$, there exists a matching $M'$ between $A_\cell\cup B_\cell$ that matches all except $O(n_\cell^{4/5})$ points of $B_\cell$ and has a cost $O(\ell_\cell n_\cell^{3/5})$.
\end{restatable}
\begin{proof}
We place a grid $G$ with cell-side-length $\ell_\cell n_\cell^{-2/5}$ inside $\cell$. For each square $\xi$ of $G$, let $\requests_\xi$ denote the subset of requests that lie inside $\xi$, and define $M_\xi$ to be the bipartite matching that corresponds to a single server serving the requests of $\requests_\xi$. Define $M':=\bigcup_{\xi\in G}M_\xi$. See Figure~\ref{fig:extended_matching}(left). 
\end{proof}

\begin{figure}
    \centering
    \includegraphics[width=0.4\linewidth]{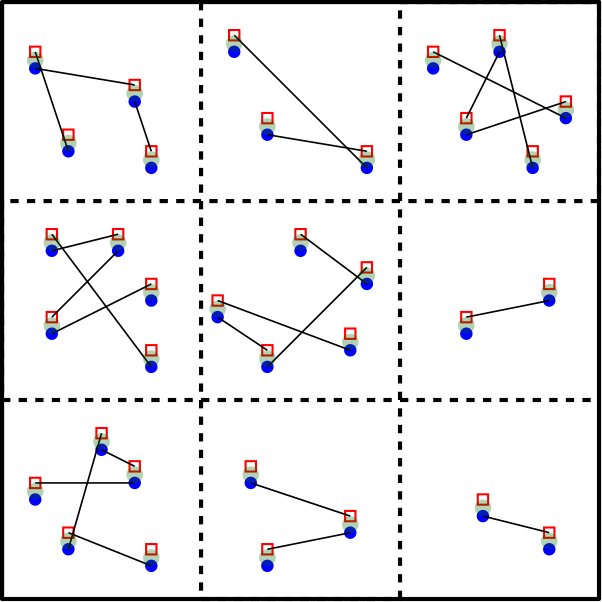}\hspace{2em}
    \includegraphics[width=0.4\linewidth]{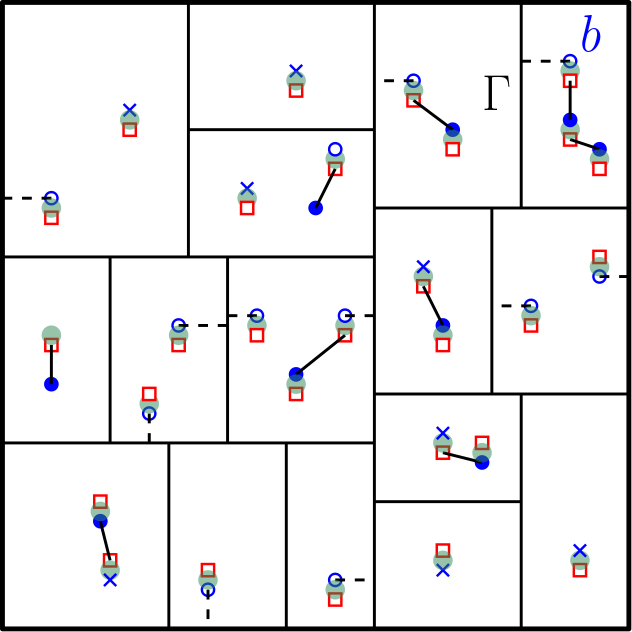}
    \caption{(left) The low-cost high-cardinality matching constructed inside a cell $\cell$ of $\mcH$, and (right) a $\mcC$-extended $20$-matching with $9$ matched points (blue discs), $11$ boundary-matched points (blue circles), and $8$ free points (blue crosses). The matching cost is the total length of the solid and dashed lines. The boundary-matched point $b$ is matched to the boundary $\Gamma$.}
    \label{fig:extended_matching}
\end{figure}

\subsection{Extended Bipartite Matching}\label{sec:primal}
Suppose we are given a subset $\mcC$ of cells from $\mcH$ that partitions $\cell^*$, i.e., the cells of $\mcC$ are interior disjoint and these cells cover the root square $\cell^*$. Let $\distance{u}{\cell}$ denote the shortest distance from $u$ to the boundaries of the cell $\cell$. For any point $u$ inside $\cell^*$, let $\cell_u$ be the cell of $\mcC$ that contains $u$. We define $\distance{b}{\mcC} = \distance{b}{\cell_b}$. We extend the definition of matching to allow points in $B$ to match to the boundaries of cells in $\mcC$. A \emph{$\mcC$-extended matching} consists of a matching $M$ as well as a subset $B^\mcC$ of points of $B$ that are unmatched in $M$ but instead matched to the boundaries of the cells of $\mcC$ that contain them. See Figure~\ref{fig:extended_matching}(right). The cost of a $\mcC$-extended matching $M^\mcC$ is
\begin{equation}
    w_\mcC(M^\mcC) := \sum_{(a,b)\in M} \distance{a}{b} + \sum_{b \in B^{\mcC}} \distance{b}{\mcC}.
\end{equation}

We refer to all points of $B^{\mcC}$ as \emph{boundary-matched}. All points of $B$ that are neither matched in $M$ nor boundary-matched are considered \emph{free}. All points of $A$ that are not matched in $M$ are also considered free. The size of $M^{\mcC}$ is equal to $|M|+|B^{\mcC}|$. When $\mcC$ is clear from the context, we refer to $M^{\mcC}$ as an extended matching. An extended matching $M^{\mcC}$ of size $t$ with the minimum cost is called a \emph{minimum-cost $\mcC$-extended $t$-matching}.

Consider the partitioning $\mcC^*=\{\cell^*\}$, where $\cell^*$ is the root cell of $\mcH$. Using the fact that the input points are far from the boundaries of $\cell^*$, we show in Lemma~\ref{lemma:root-inactive} below that any minimum-cost $\mcC^*$-extended matching $M^{\mcC^*}$ of size $t$ is also a minimum-cost $t$-matching, i.e., no point of $B$ is boundary-matched in $M^{\mcC^*}$. 
\begin{restatable}{lemma}{rootinactive}\label{lemma:root-inactive}
    Suppose $\mcC^*=\{\cell^*\}$, where $\cell^*$ is the root cell of $\mcH$. Let $M^{\mcC^*}=(M, B^{\mcC^*})$ be a minimum-cost extended $t$-matching on $\mcG_\sigma$, for $t<n$. Then, the matching $M$ is a minimum-cost $t$-matching.
\end{restatable}

Let $M^\mcC=(M, B^\mcC)$ denote a $\mcC$-extended matching.
Any path $P$ on the graph $\mcG_\sigma$ whose edges alternate between matching and non-matching edges in $M$ is called an \emph{alternating path}. 
An alternating path $P$ is called an \emph{augmenting path} if $P$ starts from a free point $b\in B$ and ends with either (i) a free point $a\in A$, or (ii) a point $b'\in B$.  We \emph{augment} the extended matching $M^{\mcC}$ along an augmenting path $P$ by updating its matching $M\leftarrow M\oplus P$ and for case (ii), we match $b'$ to the boundary and update $B^\mcC$ to include $b'$. The \emph{net-cost} of $P$ in case (i) is $\phi(P):=\sum_{(a,b)\in P\setminus M}\distance{a}{b} - \sum_{(a,b)\in P\cap M}\distance{a}{b}$, and 
in case (ii) is \[\phi(P):=\distance{b'}{\mcC} + \sum_{(a,b)\in P\setminus M}\distance{a}{b} - \sum_{(a,b)\in P\cap M}\distance{a}{b}.\]

From Lemma~\ref{lemma:root-inactive}, given a sequence $\sigma$ of $n$ requests in the unit square, one can compute an optimal solution to the $k$-SP problem on $\sigma$ by computing a minimum-cost $\mcC^*$-extended $(n-k)$-matching on $\mcG_\sigma$. Given a partitioning $\mcC$, to compute a minimum-cost $\mcC$-extended $(n-k)$-matching, similar to the Hungarian algorithm, one can start from an empty extended matching $M^\mcC$ and iteratively augment $M^\mcC$ along a minimum net-cost augmenting path. In the next section, we present a primal-dual framework that our algorithm uses to efficiently compute minimum net-cost augmenting paths.

\subsection{A Constrained Dual Formulation for Extended Matchings}\label{sec:dual}
Suppose $\mcC$ is a set of cells of $\mcH$ partitioning the root cell $\cell^*$. Consider a $\mcC$-extended matching $M^\mcC=(M, B^\mcC)$ on $\mcG_\sigma$ along with a set of non-negative dual weights $y:A\cup B\rightarrow \mbR_{\ge0}$. Let $A_F$ be the set of free points of $A$ with respect to $M^\mcC$. We say that $M^{\mcC}, y(\cdot)$ is \emph{feasible} if,
\begin{align}
    y(b) - y(a) &\le \distance{a}{b},&  \quad \forall(a,b)&\in E,\label{eq:dualfeasibility-non-matching-time}\\
    y(b) - y(a) &= \distance{a}{b},& \quad \forall(a,b)&\in M,\label{eq:dualfeasibility-matching-time}\\
    y(b) &\le \distance{b}{\mcC},&  \quad \forall b&\in B,\label{eq:dualfeasibility-b}\\
    y(b) &= \distance{b}{\mcC},&  \quad \forall b&\in B^\mcC,\label{eq:dualfeasibility-b-admissible}\\
    y(a) &= 0,& \quad \forall a&\in A_F.\label{eq:dualfeasibility-a-free}
\end{align}

For any  edge $(a,b)\in E$, the \emph{slack} of $(a,b)$ is defined as $s(a,b):=\distance{a}{b} - y(b) + y(a)$. The edge $(a,b)$ is \emph{admissible} if $s(a,b)=0$. For any point $b\in B$, the slack of $b$ is defined as $s(b):=\distance{b}{\mcC} - y(b)$. For any feasible extended matching $M^{\mcC}, y(\cdot)$, the slack of every edge as well as every point $b \in B$ is non-negative.
Recall that an augmenting path $P$ starts at a free point $b\in B$ and ends at (i) a free point $a\in A$ or (ii) a point $b'\in B$. The path $P$ is \emph{admissible} if all edges of $P$ are admissible and in case (ii), the slack of the end-point $b'$ is $s(b')=0$.
The following properties of extended feasible matchings are critical in the design of an efficient and correct algorithm.

\begin{restatable}{lemma}{netcost}\label{lemma:augpath}
Given a feasible $\mcC$-extended $t$-matching $M^{\mcC} = (M, B^{\mcC})$ and a set of non-negative dual weights $y(\cdot)$ on $A\cup B$, let $P$ be a minimum net-cost augmenting path with respect to $M^{\mcC}$. Let, for any $\cell \in \mcC$, $y_\cell=\max_{b' \in B_{\cell}} y(b')$.
Then, 
\begin{itemize}
\item[(a)] all points of $P$ lie inside a single cell of $\mcC$, and
\item[(b)] if, for every cell $\cell \in \mcC$, $y_\cell \le \phi(P)$ and for all free points $b\in B_\cell$, $y(b) = y_\cell$, then $M^\mcC$ is a minimum-cost extended $t$-matching.
\end{itemize}
\end{restatable}

\vspace{1em}
The property described in Lemma~\ref{lemma:augpath}(a) is important for the design of an efficient algorithm. Unlike the Hungarian algorithm, which searches the entire graph for the minimum net-cost augmenting path, our algorithm can find the minimum net-cost augmenting path by searching for the cheapest augmenting path inside each cell $\cell \in\mcC$ and then taking the smallest among them. Thus, we can replace a global search with a search inside each cell of $\mcC$.
The property in Lemma~\ref{lemma:augpath}(b) is important for the design of a correct algorithm since it provides conditions under which an extended matching is a minimum-cost extended matching. Our algorithm is designed to maintain these conditions as invariants during its execution.  
Lemma~\ref{lemma:fresh_duals} provides a method to update the dual weights, which will be essential during the process of merging cells.

\begin{restatable}{lemma}{updateDuals}\label{lemma:fresh_duals}
        Suppose $M^\mcC, y(\cdot)$ is a feasible $\mcC$-extended matching and $P$ is a minimum net-cost augmenting path. For any cell $\cell\in\mcC$, define $y_\cell=\max_{b' \in B_{\cell}} y(b')$, and suppose $y_\cell\le \phi(P)$ and $y(b_f)=y_\cell$ for all free points $b_f\in B_\cell$. Then, one can update the dual weights in $\tilde{O}(n_\cell\Phi(n_\cell))$ time such that $M^\mcC, y(\cdot)$ remains feasible, $y(b)\le \phi(P)$ for all $b\in B_\cell$, and $y(b_f)=\phi(P)$ for all free points $b_f\in B_\cell$.
\end{restatable}

\vspace{0.5em}
The next lemma provides critical properties that allow for correctly and efficiently merging cells in $\mcC$.

\begin{restatable}{lemma}{combination}\label{lemma:combination}
    For a cell $\cell$ of $\mcH$, suppose $\cell'$ and $\cell''$ denote its two children, and let $\mcC$ denote a partitioning containing $\cell'$ and $\cell''$. Let $\mcC'=\mcC\cup\{\cell\}\setminus\{\cell',\cell''\}$. Given a feasible $\mcC$-extended matching $(M, B^\mcC), y(\cdot)$, let $B_\cell^\mcC\subseteq B^\mcC$ denote the subset of boundary-matched points that are matched to the divider $\Gamma_\cell$ of $\cell$. Then,
    \begin{itemize}
        \item[(a)] the $\mcC'$-extended matching $(M, B^\mcC\setminus B_\cell^\mcC), y(\cdot)$ is also feasible, and,
        \item[(b)]  $|B_\cell^\mcC|=O(n^{4/5})$.
    \end{itemize}
\end{restatable}

\vspace{1em}

From Lemma~\ref{lemma:combination}(a), erasing a divider  does not cause the feasibility conditions to be violated. The proof of Lemma~\ref{lemma:combination}(a) relies on the fact that when we erase the divider of a cell in $\mcC$, the RHS of~\eqref{eq:dualfeasibility-b} will only increase and so it is not violated. Despite preserving the feasibility conditions, erasing the boundary may result in the violation of Lemma~\ref{lemma:augpath}(b), i.e., the matching may no longer be a minimum-cost extended matching. 
In our algorithm in Section~\ref{sec:k-seq}, we describe a process to adjust the matching $M^\mcC$ and dual weights $y(\cdot)$ and obtain a minimum-cost extended matching.

\myparagraph{Residual Graph.}
Similar to the Hungarian algorithm, we define a residual graph that assists in finding the minimum net-cost augmenting path. Consider a feasible $\mcC$-extended matching $M^{\mcC}=(M,B^{\mcC})$ along with a set of dual weights $y(\cdot)$ on points of $A\cup B$. For each cell $\cell \in \mcC$, we define a residual graph $\mcG_{\cell}$. The vertex set of $\mcG_{\cell}$ is a source vertex $s$ and the points in $A_{\cell}\cup B_{\cell}$. For any edge $(a,b)\in E$ inside $\cell$, if $(a,b)$ is an edge in $M$ (resp. not an edge in $M$), there is an edge directed from $a$ to $b$ (resp. from $b$ to $a$) with a weight $s(a,b)$ in $\mcG_\cell$. Furthermore, there is an edge directed from $s$ to every free point $b\in B$ with a weight $y(b)$.

\section{A Sub-Quadratic Algorithm for the \texorpdfstring{$k$}{}-SP and \texorpdfstring{$k$}{}-SPI Problems}\label{sec:k-seq}
In this section, we describe an algorithm that, given a sequence $\requests$ of $n$ requests in $2$-dimensions, computes the optimal solution to the $k$-SP and $k$-SPI problems in $\tilde{O}(n^{1.8}\Phi(n)\log(n\Delta))$ time. We begin by describing our algorithm for the $k$-SP problem and discuss how it can be extended to the $k$-SPI problem in Section~\ref{sec:kspi-extension-sub}.

\subsection{Algorithm for the \texorpdfstring{$k$}{}-SP Problem}
Initialize $\mcC$ to the leaf cells of $\mcH$. Let $M^{\mcC}=(M, B^{\mcC})$ be the extended matching maintained by the algorithm and initialized to $M=\emptyset$ and $B^{\mcC}=\emptyset$. For each point $v \in A\cup B$, let $y(v)$ denote its dual weight initialized to $y(v)=0$. Let $B_F$, initialized to $B$, be the free points of $B$ with respect to $M^{\mcC}$. For each cell $\cell$ that contains at least one free point $b \in B_F$, let $P_\cell$ denote the minimum net-cost augmenting path inside $\cell$. Initially, since $\cell$ is a leaf of $\mcH$, it contains only one point $b \in B_F$, and therefore, $P_{\cell}$ is this point with a net-cost equal to $\distance{b}{\mcC}$. Our algorithm maintains a priority queue $\pq$ storing every leaf cell $\cell\in\mcC$ with at least one free vertex with a key of $\phi(P_\cell)$. 
At any time during the execution of our algorithm, let $\cell_{\min}$ be the cell with the smallest key in $\pq$ and let $\ymax$ be the key of $\cell_{\min}$.
Execute the following steps until $\pq$ becomes empty:

\begin{itemize}
    \item While $|B_F| > k$, 
    \begin{itemize}
        \item {\it Extended Hungarian search step:} Extract the cell $\cell$ with the minimum key of $\ymax$ from $\pq$. Augment the matching $M^\mcC$ along $P_\cell$ and update the key of $\cell$ in $\pq$ (See Section~\ref{subsubsec:global} for details).
    \end{itemize}
    \item {\it Merge step:} If $\mcC=\{\cell^*\}$, remove $\cell^*$ from $\pq$ and return the matching $M$ of $M^{\mcC}$. Otherwise, pick a cell $\cell'\in\mcC$ with the smallest perimeter and let $\cell$ and $\cell''$ be the parent and sibling of $\cell'$ in $\mcH$, respectively. Erase the divider of $\cell$, i.e., set $\mcC=\mcC\setminus\{\cell',\cell''\}\cup \{\cell\}$. 
    We execute a Merge procedure that updates the matching $M^\mcC$ inside the cell $\cell$ so that $\phi(P_\cell)$ is at least $\ymax$ (See Section~\ref{subsubsec:local} for details). At this point, the size of the updated extended matching $M^{\mcC}$ may not be $(n-k)$, i.e., there may be more than $k$ free points.
\end{itemize}

\paragraph{Invariants.} For each cell $\cell \in \mcC$, let $y_\cell= \max_{b\in B_\cell}y(b)$. During the execution of our algorithm,
\begin{itemize}
    \item[(I1)] the extended matching $M^\mcC, y(\cdot)$ is feasible,
    \item[(I2)] For each cell $\cell \in \mcC$, $y_\cell \le \ymax$ and for all free point $b\in B_\cell$, $y(b)=y_{\cell}$, and,
   \item[(I3)] The $\ymax$-value is non-decreasing. Furthermore, after each step of the algorithm, $\ymax$ denotes the smallest net-cost of all augmenting paths with respect to $M^{\mcC}$.
\end{itemize}

\subsubsection{Details of the Extended Hungarian Search Step}\label{subsubsec:global} 
Given a feasible extended matching $M^\mcC, y(\cdot)$  and a cell $\cell\in\mcC$, the extended Hungarian search procedure computes the minimum net-cost augmenting path $P_\cell$ and augments $M^\mcC$ along $P_\cell$. It then computes the new minimum net-cost augmenting path and updates the key of $\cell$ in $\pq$ to be its net-cost. 
This procedure is similar to the classical Hungarian search procedure executed on $\mcG_\cell$ and is mildly modified to include the augmenting paths that end at the boundary of $\cell$. 
Details of the procedure are as follows. 
\begin{enumerate}
\item {\it Update duals:} With $s$ as the source vertex, execute Dijkstra's shortest path algorithm on the residual graph $\mcG_\cell$. Let $P_v$ be the shortest path from $s$ to each $v\in A_\cell\cup B_\cell$ and let $\kappa_v$ be the cost of $P_v$. Let
\begin{equation}\label{eq:global-kappa}
    \kappa = \min\{ \min_{a\in \freeofcell{A}{\cell}}\kappa_a, \min_{b\in B_\cell}\kappa_b + s(b)\}.
\end{equation}
For any $v \in A_\cell\cup B_\cell$, if $\kappa_v < \kappa$, set $ y(v) \leftarrow y(v) +\kappa - \kappa_v$.
\item {\it Augment:} Let $u\in \freeofcell{A}{\cell}\cup B_\cell$ be the point realizing the minimum distance in Equation~\eqref{eq:global-kappa}. Let $P$ be the augmenting path obtained by removing $s$. Augment $M^\mcC$ along $P$. 
\item {\it Update key:} If $B_\cell$ has no free points, then remove $\cell$ from $\pq$. Otherwise,
\begin{enumerate}
    \item Recompute the residual graph $\mcG_{\cell}$ with respect to the updated matching,
    \item With $s$ as the source, execute the Dijkstra's shortest path algorithm on $\mcG_\cell$. For each $v \in A_\cell\cup B_\cell$, let $\kappa_v$ be the distance from $s$ to $v$.
    \item  Update the key of $\cell$ in $\pq$ to $\kappa_\cell = \min\{ \min_{a\in \freeofcell{A}{\cell}}\kappa_a, \min_{b\in B_\cell}\kappa_b + s(b)\}$.
\end{enumerate}   
\end{enumerate}

Lemma~\ref{lemma:hung_properties} establishes the properties of the extended Hungarian search procedure.
\begin{restatable}{lemma}{hungarianProp}\label{lemma:hung_properties}\label{lemma:Hungarian_path}
After the execution of the extended Hungarian search procedure for a cell $\cell$, the extended matching $M^\mcC, y(\cdot)$ remains feasible, $y(v)\le \ymax$ for all points $v \in A_{\cell}\cup B_{\cell}$, and $y(b_f)=\ymax$ for all free points $b_f\in B_\cell$. Furthermore, the path $P$ computed by the procedure is a minimum net-cost augmenting path inside $\cell$. After augmenting along $P$, the updated key for $\cell$ is the smallest net-cost of all augmenting paths inside $\cell$ and is at least $\ymax$.
\end{restatable}

\subsubsection{Details of the Merge Step}\label{subsubsec:local} 
Given a feasible extended matching $M^\mcC=(M, B^\mcC), y(\cdot)$, any cell $\cell$ of $\mcH$ where both of its children $\cell'$ and $\cell''$ are in $\mcC$, and a value $\ymax$, the merge procedure uses the algorithm in Lemma~\ref{lemma:fresh_duals} to update the dual weights $y(\cdot)$ inside $\cell'$ (resp. $\cell''$)  so that $M^\mcC, y(\cdot)$ remains feasible, the dual weights of all points in $\cell'$ (resp. $\cell''$) are at most $\ymax$, and the dual weight of all free points $b'_f\in B_{\cell'}$ (resp. $b''_f\in B_{\cell''}$) is $y(b'_f)=\ymax$ (resp. $y(b''_f)=\ymax$). The procedure then updates $\mcC\leftarrow \mcC\cup \{\cell\}\setminus\{\cell', \cell''\}$ and makes the points matched to the divider $\Gamma_\cell$ free.  
While there exists a free point $b\in B_\cell$ with $y(b)<\ymax$,
\begin{enumerate}
    \item With $s$ as the source, execute Dijkstra's shortest path algorithm on the residual graph $\mcG_\cell$. For each $v \in A_\cell\cup B_\cell$, let $P_v$ be the shortest path from $s$ to $v$ in $\mcG_\cell$ and let $\kappa_v$ be its cost. Define
\begin{equation}\label{eq:local-kappa}
    \kappa = \min\{ \min_{a\in \freeofcell{A}{\cell}}\kappa_a, \min_{b\in B_\cell}\kappa_b + s(b), \min_{b\in B_\cell}\kappa_b + \ymax - y(b)\}.
\end{equation}

\item Let $u\in\freeofcell{A}{\cell}\cup B_\cell$ be the point realizing the minimum value in Equation~\eqref{eq:local-kappa}. Let $P$ be the path obtained by removing $s$ from the path $P_u$.
\begin{enumerate}
    \item If $u\in B_\cell$ and $\kappa=\kappa_u+\ymax-y(u)$ (i.e., $\kappa$ is determined by the third element in the RHS of Equation~\eqref{eq:local-kappa}), then $P$ is a path from a free point $b\in B_\cell$ to the matched point $u$. For each $v \in A_\cell\cup B_\cell$, if $\kappa_v < \kappa$, update its dual weight to $y(v)\leftarrow y(v)+\kappa - \kappa_v$. The path $P$ is an admissible alternating path with respect to the updated dual weights. Set $M\leftarrow M\oplus P$. Note that $u$ is now a free point with $y(u)=\ymax$.
    \item Otherwise, $P$ is an augmenting path. For each $v \in A_\cell\cup B_\cell$, if $\kappa_v < \kappa$, update its dual weight to $y(v)\leftarrow y(v)+\kappa - \kappa_v$. The path $P$ is an admissible augmenting path with respect to the updated dual weights. Augment $M$ along $P$.
\end{enumerate}
\end{enumerate}
At the end of this execution, all free points of $B$ in $\cell$ have a dual weight of $\ymax$. Finally, we update the key for $\cell$ by executing steps 3(a)--(c) from the extended Hungarian search step.

The following lemma states the useful properties of the merge step.

\begin{restatable}{lemma}{mergeProp}\label{lemma:merge_properties}
After the execution of the merge procedure on a cell $\cell$,
the updated extended matching $M^{\mcC}, y(\cdot)$ is feasible, the dual weights $y(v)$ for every $v \in A_\cell\cup B_\cell$ is at most $\ymax$, and $y(b_f)=\ymax$ for all free points $b_f\in B_\cell$. Furthermore, the updated key for $\cell$ is the smallest net-cost of all augmenting paths within $\cell$ and is at least $\ymax$.
\end{restatable}

\subsection{Analysis}\label{subsec:geometric-analysis}

We begin by showing in Section~\ref{sec:correctness} that the three invariants (I1)--(I3) hold during the execution of our algorithm and use them to show the correctness of our algorithm. We then show in Section~\ref{subsec:efficiency} that the running time of our algorithm is $\tilde{O}(n^{9/5}\Phi(n)\log n\Delta)$.

\subsubsection{Correctness}\label{sec:correctness}

Our algorithm initializes $M^{\mcC}$ with a feasible $\mcC$-extended matching and sets all dual weights and $\ymax$ to $0$. Therefore, (I1) and (I2) hold at the start of the algorithm. The extended Hungarian search (Lemma~\ref{lemma:hung_properties}) as well as the merge step (Lemma~\ref{lemma:merge_properties}) do not violate the feasibility of the extended matching and therefore (I1) holds during the execution of the algorithm. The extended Hungarian search (Lemma~\ref{lemma:hung_properties}) and the merge (Lemma~\ref{lemma:merge_properties}) procedures keep the dual weight of every point inside $\cell$ at or below $\ymax$, while ensuring that the dual weight of free points inside $\cell$ is $\ymax$; hence, (I2) holds. Finally, both the extended Hungarian search (Lemma~\ref{lemma:hung_properties}) and the merge (Lemma~\ref{lemma:merge_properties}) procedures update the key of $\cell$ to the smallest net-cost of all augmenting paths inside $\cell$ and do not decrease the key of any cell. From this observation, (I3) follows in a straightforward way.

From (I1) and (I2), our algorithm maintains a feasible $\mcC$-extended matching where, for each cell $\cell\in\mcC$, $y_\cell \le \ymax$. From (I3), $\ymax$ is equal to the minimum net-cost of all augmenting paths with respect to $M^{\mcC}$; combining this with Lemma~\ref{lemma:augpath}(b), we conclude that $M^{\mcC}$ is a minimum-cost $\mcC$-extended matching. Upon termination, the algorithm returns a minimum-cost $\mcC^*$-extended $(n-k)$-matching for $\mcC^*=\{\cell^*\}$, which from Lemma~\ref{lemma:root-inactive} has no boundary-matched points. Therefore, this matching is also a minimum-cost $(n-k)$-matching, as desired.

\subsubsection{Efficiency}\label{subsec:efficiency}

Both the merge step and the extended Hungarian search step require an execution of Dijkstra's shortest path algorithm on $\mcG_\cell$. In Section~\ref{sec:hungarian_nk}, we show that this execution takes $\tilde{O}(n_\cell\Phi(n_\cell))$ time. Using this, we analyze the execution time.

We begin by establishing a bound on the execution time of the merge step, which combines $\cell'$ and $\cell''$ into a single cell $\cell$. At the start of this step, the dual weight of all free points in $\cell'$ and $\cell''$ are raised to $\ymax$ (from Lemma~\ref{lemma:fresh_duals}). As a result, once the divider is removed, the only free points with dual weights below $\ymax$ are those that are matched to the divider $\Gamma_\cell$.  From Lemma~\ref{lemma:combination}(b), the number of points matched to $\Gamma_\cell$ is $O(n^{4/5})$. Therefore, the while-loop in the merge procedure executes only $O(n^{4/5})$ times. Since each iteration takes $\tilde{O}(n_\cell\Phi(n_\cell))$ time, the total execution time of a single execution of the merge step is $\tilde{O}(n^{4/5}n_\cell \Phi(n_\cell))$. 
Given that each point is in only $O(\log n\Delta)$ many cells of $\mcH$, the total time taken by the merge step across all cells of $\mcH$ is $\tilde{O}(n^{9/5}\Phi(n)\log n\Delta)$.

Similarly, if the execution time for the extended Hungarian search within a single cell $\cell$ is bounded by $O(n^{4/5}n_\cell \Phi(n_\cell))$, then the cumulative execution time of the extended Hungarian search across all cells -- and consequently, the total runtime of the entire algorithm -- can be bounded by $\tilde{O}(n^{9/5}\Phi(n)\log n\Delta)$.
In the remainder of our analysis, we establish a bound of $O(n^{4/5}n_\cell \Phi(n_\cell))$ for the time taken by the extended Hungarian search within a single cell $\cell$.
Recall that the algorithm selects a cell $\cell$ containing the minimum net-cost augmenting path from the priority queue $\pq$ and performs an extended Hungarian search procedure within $\cell$, requiring $\tilde{O}(n_\cell\Phi(n_\cell))$ time. To analyze the total execution time of the extended Hungarian search procedure, we show that any cell $\cell$ can be selected by the algorithm at most $O(n^{4/5})$ times.

We categorize the selection of $\cell$ as a \emph{low-net-cost} case if $\ymax\le \ell_\cell n^{-1/5}$ and a \emph{high-net-cost} case otherwise. First, we show that in the high-net-cost case ($\ymax>\ell_\cell n^{-1/5}$), the number of free points remaining in $\cell$ is $O(n^{4/5})$, and as a result, the total number of high-net-cost selections of $\cell$ is $O(n^{4/5})$.

\begin{restatable}{lemma}{highMax}\label{lemma:high-max}
    Given a feasible extended matching $M^\mcC, y(\cdot)$ and any cell $\cell\in\mcC$, if the net-cost of the minimum net-cost augmenting path inside $\cell$ is greater than $\ell_\cell n^{-1/5}$, then the number of free points of $M^\mcC$ is $O(n^{4/5})$.
\end{restatable}
\begin{proof}
    Let $M$ be the matching corresponding to the extended matching $M^\mcC$ and let $M'$ be the matching defined in Lemma~\ref{lemma:geometric-matching}. Every free point of $M^\mcC$ participates in an augmenting or an alternating path in $M\oplus M'$. The number of alternating paths in the symmetric difference cannot exceed the number of free points of $M'$, which is $O(n^{4/5})$. Furthermore, the combined net-cost of these augmenting paths is at most $w(M')=O(\ell_\cell n^{3/5})$, and each one has a net-cost at least $\ell_\cell n^{-1/5}$. Thus, there are $O(n^{4/5})$ augmenting paths in the symmetric difference. See Appendix~\ref{appendix:hung_time} for a complete proof.
\end{proof}

Next, we bound the number of low-net-cost selections of $\cell$. Define $\mcC_\cell$ as the set of all cells $\cell'\in\mcH$ that are processed by the merge procedure while $\cell\in\mcC$ and $\ymax\le \ell_\cell n^{-1/5}$ (Figure~\ref{fig:current}). Our algorithm always picks the smallest perimeter cells to merge, and as a result, we obtain the following lemma.
\begin{restatable}{lemma}{ratio}
\label{lemma:ratio}
For any non-leaf cell $\cell$ in $\mcH$ and any cell $\cell'\in\current_\cell$, $\frac{1}{3}\ell_{\cell'}\le \ell_\cell \le \frac{9}{4}\ell_{\cell'}$.
\end{restatable}

For any cell $\cell'\in\mcC_\cell$, the merge step erases the divider $\Gamma_{\cell'}$ and makes the points that are matched to $\Gamma_{\cell'}$ free. Each of these new free points might cause $\cell$ to be selected by the algorithm. Therefore, to bound the number of low-net-cost selections of $\cell$, we show that the total number of points that are matched to the dividers of the cells in $\mcC_\cell$ is at most $O(n^{4/5})$. 
For any cell $\cell'\in\mcC_\cell$, let $\mcB_{\cell'}$ denote the set of points of $B^\mcC$ that are matched to the divider $\Gamma_{\cell'}$ when our algorithm starts the merge step on $\cell'$.
Thus, we have to bound $\sum_{\cell'\in\mcC_\cell}|\mcB_{\cell'}|$ by $O(n^{4/5})$. 
By invariant (I2), for each point $b\in \mcB_{\cell'}$, $y(b)\le \ymax\le \ell_\cell n^{-1/5}$. Furthermore, by Condition~\eqref{eq:dualfeasibility-b-admissible}, $y(b)=\distance{b}{\cell'}=\distance{b}{\Gamma_{\cell'}}$. Therefore, the point $b$ is within a distance $\ell_\cell n^{-1/5}$ from the divider $\Gamma_{\cell'}$. 
\begin{figure}
    \centering
    \includegraphics[width=0.45\linewidth]{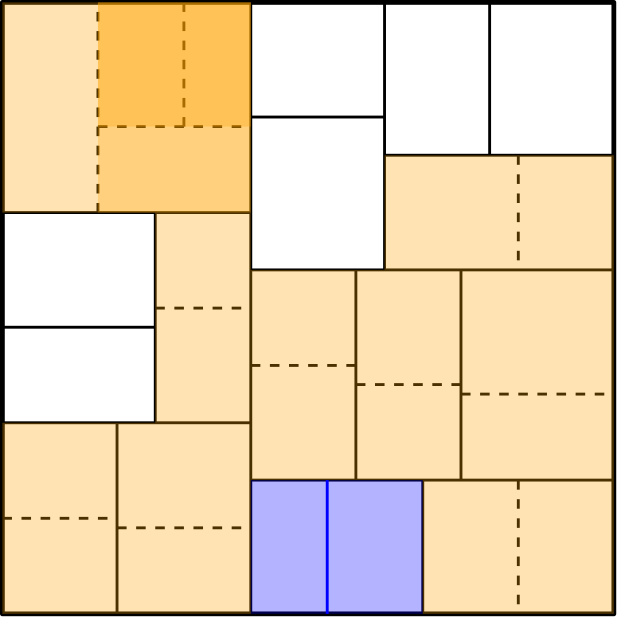}\hspace{2em}
    \includegraphics[width=0.45\linewidth]{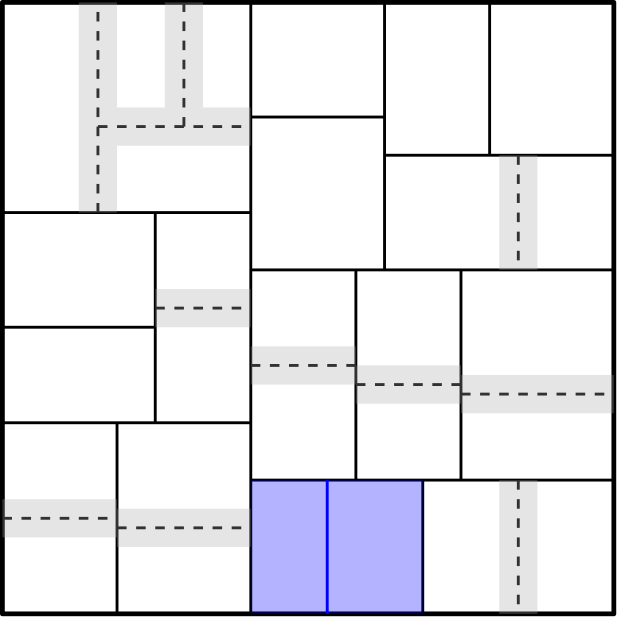}
    \caption{(left) The set $\mcC_\cell$ (the cells shaded in orange with dashed lines as their divider) for a cell $\cell$ (shaded in blue), and (right) any boundary-matched point in the gray area might cause a low-net-cost execution of the extended Hungarian search procedure on $\cell$.}
    \label{fig:current}
\end{figure}
From Lemma~\ref{lemma:ratio}, for each cell $\cell'\in\mcC_\cell$ and each point $b\in \mcB_{\cell'}$,  the point $b$ is within a distance $\ell_\cell n^{-1/5} < \ell_{\cell'}\lambda$ from the divider $\Gamma_{\cell'}$, and from the construction of $\mcH$ (Lemma~\ref{lemma:margin}),
$|\mcB_{\cell'}|=O(n_{\cell'}n^{-1/5})$. Furthermore, from Lemma~\ref{lemma:ratio} and the construction of $\mcH$, each point $b\in B$ can lie inside a constant number of cells in $\mcC_\cell$; hence, $\sum_{\cell'\in\mcC_\cell}|\mcB_{\cell'}| = O(\sum_{\cell'\in\mcC_\cell}n_{\cell'} n^{-1/5}) = O(n^{4/5})$. Therefore, the total number of low-net-cost selections of $\cell$ by the extended Hungarian search procedure is $O(n^{4/5})$, as claimed.

\subsection{Fast Search Procedures for the Geometric \texorpdfstring{$k$}{}-SP and the \texorpdfstring{$k$}{}-SPI Problems}\label{sec:hungarian_nk}

Given an instance of the geometric $k$-SP, in this section, we show that the search for the minimum net-cost augmenting path using a Hungarian search style procedure (as in the merge and extended Hungarian search procedures of our algorithm) can be done in $O(n\Phi(n)\log^3 n)$ time, where $\Phi(n)$ is the query/update time of an existing dynamic weighted nearest neighbor data structure (DWNN). 

Given a matching $M$ and dual weights $y(\cdot)$ satisfying Conditions~\eqref{eq:dualfeasibility-non-matching-time} and~\eqref{eq:dualfeasibility-matching-time}, the Hungarian search (which executes a Dijkstra's shortest path procedure) can be efficiently implemented on a complete bipartite graph using only $\tilde{O}(n)$ queries to a dynamic weighted bichromatic closest pair data structure (BCP). This search procedure grows a tree $\mathcal{T}$ by finding the cheapest cut edge under a weighted distance. See~\cite{agarwal1995vertical, ra_soda12,vaidya1989geometry} where the Hungarian search procedure that uses $O(n)$ queries to a BCP has been described. Furthermore, one can implement a BCP using a DWNN in $\tilde{O}(\Phi(n))$ time; here $\Phi(n)$ is the update/query time of the DWNN~\cite{eppstein1995dynamic}. 

In our case, the graph $\mathcal{G}_{\requests}$ is not a complete graph. However, we can easily decompose the graph into a set of $O(n)$ complete bipartite graphs as follows. Build a balanced binary search tree $T$ where the indices of requests $\{1,\ldots, n\}$ form the leaves. Let, for any node $v$, $L(v)$ denote the indices at the leaves of the left subtree of $v$ and $R(v)$ denote the indices at the leaves of the right subtree of $v$. We partition the edges as follows: At each node $v$ in $T$, we store all edges that go from the exit gates of requests with indices in $L(v)$ to the entry gates of the requests with indices in $R(V)$. See Figure~\ref{fig:cliquedec} for an example. Thus, at each node, we simply store the complete bipartite graph between points in $L(v)$ and $R(v)$. It is easy to see that any request $r_j$ participates in $O(\log n)$ cliques. 

\begin{figure}
    \centering
    \includegraphics[width=\textwidth]{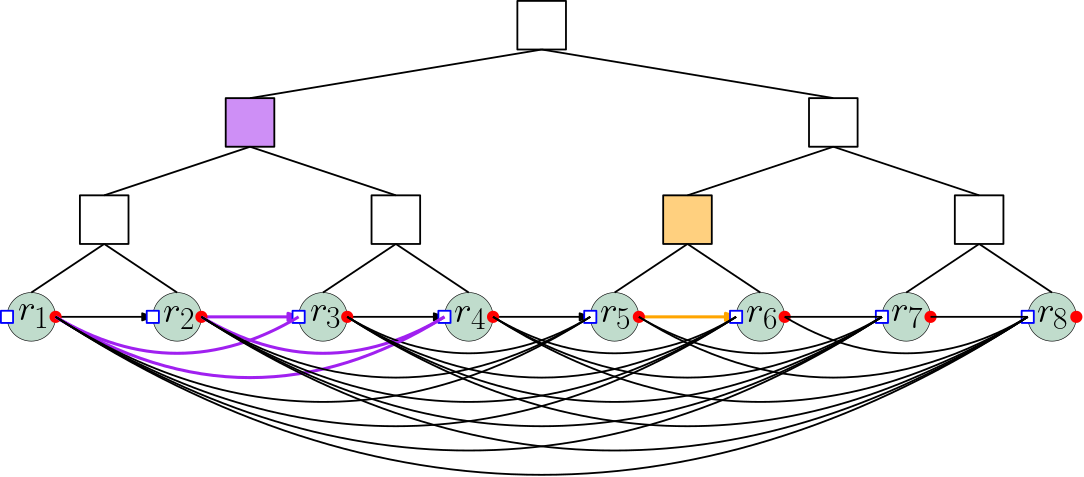}
    \caption{The balanced binary search tree $T$ used in the construction of our dynamic weighted bichromatic closest pair data structure. The purple (resp. orange) edges in the graph are stored at the purple (resp. orange) node of $T$.}
    \label{fig:cliquedec}
\end{figure}

While executing the Hungarian search, we maintain a BCP at each node $v$ of this tree. Let $\mathcal{T}$ denote the Dijkstra's shortest path tree maintained by the Hungarian search. Our data structure stores the nodes of $L(v)$ that are currently in $\mathcal{T}$ and the nodes of $R(v)$ that are currently not in $\mathcal{T}$. The weighted bichromatic closest pair for each of these cliques is stored in a global priority queue. The query to this data structure will simply return the pair at the top of this priority queue. Adding a node $v$ to the Dijkstra's shortest path tree will trigger insertion or deletion of $v$ in the $O(\log n)$ cliques that it participates in. For each of these cliques, we update their representative pair in the global priority queue. Thus, all insertions, deletions, and queries can be done in $\tilde{O}(\Phi(n))$ time.
We obtain a fast Hungarian search by using this data structure in place of the standard dynamic weighted bichromatic closest pair data structure as used in~\cite{ra_soda12}. 

\subsection{Extension of our Sub-Quadratic to the \texorpdfstring{$k$}{-}-SPI Problem}\label{sec:kspi-extension-sub}
Our algorithm for the $k$-SP problem computes a minimum-cost $(n-k)$-matching on a graph $\mcG_\sigma$ between a point set $A$ of size $n$ and a point set $B$ of size $n$. For the $k$-SPI problem on an instance $\sigma'=\eta\sigma$, since there exists a matching on the graph $\mcG_{\requests'}$ that matches all points of $B$, one can use the same algorithm by setting $k=0$ to compute a minimum-cost $n$-matching on $\mcG_{\requests'}$. Our algorithm therefore easily extends to the $k$-SPI problem. Note that the efficiency of our algorithm only requires the existence of a matching that matches the majority of the points of $B_\cell$ in each cell $\cell$ and has a low cost (Lemma~\ref{lemma:geometric-matching}). We observe that the same property holds for the $k$-SPI problem, as the initial locations of the servers only introduce new points to the set $A_\cell$ and the number of unmatched points of $B_\cell$ in the matching constructed in Lemma~\ref{lemma:geometric-matching} remains unchanged. Therefore, our algorithm runs in $\tilde{O}(n^{9/5}\Phi(n)\log\Delta)$ time for the $k$-SPI problem as well.

\section{Fast Implementations of the Hungarian Algorithm for the \texorpdfstring{$k$}{-}-SP and the \texorpdfstring{$k$}{-}-SPI Problems}\label{sec:nk-algod}
In this section, given a sequence of $n$ requests in $2$ dimensions, we present an adaptation of the Hungarian algorithm~\cite{kuhn1955hungarian} for computing a minimum $(n-k)$-matching under any $\ell_p$ norm in $\tilde{O}(nk\Phi(n))$ time.  We begin by describing the notions used extensively by existing bipartite matching algorithms and then present our algorithm for the $k$-SP problem.

\subsection{Background}
For a matching $M$ on $\mcG_\requests$, let $B_F$ (resp. $A_F$) denote the set of free points of $B$ (resp. $A$) with respect to $M$. The matching $M$ along with a set of non-negative dual weights $y:A\cup B\rightarrow \mbR_{\ge0}$ is \emph{feasible} if 
\begin{align}
    y(b) - y(a) &\le \distance{a}{b},&  \quad \forall(a,b)&\in E,\label{eq:dualfeasibility-non-matching}\\
    y(b) - y(a) &= \distance{a}{b},& \quad \forall(a,b)&\in M,\label{eq:dualfeasibility-matching}\\
    y(a) &= 0,& \quad a&\in A_F.\label{eq:dualfeasibility-free_a-nk}
\end{align}
Any feasible $t$-matching $M, y(\cdot)$, where $y(b)=\max_{b'\in B} y(b')$ for all free points $b\in B_F$, is a \emph{dual-optimal} matching. 
The following lemma relates the dual-optimal $t$-matchings and minimum-cost $t$-matchings.

\begin{restatable}{lemma}{partialoptimaltime}\label{lemma:partial-optimal-time}
    For any $t>0$ and any dual-optimal $t$-matching $M, y(\cdot)$ on a graph $\mcG_\requests$, the matching $M$ is a minimum-cost $t$-matching on $\mcG_\requests$.
\end{restatable}

For any feasible matching $M, y(\cdot)$ and any edge $(a,b)\in E$, the \emph{slack} of $(a,b)$, denoted by $s(a,b)$, is defined as $s(a,b):=\distance{a}{b} - y(b) + y(a)$. The edge $(a,b)$ is \emph{admissible} if $s(a,b)=0$. Any alternating path $P$ is admissible if all edges in $P$ are admissible.

\myparagraph{Reversed residual graph.} We construct a reversed residual graph $\bar{G}$ of $\mcG_\requests$ and matching $M$ as follows. The vertex set of $\bar{G}$ is the point set $A\cup B$ in addition to a source vertex $s$. For any pair $(a,b)\in E$, if $(a,b)$ is a matching (resp. non-matching) edge, there is an edge directed from $b$ to $a$ (resp. from $a$ to $b$) with a weight $s(a,b)$ in $\bar{G}$. Furthermore, there is an edge from $s$ to every point $b\in B$ with a weight $y_{\max} - y(b)$, where $y_{\max}:=\max_{b'\in B} y(b')$. 

\myparagraph{Reverse augmenting paths.}Define a \emph{reverse augmenting path} as an alternating path $P$ with the first and last edges as matching edges, i.e., a directed path from a point $b\in B$ to a point $a\in A$ in the reversed residual graph. Given a reverse augmenting path $P$, we \emph{reduce} the matching $M$ along $P$ by setting $M\leftarrow M\oplus P$. Reducing $M$ along $P$ decreases the number of matching edges of $M$ by one.

\subsection{Efficient Hungarian Algorithm for the \texorpdfstring{$k$}{}-SP Problem}\label{sec:hungarian_n2k_ksp}

At a high level, our algorithm starts with a matching $M$ on $\mcG_\requests$ corresponding to the optimal solution to the $1$-SP problem, i.e., the matching $M$ is the maximum matching on $\mcG_\requests$ representing the $1$-partitioning $\langle r_1, r_2,\ldots, r_n\rangle$. Our algorithm then iteratively increments the number of partitions (i.e., increases the number of free points of $B$ by one) while ensuring that the maintained partitioning has a minimum cost. In particular, using a reversed version of the Hungarian search procedure, our algorithm finds and updates the matching along the minimum (negative) net-cost reverse augmenting path (i.e., an alternating path with its first and last edges as matching edges). Our algorithm stops and returns the matching when $|M|=n-k$. We provide the details next.

\myparagraph{Algorithm.} Set $M= \{(a_i, b_{i+1})\mid i\in[1, n-1]\}$ as a matching that corresponds to serving all requests in $\langle r_1, r_2,\ldots, r_n\rangle$ by a single server. Assign dual weights to the points in $A\cup B$ as follows. Set $y(a_n)\leftarrow 0$. Starting from $i=n-1$, set $y(a_{i})\leftarrow\max\{0, \max_{j>i+1}y(b_j)-\distance{a_{i}}{b_j}\}$ and $y(b_{i+1})\leftarrow\distance{a_{i}}{b_{i+1}} + y(a_{i})$. Finally, set $y(b_1)\leftarrow \max_{i>1}y(b_i)$. The matching $M, y(\cdot)$ is a dual-optimal $(n-1)$-matching on $\mcG_\requests$.

While $|M|>n-k$ (i.e., there are less than $k$ free points in $B$), execute the \reverse\ procedure described below, which updates the dual weights and returns an admissible reverse augmenting path $P$. Reduce $M$ along $P$. 

{\it \reverse\ Procedure.} Execute the Dijkstra's shortest path procedure from the source vertex $s$ on the reversed residual graph $\bar{G}$, which computes the distance $\kappa_u$ of each vertex $u\in A\cup B$ from $s$. Define $\kappa:=\min_{a\in A} \kappa_a + y(a)$. For any point $u\in A\cup B$ with $\kappa_u<\kappa$, set $y(u)\leftarrow y(u)-\kappa+\kappa_u$. Let $a\in A$ denote the point with the minimum distance (i.e., $\kappa_a=\kappa$), and let $P$ denote the shortest path from $s$ to $a$. Return the path $P'$ obtained by removing $s$ from $P$.

\begin{restatable}{lemma}{nkAlgo}\label{lemma:nkAlgo}
    For any $1\le t\le k$, the matching $M, y(\cdot)$ maintained by our algorithm after $(t-1)$ iterations is a dual-optimal $(n-t)$-matching on $\mcG_\requests$.
\end{restatable}

The initialization step of our algorithm takes $O(n^2)$ time to compute the set of dual weights and the matching $M$. Additionally, our algorithm runs $k$ iterations, where in each iteration, it executes the \reverse\ procedure in $O(n^2)$ time. Therefore, the total running time of our algorithm would be $O(n^2k)$ time. In geometric settings, using the BCP data structure described in Section~\ref{sec:hungarian_nk}, the initialization step, as well as the \reverse\ procedure, can be executed in $\tilde{O}(nk\Phi(n))$ time.

\subsection{Extension of the Fast Hungarian Algorithm to the \texorpdfstring{$k$}{}-SPI Problem} 
For the $k$-SPI problem, given $\mcG_{\requests'}$ with $|A|=n+k$ and $|B|=n$, the goal is to compute a minimum-cost $n$-matching of $A$ and $B$ that matches all points in $B$. We extend our algorithm for the $k$-SP problem to the $k$-SPI problem in a straightforward way as follows: Our algorithm initializes a matching $M$ corresponding to the $1$-SP optimal solution for the requests $\requests$. It then iteratively introduces a new server's initial location and finds a minimum net-cost alternating path from the free point of $A$ corresponding to the new server to a matched point of $A$. The details are provided next.

Set $M\leftarrow \{(a^1, b_1)\}\cup \{(a_i, b_{i+1})\mid i\in[1, n-1]\}$ as a matching that corresponds to the $1$-partitioning $\langle s_1, r_1, r_2,\ldots, r_n\rangle$ with the $s_1$ as the initial location. Assign dual weights to the points in $A\cup B$ as follows. Set $y(a_n)\leftarrow 0$. Starting from $i=n-1$, set $y(a_{i})\leftarrow\max\{0, \max_{j>i+1}y(b_j)-\distance{a_{i}}{b_j}\}$ and $y(b_{i+1})\leftarrow\distance{a_{i}}{b_{i+1}} + y(a_{i})$. Finally, set $y(a^{1})\leftarrow\max\{0, \max_{i}y(b_i)-\distance{a^{1}}{b_i}\}$ and $y(b_1)\leftarrow \distance{a^{1}}{b_{1}} + y(a^{1})$. Define $\mcG=\mcG_{\{s_1\}\requests}$ as the graph representation of the $1$-SPI problem with $s_1$ as the only server. The matching $M, y(\cdot)$ is a dual-optimal $(n-1)$-matching on $\mcG$.

In each iteration $1\le t\le k-1$, our algorithm adds a point $a^{t+1}$ to the vertex set $A$ of the graph $\mcG$ and connects it to all points $b_i\in B$ with a cost $\distance{a^{t+1}}{b_i} = \distance{s_{t+1}}{r_i}$. Assign $y(a^{t+1})\leftarrow\max\{0, \max_{i}y(b_i)-\distance{a^{t+1}}{b_i}\}$. Construct a reversed residual graph, similar to the one constructed for the $k$-SP problem, with the difference that the source vertex $s$ is only connected to the point $a^{t+1}$ with a zero cost. Execute the \reverse\ procedure on the reversed residual graph to update the dual weights and compute an admissible alternating path $P$ from $a^{t+1}$ to a point $a\in A$ with a zero dual weight. Set $M\leftarrow M\oplus P$.

\section{Bipartite Matching on Randomly Colored Points}\label{sec:randomly-colored}
Given a set $U$ of $2n$ points inside the unit square, suppose $A$ is a subset of $n$ points of $U$ chosen uniformly at random and let $B=U\setminus A$. Define $G$ to be a complete bipartite graph between $A$ and $B$, i.e., $G=(A\cup B, E=A\times B)$ is a bipartite graph on the vertex set $A\cup B$ with an edge set $E=A\times B$. In this section, we show that our algorithm from Section~\ref{sec:k-seq} can be used to compute the minimum-cost perfect matching between $A$ and $B$ under $\ell_2^q$ distances, for any $q\ge 1$, in $\tilde{O}(n^{7/4}\Phi(n)\log\Delta)$ time in expectation. For simplicity in presentation, we first present our analysis for $q=2$. In Section~\ref{sec:appendix-GRS} in the appendix, we extend our analysis to any dimension $d\ge 2$ and any $q\ge 1$.

To compute a minimum-cost perfect matching between $A$ and $B$, we construct a hierarchical partitioning $\mcH$ as described in Section~\ref{sec:hierarchical} with a parameter $\lambda = 9n^{-1/4}$. We then execute our algorithm from Section~\ref{sec:k-seq}, where the parameter $k$ is set to $0$. In this way, our algorithm computes a minimum-cost $\mcC^*$-extended $n$-matching $M^{\mcC^*}$ for the partitioning $\mcC^*=\{\cell^*\}$ for the root cell $\cell^*$ of $\mcH$. In the following lemma, we show that the matching $M$ of the extended matching $M^{\mcC^*}$ computed by our algorithm is a minimum-cost perfect matching, as desired.

\begin{lemma}
    Given the partitioning $\mcC^*=\{\cell^*\}$ and any minimum-cost $\mcC^*$-extended $n$-matching $M^\mcC=(M, B^{\mcC^*})$, the matching $M$ is a minimum-cost $n$-matching.
\end{lemma}
\begin{proof}
    For the sake of contradiction, suppose the matching $M$ is not an $n$-matching and there exists a boundary-matched point $b\in B^{\mcC^*}$. In this case, there is a free point $a\in A$, where $\distance{a}{b}\le \diam <\distance{b}{\mcC^*}$, where $\diam$ is the diameter of the points $A\cup B$ i.e., matching the point $b$ to the free point $a$ instead of the boundaries of $\cell^*$ reduces the cost of the extended matching, which is a contradiction to the assumption that $M^{\mcC^*}$ is a minimum-cost extended matching. Hence, $B^{\mcC^*}=\emptyset$, and $M$ is a (minimum-cost) $n$-matching.
\end{proof}


We next analyze the running time of our algorithm, leading to the following lemma.

\begin{restatable}{lemma}{GRSAnalysiss}\label{lemma:GRS-analysis-p2}
    Suppose $U$ is a set of $2n$ points inside the unit square and $A$ is a subset chosen uniformly at random from all subsets of size $n$. Let $B=U\setminus A$.  Then, our algorithm computes the minimum-cost matching on the complete bipartite graph on $A$ and $B$ under squared Euclidean distances in $\tilde{O}(n^{7/4}\log \Delta)$ expected time.
\end{restatable} 


Similar to our efficiency analysis in Section~\ref{subsec:geometric-analysis}, we first show that for any cell $\cell$ of $\mcH$, there exists a low-cost high-cardinality matching inside $\cell$. We then use this matching to prove a bound on the number of iterations of the Hungarian search procedure.

\begin{restatable}{lemma}{GRSPartialMatching}\label{lemma:GRS-matching}
    For any cell $\cell$ of $\mcH$, there exists a matching $M'$ that, in expectation, matches all except $O(n_\cell^{3/4})$ points of $B_\cell$ and has a cost $O(\ell_\cell^2 n_\cell^{1/4})$.
\end{restatable}


For any cell $\cell$ of $\mcH$, we next show that the number of iterations of the merge step and the total number of executions of the extended Hungarian search procedure on $\cell$ is $\tilde{O}(n^{3/4})$.

\myparagraph{Number of Iterations.} 
Let $B_\cell^\mcC$ denote the set of all boundary-matched points of $B_\cell$ that are matched to the divider $\divider\cell$ in the matching maintained by our algorithm before executing the merge step on $\cell$. We first show in Lemma~\ref{lemma:matching-iter} that the total number of iterations of the merge step and the extended Hungarian search procedure on $\cell$ is $|B_\cell^\mcC|$.

\begin{restatable}{lemma}{matchingIterations}\label{lemma:matching-iter}
    For any cell $\cell$ of $\mcH$, the total number of iterations of the merge step and the total number of executions of the search procedure on $\cell$ is $O(|B_\cell^\mcC|)$.
\end{restatable}

We next show that $|B_\cell^\mcC|=O(n^{3/4})$. 
Since each iteration can be executed in $\tilde{O}(n_\cell\Phi(n))$ time, the total processing time of our algorithm for each cell $\cell$ of $\mcH$ is $\tilde{O}(n^{3/4}n_\cell\Phi(n))$, and the total execution time across all cells of $\mcH$ would be $\tilde{O}(n^{7/4}\Phi(n)\log\Delta)$, leading to Lemma~\ref{lemma:GRS-analysis-p2}.

\begin{restatable}{lemma}{GRSIters}\label{lemma:GRS-cell-process}
    For any cell $\cell$ of $\mcH$, $|B_\cell^\mcC|=O(n^{3/4})$.
\end{restatable}
\begin{proof}
    Let $M'$ denote the matching constructed in Lemma~\ref{lemma:GRS-matching}, and let $\mcP_{\mathrm{aug}}$ (resp. $\mcP_{\mathrm{alt}}$) denote the set of augmenting paths (resp. alternating paths) with an endpoint in $B_\cell^\mcC$ in the symmetric difference $M_\cell\oplus M'$. Then, $|\mcP_{\mathrm{alt}}|=O(n^{3/4})$ since the number of alternating paths $P$ in $\mcP_{\mathrm{alt}}$ is bounded by the number of free points in $M'$. 
    We partition the free endpoints of $\mcP_{\mathrm{aug}}$ into the set $B^{\mathrm{close}}_{\mathrm{aug}}$ that are at a distance closer than $\lambda'_\cell=\ell_\cell n^{-1/4}$ to $\divider\cell$ and the set $B^{\mathrm{far}}_{\mathrm{aug}}$ that are at a distance further than $\lambda'_\cell$ from $\divider\cell$. By Lemma~\ref{lemma:margin}, only $O(n_\cell^{3/4})$ points are close to $\divider\cell$ and $|B^{\mathrm{close}}_{\mathrm{aug}}|=O(n^{3/4})$. 
    For points in $B^{\mathrm{far}}_{\mathrm{aug}}$, we show that the dual weight of the points in $B^{\mathrm{far}}_{\mathrm{aug}}$ is bounded by the cost of $M'$, and each point in $B^{\mathrm{far}}_{\mathrm{aug}}$ has a dual weight of at least $\lambda'_\cell = \ell_\cell n^{-1/4}$. Therefore, $|B^{\mathrm{far}}_{\mathrm{aug}}|\le \frac{w(M')}{\lambda'_\cell}= O(n^{4/5})$. The complete proof is provided in Appendix~\ref{appendix-randomly-colored}.
\end{proof}

Each iteration of the merge step and each execution of the search procedure takes $\tilde{O}(n_\cell\Phi(n))$ time and therefore, the total execution time of our algorithm on $\cell$ would be $\tilde{O}(n^{3/4}n_\cell\Phi(n))$. Since each point participates in $O(\log(n\Delta))$ cells in $\mcH$, the total execution time of our algorithm across all cells would be $\tilde{O}(n^{7/4}\Phi(n)\log\Delta)$, proving Lemma~\ref{lemma:GRS-analysis-p2}.



\section{Conclusion and Open Questions}\label{sec:conclusion}
In this paper, we presented an exact algorithm for the $k$-SP and $k$-SPI problems for the case where requests are $d$-dimensional points with a spread bounded by $\Delta$. Our algorithm makes sub-quadratic calls to a weighted nearest-neighbor data structure and has a logarithmic dependence on the spread $\Delta$.  We conclude by raising two important open questions:

First, can we eliminate the $O(\log \Delta)$ dependency in the execution time of our algorithm? Achieving this would enable sub-quadratic algorithms for geometric bipartite matching, providing significant progress on a longstanding open question in computational geometry. 

Second, we reduced the computation of the optimal offline solution for the $k$-server problem to finding the minimum-cost bipartite matching on a dense bipartite graph $\mathcal{G}_\requests$. This reduction also proposes the $k$-server problem (online version) as a mild variant of the online metric matching problem. There has been extensive work on designing algorithms for the online $k$-server problem~\cite{bubeck2018k,S7,S16,S3, jamesLee2018, S14, raghvendra2022scalable, rudec2013fast}. By leveraging this reduction, the implementation of the work function algorithm, such as the scalable approach by Raghvendra and Sowle~\cite{raghvendra2022scalable} can be simplified and improved. Furthermore, the best-known algorithm for online metric matching, the RM algorithm~\cite{my_focs_paper, raghvendra2016robust}, can be adapted to solve the online $k$-server problem. Notably, it has been shown that the RM algorithm has a better competitive ratio compared to the work function algorithm for online bipartite matching on a line metric~\cite{s_socg18}. Can we derive better upper bounds for the competitive ratio of the RM algorithm when applied to the $k$-server problem? Such results would mark progress toward proving the deterministic $k$-server conjecture. 

\section*{Acknowledgement}
The research presented in this paper was funded by NSF CCF-2223871. We thank the anonymous reviewers for their useful feedback.

\bibliographystyle{plainnat}
\bibliography{main}

\newpage
\appendix
\section{Relating Bipartite Matching Problem with the \texorpdfstring{$k$}{-}-SP and \texorpdfstring{$k$}{-}-SPI Problems}

\subsection{Reducing \texorpdfstring{$k$}{}-SP and \texorpdfstring{$k$}{}-SPI Problems to the Geometric Bipartite Matching Problem}\label{sec:background}
In this section, we show that the $k$-sequence partitioning problem (resp. the $k$-sequence partitioning with initial locations problem) reduces to the problem of computing a minimum-cost matching of size $n-k$ (resp. $n$) in a bipartite graph with $2n$ (resp. $2n+k$) vertices, leading to Lemma~\ref{lemma:matching-reduction}. 


\myparagraph{Graph Representation.} Given an input sequence $\requests$ for the $k$-SP problem, recall that in our reduction, we create two vertex sets $A$ and $B$ containing $n$ points each. For each request $r_i$, we create a vertex $b_i$ (resp. $a_i$) in $B$ (resp. $A$) as the entry (resp. exit) gate for request $r_i$. The set of edges is created as follows. For any request $r_i$ and every $j > i$, we connect the exit gate $a_i$ of request $r_i$ to the entry gate $b_j$ of request $r_j$ and assign it a cost $\distance{a_i}{b_j}=\|r_i-r_j\|_p$. 
This completes the construction of the graph $\mcG_\requests$. Note that $|A|=|B|=n$.

\myparagraph{Relating a Partial Matching to a $k$-Partitioning.} 
Consider any matching $M$ of size $(n-k)$ in $\mcG_\requests$. There are exactly $k$ requests whose entry gates are free in $M$ and $k$ requests whose exit gates are free in $M$. All these requests have at most one edge incident on them. Every other request has edges incident on both the entry and exit gates. Therefore, with respect to the matching edges, $k$ requests have a degree at most one and all other requests have a degree exactly two. It is easy to see that these edges do not form a cycle, as there is no edge between the exit gate $a_i$ of a request $r_i$ to the entry gate $b_j$ of a request $r_j$ with $j<i$; therefore, the set of matching edges partitions the requests into a set of $k$ paths, representing the $k$-partitioning for the $k$-SP problem. The following lemma states and proves our reduction more formally. 



\begin{restatable}{lemma}{relationToMatching}\label{lemma:relationToMatching}
    Computing a minimum-cost $k$-partitioning for a sequence of requests $\requests$ reduces to computing a minimum-cost $(n-k)$-partial matching in the graph representation $\mcG_\requests$.
\end{restatable}
\begin{proof}
Consider any $k$-partitioning $\plan=\{\varsigma_1,\ldots, \varsigma_k\}$ of $\requests$. We construct a matching $M_\plan$ representing $\plan$ as follows. For any subsequence $\varsigma_i\in \plan$ and any pair of consecutive requests $r_{i_1}, r_{i_2}$ in $\varsigma_i$, we add the edge $(a_{i_1}, b_{i_2})$ to the matching $M_\plan$. Since each request $r_i\in\requests$ is included in exactly one subsequence of $\plan$, at most one edge of $M_\plan$ is incident on each point of $A\cup B$. Furthermore, for each sub-sequence $\varsigma\in\plan$, among all entry and exit gates of the requests in $\varsigma$, only the entry gate of the first request and the exit gate of the final request in $\varsigma$ are free in $M_\plan$. Hence, $M_\plan$ is a matching of size $n-k$ and by the definition of the cost of $k$-partitionings and matchings, has the same cost as $\plan$.

Next, given a $(n-k)$-partial matching $M$ on $\mcG_\requests$, we construct a $k$-partitioning $\plan_M$ of the requests $\requests$ with the same cost as $M$. Our construction relies on the following observation: since $M$ is a $(n-k)$-partial matching, all except $k$ points in $B$ are matched in $M$. For each request $r_i\in\requests$ whose entry gate $b_i$ is free in $M$, we construct a subsequence $\varsigma_i$ of requests by following the matching edges as described next. Initialize $\varsigma_i=\langle r_i\rangle$. At any step, suppose $\varsigma_i=\langle r_{i_1}, \ldots, r_{i_j}\rangle$. If the exit gate $a_{i_j}$ of the last request $r_{i_j}$ is free in $M$, we stop the construction of $\varsigma_i$; otherwise, $a_{i_j}$ is connected to an entry gate $b_{l}$ for a request $r_{l}$ with $l>i_j$. Add $r_l$ to the subsequence $\varsigma_i$ as $r_{i_{j+1}}$. This completes the construction of $k$ subsequences, one for each of the $k$ free entry gates. It is easy to confirm that the resulting subsequences are a $k$-partitioning of $\requests$ and its cost is equal to the cost of the matching $M$.
\end{proof}

\myparagraph{Extension to the $k$-SPI Problem.}
For the $k$-sequence partitioning with initial locations problem, given a sequence $\requests$ of $n$ requests and a set $\servers$ of the initial locations of the $k$ servers, we create two vertex sets $A$ and $B$, where $|A|=n+k$ and $|B|=n$. For each request $r_i\in\requests$, we create a vertex $b_i$ (resp. $a_i$) in $B$ (resp. $A$) as the entry (resp. exit) gate for request $r_i$. We also add a vertex $a^j$ for each server $s_j\in\servers$ to $A$. The set of edges is created as follows. For any server $s_j$, we add an edge from $a^j$ to the entry gate $b_i$ of all requests $r_i$ with a cost $\distance{a^j}{b_i}=\|s_j-r_i\|_p$. For any request $r_i$ and every $j > i$, we connect the exit gate $a_i$ of request $r_i$ to the entry gate $b_j$ of request $r_j$ and assign it a cost $\distance{a_i}{b_j}=\|r_i-r_j\|_p$.  See Figure~\ref{fig:enter-label}. The following lemma, whose proof is a straightforward extension of the proof of Lemma~\ref{lemma:relationToMatching}, states that the $k$-SPI problem reduces to the minimum-cost maximum matching problem on $\mcG_\requests$.

\begin{figure}
    \centering    \hspace{2em}\includegraphics[width=0.4\linewidth]{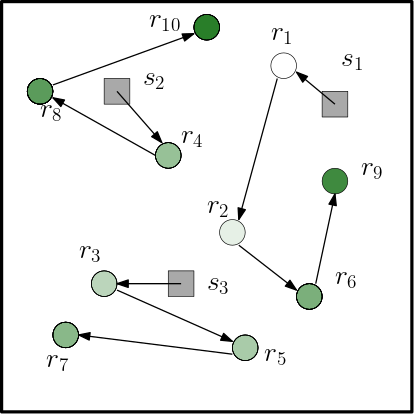}\hfill
    \includegraphics[width=0.4\linewidth]{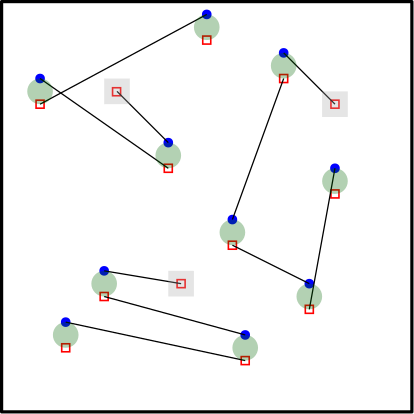}\hspace{2em}
    \caption{(Left) A $3$-partitioning $\langle s_1, r_1, r_2, r_6, r_9\rangle$, $\langle s_2, r_4, r_8, r_{10}\rangle$, $\langle s_3, r_3, r_5, r_7\rangle$ for the $3$-SPI problem and (right) its corresponding maximum matching on the graph representation of the problem.}
    \label{fig:enter-label}
\end{figure}

\begin{restatable}{lemma}{relationToMatchingServers}\label{lemma:relationToMatchingServers}
    Computing a minimum-cost $k$-partitioning with initial locations for a sequence of requests $\requests$ and a set of servers $\servers$ reduces to computing a minimum-cost maximum matching in the graph representation $\mcG_{\servers\requests}$.
\end{restatable}
\begin{proof}
For any $k$-partitioning $\plan=\{\varsigma_1,\ldots, \varsigma_k\}$ of $\requests$, we construct a matching $M_\plan$ with the same cost as follows. For any subsequence $\varsigma_i=\langle s_i, r_{i_1}, \ldots, r_{i_t}\rangle$, we add the edge $(a^i, b_{i_1})$ along with $(a_{i_j}, b_{i_{j+1}})$ for each $j\in[1, t-1]$ to $M_\plan$. Each server and each request is included in exactly one subsequence of $\plan$ and therefore, each entry gate of each request has a degree exactly one in $M_\plan$ and the exit gate of each server and each request is incident on at most one edge of $M_\plan$. Hence, $M_\plan$ is a maximum matching, and by the construction of $\mcG_\requests$, $w(M_\plan)$ is the same as the cost of $\plan$.

Next, for a maximum matching $M$ on $\mcG_\requests$, we construct a $k$-partitioning $\plan_M$ of the servers and requests $\servers\requests$ with the same cost as $M$. For each server $s_i$, we construct a subsequence $\varsigma_i$ as follows. If $a^i$ is free, then we add the subsequence $\varsigma_i=\langle s_i\rangle$ to $\plan_M$. Otherwise, $a^i$ is matched to a point $b_j$. Initialize $\varsigma_i=\langle s_i, r_j\rangle$. At any step, suppose $\varsigma_i=\langle s_i, r_{i_1}, \ldots, r_{i_j}\rangle$. If the exit gate $a_{i_j}$ of the last request $r_{i_j}$ is free in $M$, we stop the construction of $\varsigma_i$; otherwise, $a_{i_j}$ is connected to an entry gate $b_{l}$ for a request $r_{l}$ with $l>i_j$. Add $r_l$ to the subsequence $\varsigma_i$ as $r_{i_{j+1}}$. This completes the construction of $k$ subsequences, one for each of the $k$ servers. By construction, the cost of $\plan_M$ is $w(M)$.
\end{proof}

\subsection{Reducing Geometric Bipartite Matching to the \texorpdfstring{$k$}{}-SPI Problem}\label{sec:reduction}

Given an instance of the geometric bipartite matching on $d$-dimensional points $A=\{a_1,\ldots, a_n\}$ and $B=\{b_1,\ldots, b_n\}$ under the $\ell_1$ costs, we create a $(d+1)$-dimensional instance of the $k$-SPI problem. The points of $B$ will be used to create the initial locations of the $n$ servers, and the points of $A$ will be used to create the request sequence. Our construction will be such that the optimal solution to this $n$-SPI problem will create sub-sequences of length $2$ mapping the initial location of each of the $n$ servers to exactly one request. These subsequences can be shown to also correspond to the minimum-cost matching of $A$ and $B$. We provide the details below. 

Let $\diam$ denote the diameter of $A\cup B$ under $\ell_1$ costs. We create two sets $A'$ and $B'$ of $(d+1)$-dimensional points as follows. For each point $a_i \in A$, we create a point $a_i' \in A'$ by setting the $(d+1)$st coordinate of $a_i$ to $3^{n-i+1}\diam$. For each point $b_i \in B$, we create a point $b_i' \in B'$ and set its $(d+1)$st coordinate to $0$. Consider the $n$-SPI problem with $\servers=\langle s_1=b_1',\ldots, s_n=b_n'\rangle$ and $\requests=\langle r_1=a_1',\ldots, r_n=a_n'\rangle$ as inputs. By our construction, for any request $r_i$, any server $s_l$, and any request $r_j$ with $j<i$,
\begin{align}
    \|r_i- s_l\|_1 &= \|a'_i- b'_l\|_1 = 3^{n-i+1}\diam + \|a_i- b_l\|_1\le 3^{n-i+1}\diam + \diam \nonumber \\&<  (3^{n-j+1}-3^{n-i+1})\diam \le (3^{n-j+1}-3^{n-i+1})\diam + \|a_i- a_j\|_1 \nonumber \\&= \|a'_i- a'_j\|_1= \|r_i-r_j\|.\label{eq:distance-EBM}
\end{align}
Intuitively, for any request $r_i$, it would be cheaper to directly serve $r_i$ by a new server (e.g., $s_l$) rather than serving $r_i$ by a server used to serve an earlier request (e.g., $r_j$).
We next formally show that the optimal $n$-partitioning under $\ell_1$ norm matches each of the $n$ servers to exactly one request and the matching of servers to requests is also an optimal matching of points in $A$ to $B$. 

\begin{figure}
    \centering
    \includegraphics[width=0.45\textwidth]{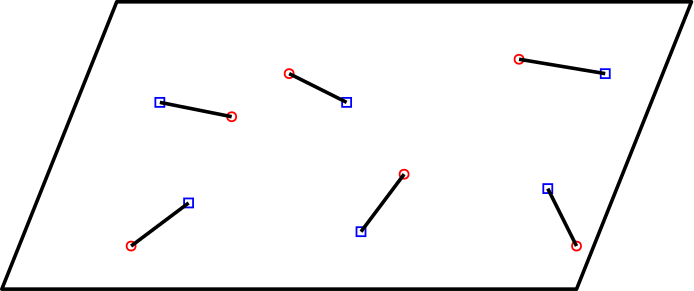}
    \includegraphics[width=0.45\textwidth]{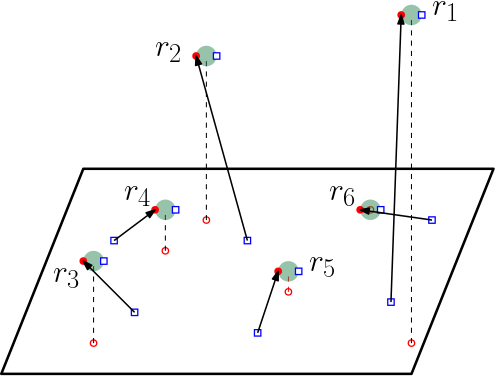}
    \caption{(left) the point set $A$ (red points) and $B$ (blue points) and a minimum-cost bipartite matching between them, (right) the instance of $n$-SPI problem we construct from $A\cup B$.}
    \label{fig:EBM}
\end{figure}

Let $\plan=\{\varsigma_1, \ldots, \varsigma_n\}$ be an optimal $n$-partitioning for $\requests'=\servers\requests$, where $\varsigma_i$ is the sub-sequence of $\requests'$ that is assigned to be served by the server $s_i$. Suppose there exists a server $s_i\in\servers$ with more than one request in $\varsigma(s_i)$, i.e., $\varsigma(s_i)=\langle s_i, r_{i_1}, \ldots, r_{i_t}\rangle$ for some $t>1$. In this case, since there are $n$ servers in $\servers$, there exists at least one server $s_j\in\servers$ with $\varsigma_j = \langle s_j\rangle$. By Equation~\eqref{eq:distance-EBM}, $\|s_j-r_{i_t}\|_1<\|r_{i_{t-1}}-r_{i_t}\|_1$. Therefore, assigning the request $r_{i_t}$ to $s_j$ reduces the cost of $\plan$, which contradicts the assumption that $\plan$ is optimal; thus, any optimal $n$-partitioning of $\requests'$ matches each server to a single request.

Note that any matching $M$ of $A$ and $B$ can be used to construct an $n$-partitioning $\plan_M$ of $A'$ and $B'$. For any vertex $a_i \in A$, let $b_{m(i)}$ be its match in $M$. 
\begin{align}
    w(M)& = \sum_{i=1}^n \|a_i-b_{m(i)}\|_1=  \sum_{i=1}^n \|a_i'-b_{m(i)}'\|_1 -  3^{n-i}\diam = w(\plan_M)-\sum_{i=1}^n 3^{n-i}\diam.\label{eq:matching-EBM}
\end{align}
In other words, for any matching $M$ and $n$-partitioning $\plan_M$, the difference in the cost of $M$ and $\plan_M$ is a constant that is independent of the input. Since the optimal $n$-partitioning maps each server to a unique request and has the smallest possible cost, its corresponding matching is also a matching of $A$ and $B$ with the smallest possible cost.   
Therefore, any algorithm that finds an optimal partitioning for the sequence partitioning problem can be used to compute a geometric bipartite matching.
Note that our reduction extends to $\ell_q^q$ distances for any parameter $q\ge 1$ in a straightforward way.

\section{Missing Proofs and Details of Section~\ref{sec:geo-primal-dual}}\label{sec:appendix-geo-primal-dual}
In this section, we first present a sweep-line algorithm for constructing the hierarchical partitioning described in Section~\ref{sec:hierarchical}. We then present the missing proofs of lemmas in Sections~\ref{sec:geo-primal-dual}.

\myparagraph{Constructing the hierarchical partitioning.} Recall that, as described in Section~\ref{sec:hierarchical}, our algorithm constructs a hierarchical partitioning by recursively partitioning each non-leaf cell of $\mcH$ into two smaller cells as its children. We describe the partitioning procedure with respect to an arbitrary cell $\cell$ of $\mcH$. Let $\ell_x$ and $\ell_y$ denote the width and length of $\cell$, and without loss of generality, assume $\ell_x\ge \ell_y$. Let $x_{\min}$ (resp. $x_{\max}$) denote the $x$ coordinate of the left (resp. right) side of $\cell$. Recall that $\lambda=9n^{-1/5}$. Define a priority queue $\mcQ$ storing two events for each point $u\in A_\cell\cup B_\cell$, namely an entry event $e_u^\downarrow$ at $u_x-\ell_\cell\lambda$ and an exit event $e_u^\uparrow$ at $u_x + \ell_\cell\lambda$; here, $u_x$ denotes the $x$ coordinate of point $u$. See the events $b^\downarrow$ and $b^\uparrow$ in Figure~\ref{fig:sweepline}.

Consider a vertical sweep line that moves from left to right, searching for a value $x^* \in [x_{\min} + \frac{\ell_x}{3}, x_{\min} + \frac{2\ell_x}{3}]$ that minimizes $|\Lambda(x^*)|$; recall that $\Lambda(x):=\{u\in A_\cell\cup B_\cell: |u_x-x|\le \ell_\cell\lambda\}$ for any $x\in [x_{\min} + \frac{\ell_x}{3}, x_{\min} + \frac{2\ell_x}{3}]$. Let $x'$ denote the state of the sweep line. Intuitively, as the sweep line moves from left to right (i.e., $x'$ increases), any point $u\in A_\cell\cup B_\cell$ enters the set $\Lambda(x')$ at its entry event and exits $\Lambda(x')$ at its exit event. Thus, we can keep track of the number of points close to the sweep line by incrementing (resp. decrementing) the count at each entry (resp. exit) event and return the value $x^*$ realizing the minimum count during the sweep-line procedure. 

More formally, we maintain a value $\gamma$ representing the number of points of $A_\cell\cup B_\cell$ at a distance at most $\ell_\cell\lambda$ to the sweep line. We also store, for a set of $O(n_\cell)$ values $x''\in [x_{\min}, x_{\max}]$, a function $\eta:[x_{\min}, x_{\max}] \rightarrow \mbZ_{\ge 0}$ representing $|\Lambda(x'')|$. Initially, for the minimum event $e$ in $\mcQ$, set $\gamma\leftarrow 1$ and $\eta(e)\leftarrow \gamma$. Iteratively, our sweep-line algorithm extracts the minimum event $e$ from $\mcQ$. If $e$ is an entry event, set $\gamma\leftarrow \gamma + 1$, and otherwise, $e$ is an exit event and set $\gamma\leftarrow \gamma - 1$. Set $\eta(e)\leftarrow \gamma$. After processing all events in $\mcQ$, define the entry $x'$ of $\eta(\cdot)$ in the middle third of $x$ side of $\cell$ with the minimum value, i.e., \[x^*:=\arg\min_{x\in [x_{\min}+\frac{\ell_\cell}{3}, x_{\min}+\frac{2\ell_\cell}{3}]\cap S(\eta)}\eta(x),\] where $S(\eta)$ denotes the support of $\eta$. Partition $\cell$ using the vertical line $x=x^*$ into two smaller cells $\cell_1$ and $\cell_2$, and add the non-empty ones as the children of $\cell$ to $\mcH$. 
\begin{figure}
    \centering
    \includegraphics[width=0.7\linewidth]{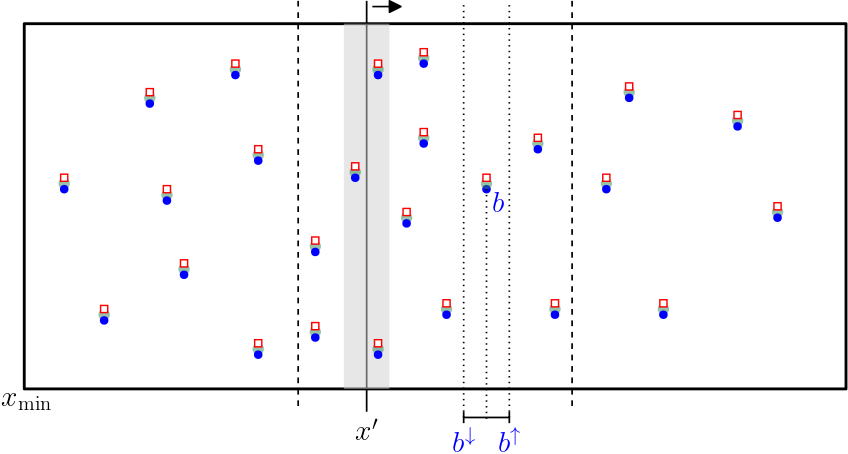}
    \caption{The solid vertical line shows the line our sweep-line procedure maintains. The two values $b^\downarrow$ and $b^\uparrow$ show the entry and exit events created for the blue point $b$.}
    \label{fig:sweepline}
\end{figure}

\subsection{Missing Proofs of Section~\ref{sec:geo-primal-dual}}

\margin*
\begin{proof}
    We use an inductive proof to show that the aspect ratio of all cells of $\mcH$ is at most $3$. Note that the root cell $\cell^*$ is square and has a unit aspect ratio. For any non-root cell $\cell\in \mcH$, let $\cell'$ denote the parent of $\cell$ in $\mcH$. By the inductive hypothesis, the aspect ratio of $\cell'$ is at most $3$. Let $\ell_x$ (resp. $\ell_y$) denote the width (resp. length) of $\cell$. Similarly, let $\ell'_x$ (resp. $\ell'_y$) denote the width (resp. length) of $\cell'$. Without loss of generality, assume $\ell'_x\ge \ell'_y$. In this case, we split $\cell'$ into two rectangles using a vertical line in the middle third part of the $x$ side of $\cell'$, i.e., $\ell_y=\ell'_y$ and $\ell_x\in [\frac{\ell'_x}{3}, \frac{2\ell'_x}{3}]$. In this case, if $\ell_x\ge \ell_y$, then $\frac{\ell_x}{\ell_y}\le \frac{\ell'_x}{\ell'_y}\le 3$. Otherwise, $\ell_x< \ell_y$ and $\frac{\ell_y}{\ell_x}\le \frac{\ell'_y}{\ell'_x/3}\le 3$.

    Next, we show that the number of points of $A_\cell\cup B_\cell$ at a distance at most $\ell_\cell\lambda$ from the divider $\divider\cell$ is $O(n_\cell \lambda)$. Consider the set $\zeta=\{x_{\min}+\frac{\ell_\cell}{3} + 3i \ell_\cell\lambda: 0\le i\le \lfloor \frac{1}{9\lambda}\rfloor\}$. Recall that for any value $x'\in[x_{\min}+\frac{\ell_\cell}{3}, x_{\min}+\frac{2\ell_\cell}{3}]$, $\Lambda(x')$ denotes the set of all points of $A_\cell\cup B_\cell$ that are at a distance at most $\ell_\cell\lambda$ from the vertical line $x=x'$. 
    Note that for any pair of distinct values $x_1, x_2\in \zeta$, $\Lambda(x_1)\cap\Lambda(x_2)=\emptyset$. See Figure~\ref{fig:sweepline2}. 
    \begin{figure}
        \centering
        \includegraphics[width=0.9\linewidth]{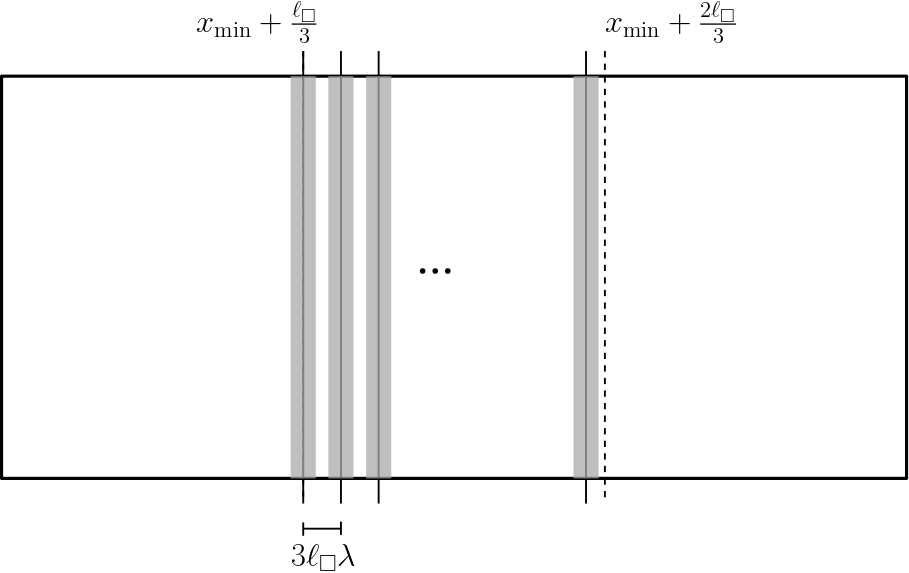}
        \caption{The solid vertical line shows the set $\zeta$.}
        \label{fig:sweepline2}
    \end{figure}
    Furthermore, 
    \[
    \sum_{x'\in \zeta} |\Lambda(x')| \le n_\cell.\]
    Define $x^*:=\arg\min_{x'\in\zeta}|\Lambda(x')|$. In this case, the size of $\Lambda(x^*)$ is no more than the average size of the $\Lambda(x')$ for the values $x'\in \zeta$; more precisely,
    \[|\Lambda(x^*)| \le \frac{\sum_{x'\in \zeta} |\Lambda(x')|}{|\zeta|} \le 9\lambda n_\cell. \]
\end{proof}

\geometricAnalysis*
\begin{proof}
    Define $\mbG$ to be a grid of cell-side-length $\ell_\cell n_\cell^{-2/5}$ that partitions $\cell$ into $O(n_\cell^{4/5})$ equal-sized rectangles. For each rectangle $\xi$ of the grid $\mbG$, let $\requests_\xi=\langle r'_1, \ldots, r'_m\rangle$ denote the sub-sequence of requests in $\requests$ that lie inside $\xi$, and let $A_\xi$ (resp. $B_\xi$) denote the subset of $A_\cell$ (resp. $B_\cell$) that lie inside $\xi$. Define $\mcG_\xi$ as the sub-graph of $\mcG$ induced by $A_\xi\cup B_\xi$. Let $M_\xi$ denote the set of edge $(a_{r'_i}, b_{r'_{i+1}})$ for all $i\in[1, m-1]$. $M_\xi$ is a matching on $\mcG_\xi$ that matches all except one point of $B_\xi$. Furthermore, since each matching edge in $M_\xi$ has a cost at most $2\ell_\cell n_\cell^{-2/5}$, the cost of $M_\xi$ would be $O(\ell_\cell |B_\xi| n_\cell^{-2/5})$.

    \begin{figure}
        \centering
        \includegraphics[width=0.5\linewidth]{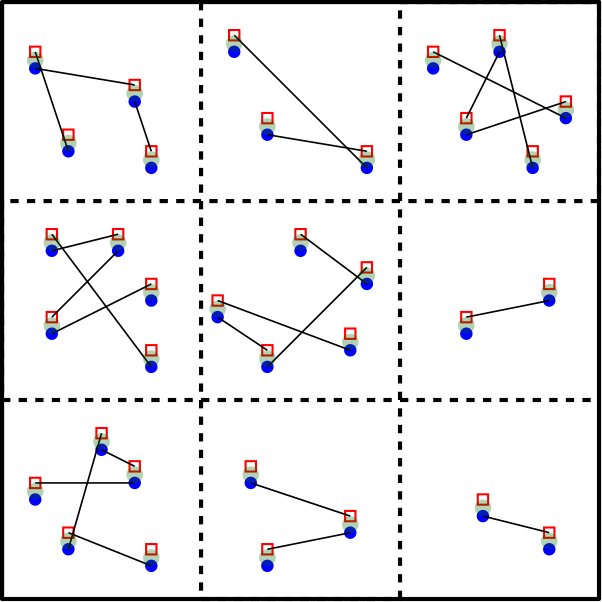}
        \caption{A matching $M_\xi$ is computed inside each square $\xi$ by connecting the requests in their arrival order.}
        \label{fig:geo}
    \end{figure}
    
    Define $M':=\bigcup_{\xi\in\mbG}M_\xi$ as the union of the matchings computed inside each cell of the grid $\mbG$. Since $|\mbG|=O(n_\cell^{4/5})$ and each cell leaves at most one point of $B_\cell$ unmatched, the matching $M'$ matches all except $O(n_\cell^{4/5})$ points of $B_\cell$. In addition, the cost of $M$ is at most \[w(M')=\sum_{\xi\in\mbG}w(M_\xi)=O\left(\ell_\cell n_\cell^{-2/5}\sum_{\xi\in\mbG}|B_\xi|\right) = O(\ell_\cell n_\cell^{3/5}).\]
\end{proof}

We next present a set of definitions, which are used in the proof of Lemma~\ref{lemma:root-inactive}. Given a matching $M$ on the graph $\mcG_\sigma$, any simple path $P$ is an alternating path if the edges of $P$ alternate between matching and non-matching edges with respect to $M$. The path $P$ is called a (standard) augmenting path if it starts from an unmatched point $b\in B$ and ends at an unmatched point $a\in A$. One can augment $M$ along $P$ by setting $M\leftarrow M\oplus P$. The net-cost of $P$ is defined as $\phi(P):=\sum_{(a,b)\in P\setminus M}d(a,b) - \sum_{(a,b)\in P\cap M}d(a,b)$, i.e., the net-cost of $P$ is the change in the matching cost introduced by augmenting $M$ along $P$.

\rootinactive*
\begin{proof}
    For the sake of contradiction, suppose the matching $M$ is not a $t$-matching, and there exist boundary-matched points in $B^{\mcC^*}$. To prove this lemma, we show that there exists an extended matching $\hat{M}^{\mcC^*}=(\hat{M}, \hat{B}^{\mcC^*})$ of size $t$ such that $|\hat{M}|>|M|$ and $w_{\mcC^*}(\hat{M}^{\mcC^*})<w_{\mcC^*}(M^{\mcC^*})$, contradicting the assumption that $M^{\mcC^*}$ is a minimum-cost extended $t$-matching. 

    Consider a minimum net-cost augmenting path $P$ with respect to the matching $M$ from an unmatched point $b\in B$. Since the path $P$ has a length of at most $2n-1$, the net-cost of $P$ is 
    \begin{align*}
        \phi(P) = \sum_{(a,b)\in P\setminus M}d(a,b) - \sum_{(a,b)\in P\cap M}d(a,b) \le \sum_{(a,b)\in P\setminus M}d(a,b) \le 2n,  
    \end{align*}
    where the last inequality holds since all points are in the unit square and the $\ell_p$ distance of each pair is at most $2$. By augmenting the matching $M$ along the path $P$, we get a matching $\hat{M}$ with a cost $w(\hat{M})=w(M)+\phi(P)\le w(M) + 2n$. Also note that by the construction of $\cell^*$, for any point $b'\in B$, $\distance{b'}{\cell^*}\ge 3n-1$. If the endpoint $b$ of $P$ is in $B^{\mcC^*}$, then for the extended matching $\hat{M}^{\mcC^*}=(\hat{M}, B^{\mcC^*}\setminus\{b\})$, 
    \begin{align*}
        w_{\mcC^*}(\hat{M}^{\mcC^*})&= w(\hat{M}) + \sum_{b'\in B^{\mcC^*}\setminus\{b\} }\distance{b'}{\mcC^*}
        \le w(M) + 2n + \sum_{b'\in B^{\mcC^*}\setminus\{b\} }\distance{b'}{\mcC^*}\\&<w(M) + \distance{b}{\mcC^*} + \sum_{b'\in B^{\mcC^*}\setminus\{b\}}\distance{b'}{\mcC^*}=
        w_{\mcC^*}(M^{\mcC^*}).
    \end{align*}
    Otherwise, let $\hat{b}\in B^{\mcC^*}$ denote an arbitrary boundary-matching point of $M^{\mcC^*}$. Then, for the extended matching $\hat{M}^{\mcC^*}=(\hat{M}, B^{\mcC^*}\setminus\{\hat{b}\})$, 
    \begin{align*}
        w_{\mcC^*}(\hat{M}^{\mcC^*})&= w(\hat{M}) + \sum_{b'\in B^{\mcC^*}\setminus\{\hat{b}\} }\distance{b'}{\mcC^*}
        \le w(M) + 2n + \sum_{b'\in B^{\mcC^*}\setminus\{\hat{b}\} }\distance{b'}{\mcC^*}\\&<w(M) + \distance{\hat{b}}{\mcC^*} + \sum_{b'\in B^{\mcC^*}\setminus\{\hat{b}\}}\distance{b'}{\mcC^*}=
        w_{\mcC^*}(M^{\mcC^*}).
    \end{align*}
    Note that in both cases, $\hat{M}^{\mcC^*}$ is also of size $t$, which is a contradiction to the assumption that $M^{\mcC^*}$ is a minimum-cost extended $t$-matching. Therefore, the extended matching $M^{\mcC^*}=(M, B^{\mcC^*})$ has no boundary-matched points and $|M| = t$. 
    Consequently, $w_{\mcC^*}(M^{\mcC^*}) = w(M)$, and since $M^{\mcC^*}$ has the minimum cost among all extended $t$-matchings, the matching $M$ would have a minimum cost among all $t$-matchings on $\mcG_\sigma$, i.e., $M$ is a minimum-cost $t$-matching.
\end{proof}

\subsection{Auxiliary Lemmas}
The following lemmas are useful in proving the lemmas in Sections~\ref{sec:geo-primal-dual} and~\ref{sec:k-seq}.

\begin{restatable}{lemma}{slackcost}
\label{lem:slackcost}
Given a feasible $\mcC$-extended matching $M^{\mcC} = (M, B^{\mcC}),y(\cdot)$ on $\mcG_\sigma$, let $P$ be any augmenting path from a free point $b\in B$ with respect to $M^{\mcC}$. If $P$ ends with a free point of $A$, then $\phi(P) = y(b) +\sum_{(a'',b'') \in P}s(a'',b'')$. If $P$ ends with a point $b' \in B$, then $\phi(P)= y(b) + s(b') + \sum_{(a'',b'') \in P}s(a'',b'')$.  
\end{restatable}
\begin{proof}
    Note that by the definition of the slack of an edge, for any non-matching edge $(a,b)\in P$, $\distance{a}{b} = s(a,b) + y(b)-y(a)$. Furthermore, since the slack on the matching edges are $0$, for each matching edge $(a,b)\in M$, $-\distance{a}{b} = s(a,b) - y(b)+y(a)$. 
    If $P$ is an augmenting path from $b$ to a free point $a\in A$ (case (i)), then
    \begin{align*}
        \phi(P)&=\sum_{(a'',b'')\in P\setminus M} \distance{a''}{b''} + \sum_{(a'',b'')\in P\cap M} (-\distance{a''}{b''})\\&= \sum_{(a'',b'')\in P\setminus M} (s(a'',b'') + y(b'')-y(a'')) + \sum_{(a'',b'')\in P\cap M} (s(a'',b'') - y(b'')+y(a''))\\ &= \sum_{(a'',b'')\in P} s(a'',b'') + y(b) - y(a)=  y(b) \sum_{(a'',b'')\in P} s(a'',b''),
    \end{align*}
    where the last equality holds since $y(a)=0$ by Condition~\eqref{eq:dualfeasibility-a-free}. Otherwise, $P$ is an alternating path from $b$ to a point $b'\in B$ (case (ii)), and 
        \begin{align*}
            \phi(P)&=\distance{b'}{\mcC} + \sum_{(a'',b'')\in P\setminus M} \distance{a''}{b''} + \sum_{(a'',b'')\in P\cap M} (-\distance{a''}{b''})\\&= \distance{b'}{\mcC} + \sum_{(a'',b'')\in P\setminus M} (s(a'',b'') + y(b'')-y(a'')) + \sum_{(a'',b'')\in P\cap M} (s(a'',b'') - y(b'')+y(a''))\\ &= \distance{b'}{\mcC} + \sum_{(a'',b'')\in P} s(a'',b'') + y(b) - y(b')= y(b) + s(b') + \sum_{(a'',b'')\in P} s(a'',b''),
        \end{align*}
        where the last equality holds since $s(b') = \distance{b'}{\mcC} - y(b')$ by the definition of the slack of a point.
\end{proof}

The following is a straightforward corollary from Lemma~\ref{lem:slackcost} and the definition of admissible augmenting paths.

\begin{corollary}
\label{cor:slackcost}
Given a feasible $\mcC$-extended matching $M^{\mcC} = (M, B^{\mcC}),y(\cdot)$ on $\mcG_\sigma$, let $P$ be an admissible augmenting path from a free point $b\in B$ with respect to $M^{\mcC}$. Then, $\phi(P)= y(b)$.
\end{corollary}

The next lemma shows that no edges of a feasible $\mcC$-extended matching would cross the boundaries of the cells in $\mcC$, which allows us to localize the computations inside the cells and is essential for the correctness of our algorithm.

\begin{lemma}\label{lemma:no-cross}
    For any feasible extended matching $M^\mcC=(M, B^\mcC), y(\cdot)$, no edges of the matching $M$ cross the boundaries of the cells in $\mcC$.
\end{lemma}
\begin{proof}
    For the sake of contradiction, suppose there is an edge $(a,b)\in M$, where $a$ and $b$ lie inside a cell $\cell_a$ and $\cell_b$ of $\mcC$ and $\cell_a\neq \cell_b$. By the feasibility condition~\eqref{eq:dualfeasibility-matching},
    \begin{equation*}
        \distance{a}{b} = y(b) - y(a) \le y(b).
    \end{equation*}
    Since $a$ is outside of $\cell_b$, then $\distance{b}{\cell_b}<\distance{a}{b}$. Therefore, 
    \[y(b)\ge \distance{a}{b}> \distance{b}{\cell_b}=\distance{b}{\mcC},\]
    which is a contradiction to the assumption that $M^\mcC, y(\cdot)$ is feasible (Condition~\eqref{eq:dualfeasibility-b} is violated).
\end{proof}

\subsection{Proof of Lemma~\ref{lemma:augpath}}

\netcost*

\myparagraph{Proof of Lemma~\ref{lemma:augpath}(a).} For the sake of contradiction, suppose $P = \langle b_1, a_1, b_2, \ldots, b_m, a_m\rangle$ is a minimum net-cost augmenting path that does not lie inside a single cell of $\mcC$. Let $\cell$ denote the cell of $\mcC$ containing $b_1$, and let $(b_i, a_i)$ be the first edge that goes outside of $\cell$, i.e., all vertices $\{b_1, a_1,\ldots, b_i\}$ reside inside $\cell$ (Note that, by Lemma~\ref{lemma:no-cross}, no matching edges of $M$ cross the boundaries of $\mcC$ and the edge $(b_i, a_i)$ has to be a non-matching edge). For the alternating path $P'=\langle a_i, b_{i+1}, \ldots, a_{m}\rangle$, the net-cost of $P'$ is 
\begin{align*}
    \phi(P')&=\sum_{j=i+1}^m \distance{a_{j}}{b_{j}} - \sum_{j=i}^{m-1} \distance{a_{j}}{b_{j+1}}\\&= \sum_{j=i+1}^m (s(a_j,b_j) + y(b_j)-y(a_j)) + \sum_{j=i}^{m-1} (s(a_j,b_{j+1}) - y(b_{j+1})+y(a_j))\\ &= \sum_{j=i+1}^m s(a_j,b_j) + \sum_{j=i}^{m-1} s(a_j,b_{j+1}) + y(a_i) - y(a_m) \ge 0,
\end{align*}
where the last inequality holds since $a_m$ is an unmatched point and, by Condition~\eqref{eq:dualfeasibility-a-free}, has a zero dual weight, all edges have non-negative slacks, and all points have non-negative dual weights. 

Additionally, the cost of matching the point $b_i$ to the boundaries of $\cell$ is less than the cost of matching it to the point $a_i$ outside of $\cell$, i.e., $\distance{b_i}{\mcC}\le \distance{a_i}{b_i}$. Define $P''=\langle b_1, a_1, \ldots, b_i\rangle$. Then, 
\begin{align*}
    \phi(P) = \phi(P'')+\distance{a_i}{b_i} + \phi(P') \ge \phi(P'')+\distance{a_i}{b_i} > \phi(P'')+\distance{b_i}{\mcC}.
\end{align*}
Therefore, the augmenting path $P''$, which is an augmenting path from the free point $b_1$ to the point $b_i$ (case (ii)) has a lower net-cost than $P$, contradicting the assumption that $P$ is a minimum net-cost augmenting path.

\myparagraph{Proof of Lemma~\ref{lemma:augpath} (b).} To prove this lemma, we first construct a set of dual weights $y'(\cdot)$ (as described in Lemma~\ref{lemma:adjust_duals} below) such that $M^\mcC, y'(\cdot)$ is feasible, $y'(b)=\phi(P)$ for all free points $b\in B$ and $y'(b)\le \phi(P)$ for all points $b\in B$. We use the dual weights $y'(\cdot)$ as a certificate for the optimality of $M^\mcC$. 

Let $A_F$ (resp. $B_F$) denote the set of free points of $A$ (resp. $B$) with respect to $M^\mcC$. Let $y_{\max}:=\max_{b\in B}y'(b)$. By Condition~\eqref{eq:dualfeasibility-a-free}, for each point $a\in A_F$, $y'(a)=0$, and by Condition~\eqref{eq:dualfeasibility-b-admissible}, for each boundary-matched point $b\in B^\mcC$, $y'(b)=\distance{b}{\mcC}$.
Using Condition~\eqref{eq:dualfeasibility-matching-time}, we can rewrite the cost of $M^\mcC$ as
\begin{align}
    w_\mcC(M^\mcC) &= \sum_{(a,b)\in M} \distance{a}{b} + \sum_{b\in B^\mcC} \distance{b}{\mcC}\nonumber \\ &= \sum_{(a,b)\in M} y'(b) - y'(a) + \sum_{b\in B^\mcC} y'(b)\nonumber \\ &= \left(\sum_{b\in B} y'(b) - \sum_{a\in A} y'(a)\right) - \sum_{b\in B_F} y'(b) + \sum_{a\in A_F} y'(a)\nonumber \\ &= \left(\sum_{b\in B} y'(b) - \sum_{a\in A} y'(a)\right) - |B_F|\cdot y_{\max}.
    \label{eq:dualOptimal-proof-1-classic}
\end{align}
Let $\hat{M}^\mcC= (\hat{M}, \hat{B}^\mcC)$ denote any minimum-cost extended $t$-matching on $\mcG_\requests$. Let $\hat{A}_F$ (resp. $\hat{B}_F$) denote the set of points of $A$ (resp. $B$) that are free in $\hat{M}^\mcC$. Since both $M^\mcC$ and $\hat{M}^\mcC$ are $t$-matchings, $|B_F|=|\hat{B}_F|$. Using Conditions~\eqref{eq:dualfeasibility-non-matching-time} and~\eqref{eq:dualfeasibility-b},
\begin{align}
    w_\mcC(\hat{M}^\mcC) &= \sum_{(a,b)\in \hat{M}} \distance{a}{b} + \sum_{b\in \hat{B}^\mcC} \distance{b}{\mcC}\nonumber \\ &\ge \sum_{(a,b)\in \hat{M}} y'(b) - y'(a) + \sum_{b\in \hat{B}^\mcC} y'(b)\nonumber \\ &= \left(\sum_{b\in B} y'(b) - \sum_{a\in A} y'(a)\right) - \sum_{b\in \hat{B}_F} y'(b) + \sum_{a\in \hat{A}_F} y'(a)\nonumber \\ &\ge \left(\sum_{b\in B} y'(b) - \sum_{a\in A} y'(a)\right) - |\hat{B}_F|\cdot y_{\max},
    \label{eq:dualOptimal-proof-2-classic}
\end{align}
where the last inequality holds since $y'(b)\le y_{\max}$ for each point $b\in B$ and $y'(a)\ge 0$ for each point $a\in A$. Combining Equations~\eqref{eq:dualOptimal-proof-1-classic} and~\eqref{eq:dualOptimal-proof-2-classic},
\[w_\mcC(M^\mcC) = \sum_{b\in B} y'(b) - \sum_{a\in A} y'(a) - |B_F|\cdot y_{\max} \le w_\mcC(\hat{M}^\mcC). \]
Since $\hat{M}^\mcC$ is a minimum-cost extended $t$-matching, $w_\mcC(M^\mcC)=w_\mcC(\hat{M}^\mcC)$, and $M^\mcC$ is also a minimum-cost extended $t$-matching.

\subsection{Proof of Lemma~\ref{lemma:fresh_duals}}
\updateDuals*
\begin{proof}
    To prove this lemma, we first construct a residual graph $\mcG_\cell$ for each cell $\cell\in\mcC$ with respect to the extended matching $M^\mcC=(M, B^{\mcC})$. Our construction is identical to what is described in Section~\ref{sec:dual}. The vertex set of $\mcG_\cell$ is a source vertex $s$ and the points in $A_\cell\cup B_\cell$. For any edge $(a,b)\in E$ inside $\cell$, if $(a,b)$ is a matching edge (resp. non-matching edge) in $M$, there is an edge directed from $a$ to $b$ (resp. from $b$ to $a$) with a weight $s(a,b)$ in $\mcG_\cell$. Furthermore, there is an edge directed from $s$ to every free point $b\in B$ with a weight $y(b)$. 

    Define $\kappa_v$ as the distance of each point $v\in A_\cell\cup B_\cell$ from the source vertex $s$. Define $\kappa=\phi(P)$. For any vertex $v\in A_\cell\cup B_\cell$ with $\kappa_v<\kappa$, set $y'(v)\leftarrow y(v)-\kappa_v+\kappa$; otherwise, set $y'(v)\leftarrow y(v)$. As discussed in Section~\ref{sec:hungarian_nk}, the distances $\kappa_v$ can be computed in $\tilde{O}(n_\cell\Phi(n_\cell))$ time. Given the distances from the source, computing the set of dual weights $y'(\cdot)$ from $y(\cdot)$ takes $O(n_\cell)$ time, and therefore, the total construction takes $\tilde{O}(n_\cell\Phi(n_\cell))$ time. 
    
    We next show that $M^\mcC, y'(\cdot)$ is feasible, $y'(b)\le \phi(P)$ for all points $b\in A_\cell\cup B_\cell$, and $y'(b_f)=\phi(P)$ for all free points $b_f\in B_\cell$. Note that by Lemma~\ref{lemma:no-cross}, since the extended matching $M^\mcC, y(\cdot)$ is feasible, then no matching edges of $M^\mcC$ cross the boundaries of the cells in $\mcC$.

    \myparagraph{Feasibility of points.}
    \begin{enumerate}
        \item For any point $b\in B\setminus B_\cell$, $y'(b)=y(b)$ and Conditions~\eqref{eq:dualfeasibility-b} and~\eqref{eq:dualfeasibility-b-admissible} holds. Similarly, for any free point $a\in A\setminus A_\cell$, $y'(a)=y(a)=0$ and Condition~\eqref{eq:dualfeasibility-a-free} holds.
        \item For any point $b\in B_\cell$:
        \begin{itemize}
            \item if $b$ is a free point, then $\kappa_b= y(b)$, since the only path from $s$ to $b$ is an edge from $s$ to $b$ with a weight $y(b)$. In this case, $y'(b)= y(b) -\kappa_b+\kappa=\kappa$.
            \item if $b$ is a boundary-matched point, then there are no edges coming into $b$ in $\mcG_\cell$ and therefore, $y'(b)=y(b)=\distance{b}{\mcC}$ and Condition~\eqref{eq:dualfeasibility-b-admissible} holds.
            \item Otherwise, $b$ is a matched point. In this case, we show that $\kappa\le \kappa_b+s(b)$: Let $P_b$ denote the shortest path from $s$ to $b$, and let $P'$ denote the path obtained by removing $s$ from $P_b$, which has a free endpoint $b'\in B_\cell$. If $\kappa> \kappa_b+s(b)$, then the net-cost of the augmenting path $P'$ is, by Lemma~\ref{lem:slackcost},
            \[\phi(P')=y(b')+s(b)+ \sum_{(a'',b'') \in P}s(a'',b'') = \kappa_b + s(b)< \kappa=\phi(P),\]
            which contradicts the assumption that $P$ is a minimum net-cost augmenting path. 
            Therefore, $\kappa\le \kappa_b+s(b)$ and Condition~\eqref{eq:dualfeasibility-b} holds for $b$ since 
            \begin{equation}
                y'(b) = y(b) + \kappa - \kappa_b \le y(b) + s(b) =\distancetocell{b}{\mcC}. 
            \end{equation}
        \end{itemize}
        \item For any free point $a\in A_\cell^F$, $\kappa_a\ge \kappa$, since otherwise, if $\kappa_a<\kappa$, then the path from the source to $a$ defines an augmenting path whose net-cost is $\kappa_a<\kappa=\phi(P)$, which contradicts the assumption that $P$ is a minimum net-cost augmenting path. Thus, $\kappa_a\ge \kappa$ and the procedure does not update the dual weight of $a$, i.e., $y'(a)=0$ satisfying Condition~\eqref{eq:dualfeasibility-a-free}. 
        
    \end{enumerate}

    \myparagraph{Feasibility of edges.} For any edge $(a,b)\in E$, let $s(a,b)$ denote the slack of $(a,b)$ with respect to $y(\cdot)$. For any edge $(a,b)\in E$:
    \begin{enumerate}
        \item if $a\in A\setminus A_\cell$ and $b\in B\setminus B_\cell$, then $y'(a)=y(a)$ and $y'(b)=y(b)$ and Conditions~\eqref{eq:dualfeasibility-non-matching-time} and~\eqref{eq:dualfeasibility-matching-time} remains satisfied for $(a,b)$.
        
        \item if $a\in A_\cell$ and $b\in B\setminus B_\cell$, then the edge $(a,b)$ is a non-matching edge, $y'(b)=y(b)$, and $y'(a)\ge y(a)$; hence, $y'(b)-y'(a)\le y(b)-y(a)\le \distance{a}{b}$ and Condition~\eqref{eq:dualfeasibility-non-matching-time} is satisfied.

        \item if $a\in A\setminus A_\cell$ and $b\in B_\cell$, then $(a,b)$ is a non-matching edge and as discussed above, $y'(b)\le \distance{b}{\mcC}\le \distance{a}{b}$; hence, $y'(b)-y'(a)\le y'(b)\le \distance{a}{b}$ and Condition~\eqref{eq:dualfeasibility-non-matching-time} is satisfied.

        \item if $a\in A_\cell$ and $b\in B_\cell$: 
        \begin{itemize}
            \item If $(a,b)\in M$ is a matching edge, then $\kappa_b = \kappa_a$ since the only edge directed to $b$ in the residual graph is the zero-slack edge $(a,b)$. Thus, Condition~\eqref{eq:dualfeasibility-matching-time} holds since
            \begin{equation*}
            y'(b) - y'(a) = (y(b)+\kappa-\kappa_b) - (y(a)+\kappa-\kappa_a) = y(b)-y(a) = \distance{a}{b}.
            \end{equation*}
            \item Otherwise, $(a,b)$ is a non-matching edge and $\kappa_a\le \kappa_b + s(a,b)$, and Condition~\eqref{eq:dualfeasibility-non-matching-time} holds since
            \begin{equation*}
                y'(b)-y'(a) = (y(b)+\kappa-\kappa_b) - (y(a)+\kappa-\kappa_a) \le y(b)-y(a) + s(a, b) = \distance{a}{b}.
            \end{equation*}
            Note that if $(a,b)$ is a non-matching edge on the shortest path tree, then $\kappa_a= \kappa_b + s(a,b)$ and $y'(b)-y'(a) = \distance{a}{b}$, i.e., $(a,b)$ is admissible.
        \end{itemize}
    \end{enumerate}

    Finally, we show that the dual weight of each point $v\in A_\cell\cup B_\cell$ is at most $\phi(P)$.
    For each point $b\in B_\cell$ such that $y'(b)> y(b)$, we have $\kappa_b<\kappa$. Suppose $P_b$ denotes the shortest path from the source $s$ to $b$, and let $P'$ be the path obtained by removing $s$ from $P_b$. Let $b'$ denote the free endpoint of $P'$. By the assumption of the lemma, $y(b')=y_\cell\ge y(b)$. Furthermore, since $P'$ is a path on the shortest path tree, $\kappa_b = \kappa_{b'} + \sum_{(u,v)\in P'} s(u,v)$ and since all slacks are non-negative, $\kappa_{b'}\le \kappa_b$; therefore, 
    \[y'(b)=y(b)-\kappa_b+\kappa\le y(b')-\kappa_{b'}+\kappa=\kappa.\]
    
\end{proof}

The following lemma is resulted from applying Lemma~\ref{lemma:fresh_duals} on all cells $\cell\in \mcC$. 

\begin{lemma}\label{lemma:adjust_duals}
    Given a partitioning $\mcC$ and a feasible extended matching $M^\mcC, y(\cdot)$, let $P$ denote a minimum net-cost augmenting path. Suppose $y(b)\le \phi(P)$ for all points $b\in B$ and $y(b)=\max_{b'\in B_\cell}y(b')$ for all cells $\cell\in \mcC$ and all free points $b\in B_\cell$. Then, there exists a set of dual weights $y'(\cdot)$ such that $M^\mcC, y'(\cdot)$ is feasible, $y'(b)\le \phi(P)$ for all $b\in B$, and $y'(b)=\phi(P)$ for all free points $b\in B$.
\end{lemma}

\subsection{Proof of Lemma~\ref{lemma:combination}}

\combination*
\myparagraph{Proof of Lemma~\ref{lemma:combination} (a).}
To prove this lemma, we show that the matching $\hat{M}^\mcC=(M, B^\mcC\setminus B_\cell^\mcC)$ along with the dual weights $y(\cdot)$ satisfy the feasibility conditions~\eqref{eq:dualfeasibility-non-matching-time}--\eqref{eq:dualfeasibility-a-free}. 
First, note that for any edge $(a,b)\in E$ (resp. matching edge $(a,b)\in M$), the dual weights of $a$ and $b$ as well as their distance is unchanged; hence, Condition~\eqref{eq:dualfeasibility-non-matching-time} (resp.~\eqref{eq:dualfeasibility-matching-time}) holds trivially. Similarly, for any free point $a\in A_F$, the dual weight of $a$ remains zero, and for any point $b\in B\setminus B_\cell$ outside of $\cell$, the dual weight of $b$, as well as its distance to the boundaries of the partitioning remains unchanged; therefore, Conditions~\eqref{eq:dualfeasibility-b}--\eqref{eq:dualfeasibility-a-free} hold for all free points of $A$ and all points of $B$ that reside outside of $\cell$. Next, we show that Conditions~\eqref{eq:dualfeasibility-b} and~\eqref{eq:dualfeasibility-b-admissible} are also true for the points of $B_\cell$.

For any point $b\in B_\cell$, let $\cell_b\in\{\cell', \cell''\}$ be the child of $\cell$ containing $b$. From the feasibility of $M^\mcC, y(\cdot)$, we have $y(b)\le \distancetocell{b}{\mcC}=\distancetocell{b}{\cell_b} \le \distancetocell{b}{\cell}=\distancetocell{b}{\mcC'}$ and Condition~\eqref{eq:dualfeasibility-b} holds. For any boundary-matched point $b\in B_\cell \cap (B^\mcC\setminus B_\cell^\mcC)$, since $b$ is matched to the boundaries of $\cell_b$ that are different from the divider $\Gamma_\cell$, then $\distance{b}{\mcC}=\distance{b}{\mcC'}$ and therefore, Condition~\eqref{eq:dualfeasibility-b-admissible} holds for all boundary-matched points of $\hat{M}^\mcC$ inside $\cell$ as well.

\myparagraph{Proof of Lemma~\ref{lemma:combination} (b).} Let $M$ denote the matching of the extended matching $M^\mcC$, and let $M'$ denote the matching constructed in Lemma~\ref{lemma:geometric-matching} inside $\cell$. From Lemma~\ref{lemma:no-cross}, any point $b\in B_\cell$ that is matched in $M$ is matched to a point $a\in A_\cell$ inside $\cell$. Let $M_\cell$ denote the subset of the matching edges of $M$ that lie inside $\cell$. Define $F(M_\cell)$ (resp. $F(M')$) as the set of unmatched points of $B_\cell$ in the matching $M_\cell$ (resp. $M'$). Note that $B_\cell^\mcC\subseteq F(M_\cell)$. Any simple path $P$ in the symmetric difference $M_\cell\oplus M'$ from a free point $b\in B_\cell^\mcC$ is called a (standard) alternating path if $P$ ends at a point $b\in F(M')$ and a (standard) augmenting path if it ends at an unmatched point $a\in A_\cell$ with respect to $M_\cell$.  Let $\mcP_{\mathrm{aug}}$ (resp. $\mcP_{\mathrm{alt}}$) denote the set of (standard) augmenting paths (resp. alternating paths) with an endpoint in $B_\cell^\mcC$ in the symmetric difference $M_\cell\oplus M'$. Note that $|B_\cell^\mcC|=|\mcP_{\mathrm{aug}}| + |\mcP_{\mathrm{alt}}|$. For $\mcP_{\mathrm{alt}}$, since each alternating path in the symmetric difference has one free endpoint in $B_\cell^\mcC$ and the other free endpoint in $F(M')$, $|\mcP_{\mathrm{alt}}|\le |F(M')|=O(n^{4/5})$. 
Next, we show that $|\mcP_{\mathrm{aug}}|=O(n^{4/5})$.
Similar to Section~\ref{sec:primal}, define the net-cost of an augmenting path $P$ as $\phi(P) = \sum_{(a,b)\in P\cap M'}\distance{a}{b} - \sum_{(a,b)\in P\cap M_\cell}\distance{a}{b}$. Then,
\begin{align}
    \sum_{P\in \mcP_{\mathrm{aug}}} \phi(P) &= \sum_{P\in \mcP_{\mathrm{aug}}} \left(\sum_{(a,b)\in P\cap M'}\distance{a}{b} - \sum_{(a,b)\in P\cap M_\cell}\distance{a}{b} \right)\nonumber\\ &\le \sum_{P\in \mcP_{\mathrm{aug}}} \left(\sum_{(a,b)\in P\cap M'}\distance{a}{b}\right)\le w(M').\label{eq:net-cost-sum-1}
\end{align}
For each path $P\in \mcP_{\mathrm{aug}}$, let $b_P\in B_\cell$ and $a_P\in A_\cell$ denote the two end-points of $P$. Define $B_{\mathrm{aug}}:=\{b_P:P\in \mcP_{\mathrm{aug}}\}$.  Using Lemma~\ref{lem:slackcost} and Equation~\eqref{eq:net-cost-sum-1}, 
\begin{align}
    \sum_{b\in B_{\mathrm{aug}}} y(b) = \sum_{P\in \mcP_{\mathrm{aug}}} y(b_P)\le \sum_{P\in \mcP_{\mathrm{aug}}} \phi(P) \le w(M') .\label{eq:net-cost-sum-2}
\end{align}

Define $\alpha:=\ell_\cell n^{-1/5}$. Each free point $b\in B_{\mathrm{aug}}$ is called a \emph{close} (resp. \emph{far}) point if the distance of $b$ to the divider of $\cell$ is at most (resp. more than) $\alpha$. Let $B_{\mathrm{aug}}^{\mathrm{close}}$ (resp. $B_{\mathrm{aug}}^{\mathrm{far}}$) denote the set of all close (resp. far) points of $B_{\mathrm{aug}}$. By Lemma~\ref{lemma:margin},
\begin{equation}\label{eq:close-points}
    |B_{\mathrm{aug}}^{\mathrm{close}}| = O(n_\cell n^{-1/5}) = O(n^{4/5}).
\end{equation}
For each far point $b\in B_{\mathrm{aug}}^{\mathrm{far}}$, since $b$ is mapped to the divider $\divider\cell$, $y(b)=\distance{b}{\divider\cell}\ge \alpha$; therefore, 
\begin{equation}\label{eq:net-cost-3}
    \sum_{b\in B_{\mathrm{aug}}} y(b) \ge \sum_{b\in B_{\mathrm{aug}}^{\mathrm{far}}} y(b) \ge \alpha\times |B_{\mathrm{aug}}^{\mathrm{far}}|.
\end{equation}
Combining Equations~\eqref{eq:net-cost-sum-2} and~\eqref{eq:net-cost-3},
\begin{equation}\label{eq:far-points}
    |B_{\mathrm{aug}}^{\mathrm{far}}| \le \frac{\sum_{b\in B_{\mathrm{aug}}} y(b)}{\alpha} \le \frac{w(M')}{\alpha} = O(n^{4/5}).
\end{equation}
Combining with Equation~\eqref{eq:close-points},
\begin{align*}
    |B_\cell^\mcC| &\le |\mcP_{\mathrm{aug}}| + |\mcP_{\mathrm{alt}}| \le |B_{\mathrm{aug}}^{\mathrm{close}}| + |B_{\mathrm{aug}}^{\mathrm{far}}| + |\mcP_{\mathrm{alt}}| = O(n^{4/5}).
\end{align*}

\section{Missing Proofs of Section~\ref{sec:k-seq}}\label{sec:k-seq-appendix}

In this section, we provide the missing proofs of lemmas and claims made in Section~\ref{sec:k-seq}.

\begin{lemma}\label{lemma:augment}
    For any partitioning $\mcC$ of the root square $\cell^*$ of $\mcH$, any feasible extended matching $M^\mcC=(M, B^\mcC), y(\cdot)$ on $\mcG_\sigma$, and any admissible alternating or augmenting path $P$, the matching obtained after updating $M^\mcC$ along $P$ remains feasible.
\end{lemma}
\begin{proof}
    For any edge $(a,b)\in P$, due to the admissibility of the edge, $y(b)-y(a)=\distance{a}{b}$. If $(a,b)\notin M$ prior to augmentation, it is a matching edge after the augmentation and Condition~\eqref{eq:dualfeasibility-matching-time} holds for $(a,b)$. Otherwise, $(a,b)\in M$ prior to augmentation, it is a non-matching edge after the augmentation, and Condition~\eqref{eq:dualfeasibility-non-matching-time} holds for $(a,b)$. Note that all dual weights remain unchanged and consequently, Conditions~\eqref{eq:dualfeasibility-b}--\eqref{eq:dualfeasibility-a-free} remain satisfied.
\end{proof}


\subsection{Correctness of the Extended Hungarian Search Procedure}

\hungarianProp*

Let $y(\cdot)$ (resp. $y'(\cdot)$) denote the dual weights of the points after (resp. before) the update duals step of the extended Hungarian search procedure. To prove this lemma, we first show that after the update duals step, the extended matching $M^\mcC, y(\cdot)$ is feasible, $y(v)\le \ymax$ for all $v\in A_\cell\cup B_\cell$, and $y(b_f)=\ymax$ for all free points $b_f\in B_\cell$. We then show that the augmenting path $P$ computed by the procedure is admissible and conclude that $P$ is a minimum net-cost augmenting path inside $\cell$. We then use Lemma~\ref{lemma:augment} to show that after augmentation, the extended matching $M^\mcC, y(\cdot)$ remains feasible. Finally, we show that the updated key of $\cell$ is the smallest net-cost of augmenting paths inside $\cell$ and is at least $\ymax$.

    \myparagraph{Feasibility of points.}
    \begin{enumerate}
        \item For any point $b\in B$:
        \begin{itemize}
            \item if $b\in B\setminus B_\cell$ is outside of $\cell$, then $y(b)=y'(b)$; therefore, Conditions~\eqref{eq:dualfeasibility-b} and~\eqref{eq:dualfeasibility-b-admissible} remains satisfied.
            \item Otherwise, $b\in B_\cell$ and by definition, $\kappa\le \kappa_b+s(b)$. Therefore, Condition~\eqref{eq:dualfeasibility-b} holds for $b$ since 
            \begin{equation}\label{eq:global-hung-2}
                y(b) = y'(b) + \kappa - \kappa_b \le y'(b) + s(b) =\distancetocell{b}{\cell}. 
            \end{equation}
        \end{itemize}
        \item For any free point $a\in A_F$:
        \begin{itemize}
            \item if $a\in A\setminus A_\cell$ is a free point outside of $\cell$, then $y(a)=y'(a)=0$ and Condition~\eqref{eq:dualfeasibility-a-free} remains true. 
            \item Otherwise, $a\in A_\cell$ and by definition, $\kappa_a\ge \kappa$; hence, the procedure does not update the dual weight of $a$, i.e., the dual weight of $a$ remains $0$, satisfying Condition~\eqref{eq:dualfeasibility-a-free}. 
        \end{itemize}
        
    \end{enumerate}

    \myparagraph{Feasibility of edges.} For any edge $(a,b)\in E$, let $s(a,b)$ denote the slack of $(a,b)$ before updating the dual weights. 
    \begin{enumerate}
        \item if $a\in A\setminus A_\cell$ and $b\in B\setminus B_\cell$, then $y(a)=y'(a)$ and $y(b)=y'(b)$ and Conditions~\eqref{eq:dualfeasibility-non-matching-time} and~\eqref{eq:dualfeasibility-matching-time} remains satisfied for $(a,b)$.
        
        \item if $a\in A_\cell$ and $b\in B\setminus B_\cell$, then $y(b)=y'(b)$ and $y(a)\ge y'(a)$, since the extended Hungarian search procedure does not decrease the dual weight of any point inside $\cell$; hence, $y(b)-y(a)\le y'(b)-y'(a)\le \distance{a}{b}$ and Condition~\eqref{eq:dualfeasibility-non-matching-time} is satisfied. (Note that by Lemma~\ref{lemma:no-cross}, the edge $(a,b)$ is a non-matching edge).

        \item if $a\in A\setminus A_\cell$ and $b\in B_\cell$, then $y(b)\le \distance{b}{\mcC}\le \distance{a}{b}$; hence, $y(b)-y(a)\le y(b)\le \distance{a}{b}$ and Condition~\eqref{eq:dualfeasibility-non-matching-time} is satisfied. (Note that by Lemma~\ref{lemma:no-cross}, the edge $(a,b)$ is a non-matching edge).

        \item if $a\in A_\cell$ and $b\in B_\cell$: 
        \begin{itemize}
            \item If $(a,b)\in M$ is a matching edge, then $\kappa_b = \kappa_a$ since the only edge directed to $b$ in the residual graph is the zero-slack edge $(a,b)$. Thus, Condition~\eqref{eq:dualfeasibility-matching-time} holds since
            \begin{equation*}
            y(b) - y(a) = (y'(b)+\kappa-\kappa_b) - (y'(a)+\kappa-\kappa_a) = y'(b)-y'(a) = \distance{a}{b}.
            \end{equation*}
            \item Otherwise, $(a,b)$ is a non-matching edge and $\kappa_a\le \kappa_b + s(b,a)$; therefore, Condition~\eqref{eq:dualfeasibility-non-matching-time} holds since
            \begin{equation}\label{eq:global-hung-1}
                y(b)-y(a) = (y'(b)+\kappa-\kappa_b) - (y'(a)+\kappa-\kappa_a) \le y'(b)-y'(a) + s(b,a) = \distance{a}{b}.
            \end{equation}
        \end{itemize}
    \end{enumerate}

\myparagraph{Maximum dual weight.} Let $y_{\cell}:=\max_{b\in B_\cell}y'(b)$. By invariant (I2) prior to the extended Hungarian search procedure, for all free points $b\in B_\cell$, $y'(b)=y_{\cell}$. By the construction of the residual graph, $\kappa_b=y'(b)=y_\cell$ for all free points $b\in B_\cell$ and $\kappa_b\ge y_\cell$ for all points $b\in B_\cell$. Therefore, for any free point $b\in B_\cell$, 
\[y(b) = y'(b)+\kappa-\kappa_{b} = \kappa=\ymax,\]
where the last equality holds since $\kappa$ is the net-cost of the minimum net-cost path inside $\cell$, which is $\ymax$.
Furthermore, for any point $b'\in B_\cell$,
\[y(b) = y'(b)+\kappa-\kappa_{b} \le y'(b)+\kappa-y_\cell \le \kappa=\ymax.\]
Note that for any free point $a\in A_\cell$, $y(a)=0$ and for any matched point $a\in A_\cell$, if $a$ is matched to a point $b\in B_\cell$, then by Condition~\eqref{eq:dualfeasibility-matching-time}, $y(a)=y(b)-\distance{a}{b}\le y(b)\le \ymax$. 
Hence, we conclude $y(v)\le \ymax$ for all points $v\in A_\cell\cup B_\cell$ and $y(b_f)=\ymax$ for all free points $b_f\in B_\cell$ after the dual updates.

\myparagraph{Net-cost of $P$.} To prove that $P$ is a minimum net-cost augmenting path, we show that $P$ is an admissible augmenting path. Consequently, if $b\in B_\cell$ is the free endpoint of $P$, by Corollary~\ref{cor:slackcost}, $\phi(P)=y(b)=\ymax$. Note that for all other augmenting paths $P'$ from a free point $b'\in B_\cell$, from Lemmas~\ref{lem:slackcost}, $\phi(P')\ge y(b')=\ymax$, and therefore, $P$ would be a minimum net-cost augmenting path inside $\cell$.

For each edge $(u,v)\in P$, $\kappa_v = \kappa_u + s(u,v)$, since $(u,v)$ is an edge of the shortest path tree of the residual graph. Plugging into Equation~\eqref{eq:global-hung-1}, for each non-matching edge $(b,a)$, $y(b)-y(a)=\distance{a}{b}$ and the edge $(b,a)$ is admissible. Furthermore, if $\kappa=\kappa_b+s(b)$ for a point $b\in B$, then by Equation~\eqref{eq:global-hung-2}, $y(b)=\distancetocell{b}{\cell}$ and the point $b$ would be a zero-slack point. Therefore, the path $P$ computed by the algorithm would be an admissible path from a free point $b\in B_\cell$ to either a free point $a\in A_\cell$ or a zero-slack point $b'\in B_\cell$, i.e., $P$ is an admissible augmenting path. 

From Lemma~\ref{lemma:augment}, the extended matching $M^\mcC, y(\cdot)$ obtained after augmenting the matching $M^\mcC$ along $P$ remains feasible.

\myparagraph{Updated key.} Note that the extended matching $M^\mcC, y(\cdot)$ after the augmentation step is feasible, $y(v)\le \ymax$ for all vertices $v\in A_\cell\cup B_\cell$, and $y(b_f)=\ymax$ for all free points $b_f\in B_\cell$. Let $\kappa_v$, for each $v\in A_\cell\cup B_\cell$, denote the distances computed in the update key step of the extended Hungarian search procedure. By Lemma~\ref{lem:slackcost}, for any augmenting path $P$ from a free point $b\in B_\cell$ to a free point $a\in A_\cell$, the net-cost of $P$ is 
\[\phi(P)=y(b) +\sum_{(a'',b'')\in P}s(a'',b'') \ge \kappa_a,\] where the last inequality holds by the construction of the residual graph and the definition of $\kappa_a$. In this case, the equality happens when $P$ is a shortest path from $s$ to $a$. Similarly, for any augmenting path $P$ from a free point $b\in B_\cell$ to a point $b'\in B_\cell$, the net-cost of $P$ is 
\[\phi(P)=y(b) + s(b') +\sum_{(a'',b'')\in P}s(a'',b'') \ge \kappa_{b'} + s(b'),\]
and we get an equality if $P$ is a shortest path from the source vertex $s$ to $b'$.
Therefore, the updated key of $\cell$, i.e., $\kappa_\cell=\min\{\min_{a\in A_\cell^F}\kappa_a, \min_{b\in B_\cell}\kappa_b+s(b)\}$ correctly computes the net-cost of the minimum net-cost augmenting path inside $\cell$. Finally, note that for any augmenting path $P$ from a free point $b\in B_\cell$, by Lemma~\ref{lem:slackcost}, $\phi(P)\ge y(b)=\ymax$ and the key of $\cell$ would be at least $\ymax$.

\subsection{Correctness of the Merge Procedure}\label{subsec:merge_correctness}

\mergeProp*
Let $\hat{M}^\mcC=(M, \hat{B}^\mcC), \hat{y}(\cdot)$ denote the feasible extended matching maintained by our algorithm before the execution of the merge procedure. Let $B^\mcC_\cell$ denote the subset of points in $\hat{B}^\mcC$ that are matched to the divider $\Gamma_\cell$ of $\cell$, and let $M^\mcC=(M, B^\mcC = \hat{B}^\mcC\setminus B^\mcC_\cell)$ be the extended matching after adding the boundary-matched points in $B_\cell^\mcC$ to the free points. By Lemma~\ref{lemma:combination}, the matching $M^\mcC, \hat{y}(\cdot)$ is feasible. Next, we show that, given a feasible extended matching $M^\mcC, y'(\cdot)$, after one iteration of the while-loop in the merge step, the matching remains feasible and the dual weights of points in $\cell$ are at most $\ymax$.

Let $y(\cdot)$ (resp. $y'(\cdot)$) denote the dual weights of the points after (resp. before) the execution of one iteration of the merge step. We first show that $M^\mcC, y(\cdot)$ is a feasible extended matching where $y(v)\le \ymax$ for all $v\in A_\cell\cup B_\cell$. We then show that the path $P$ is admissible and use Lemma~\ref{lemma:augment} to show that after augmentation, the extended matching $M^\mcC, y(\cdot)$ remains feasible.

\myparagraph{Feasibility of points.}
    \begin{enumerate}
        \item For any point $b\in B$:
        \begin{itemize}
            \item if $b\in B\setminus B_\cell$ is outside of $\cell$, then $y(b)=y'(b)$; therefore, Conditions~\eqref{eq:dualfeasibility-b} and~\eqref{eq:dualfeasibility-b-admissible} remains satisfied.
            \item Otherwise, $b\in B_\cell$ and $\kappa\le \kappa_b+s(b)$. Therefore, Condition~\eqref{eq:dualfeasibility-b} holds for $b$ since 
            \begin{equation}\label{eq:local-hung-2}
                y(b) = y'(b) + \kappa - \kappa_b \le y'(b) + s(b) =\distancetocell{b}{\cell}. 
            \end{equation}
        \end{itemize}
        \item For any free point $a\in A_F$:
        \begin{itemize}
            \item if $a\in A\setminus A_\cell$ is a free point outside of $\cell$, then $y(a)=y'(a)=0$ and Condition~\eqref{eq:dualfeasibility-a-free} remains true. 
            \item Otherwise, $a\in A_\cell$ and $\kappa_a\ge \kappa$; therefore, the procedure does not update the dual weight of $a$, i.e., the dual weight of $a$ remains $0$, satisfying Condition~\eqref{eq:dualfeasibility-a-free}. 
        \end{itemize}
        
    \end{enumerate}

    \myparagraph{Feasibility of edges.} For any edge $(a,b)\in E$, let $s(a,b)$ denote the slack of $(a,b)$ with respect to $y'(\cdot)$. 
    \begin{enumerate}
        \item if $a\in A\setminus A_\cell$ and $b\in B\setminus B_\cell$, then $y(a)=y'(a)$ and $y(b)=y'(b)$ and therefore, Conditions~\eqref{eq:dualfeasibility-non-matching-time} and~\eqref{eq:dualfeasibility-matching-time} remains satisfied for $(a,b)$.
        
        \item if $a\in A_\cell$ and $b\in B\setminus B_\cell$, then $y(b)=y'(b)$ and $y(a)\ge y'(a)$, since the merge procedure does not decrease the dual weight of any point inside $\cell$; hence, $y(b)-y(a)\le y'(b)-y'(a)\le \distance{a}{b}$ and Condition~\eqref{eq:dualfeasibility-non-matching-time} is satisfied (Note that by Lemma~\ref{lemma:no-cross}, the edge $(a,b)$ is a non-matching edge).

        \item if $a\in A\setminus A_\cell$ and $b\in B_\cell$, then $y(b)\le \distance{b}{\mcC}\le \distance{a}{b}$; hence, $y(b)-y(a)\le y(b)\le \distance{a}{b}$ and Conditions~\eqref{eq:dualfeasibility-non-matching-time} is satisfied (Note that by Lemma~\ref{lemma:no-cross}, the edge $(a,b)$ is a non-matching edge).

        \item if $a\in  A_\cell$ and $b\in B_\cell$: 
        \begin{itemize}
            \item If $(a,b)\in M$ is a matching edge, then $\kappa_b = \kappa_a$ since the only edge directed to $b$ in the residual graph is the zero-slack edge $(a,b)$. Thus, Condition~\eqref{eq:dualfeasibility-matching-time} holds since
            \begin{equation*}
            y(b) - y(a) = (y'(b)+\kappa-\kappa_b) - (y'(a)+\kappa-\kappa_a) = y'(b)-y'(a) = \distance{a}{b}.
            \end{equation*}
            \item Otherwise, $(a,b)$ is a non-matching edge and since there is a directed edge from $b$ to $a$, $\kappa_a\le \kappa_b + s(b,a)$, and Condition~\eqref{eq:dualfeasibility-non-matching-time} holds since
            \begin{equation}\label{eq:local-hung-1}
                y(b)-y(a) = (y'(b)+\kappa-\kappa_b) - (y'(a)+\kappa-\kappa_a) \le y'(b)-y'(a) + s(b,a) = \distance{a}{b}.
            \end{equation}
        \end{itemize}
    \end{enumerate}

    \myparagraph{Maximum dual weight.} Next, we show that $y(v)\le \ymax$ for all $v\in A_\cell\cup B_\cell$. Note that for any point $b\in B_\cell$, by the definition of $\kappa$, we have $\kappa\le \kappa_b+\ymax-y'(b)$. Therefore, if $\kappa_b<\kappa$, then $y(b) = y'(b)-\kappa_b+\kappa\le \ymax$. Otherwise, $\kappa_b\ge \kappa$ and $y(b)=y'(b)\le \ymax$. Furthermore, for all free points $a\in A_\cell$, by Condition~\eqref{eq:dualfeasibility-a-free}, $y(a)=0$ and for all matched points $a\in A_\cell$, if $a$ is matched to a point $b\in B_\cell$, then $y(a)=y(b)-\distance{a}{b}\le y(b)\le \ymax$.

    \myparagraph{Net-cost of $P$.} Next, we show that the path $P$ computed in an iteration of the merge step is either (i) an admissible augmenting path, or (ii) an admissible alternating path from a free point $b\in B_\cell$ to a point $b'\in B_\cell$ with $y(b')=\ymax$. Then, from Lemma~\ref{lemma:augment}, the extended matching obtained after updating $M^\mcC, y(\cdot)$ along $P$ is feasible.
    
    For each non-matching edge $(b,a)\in P$, since $(b,a)$ is in the Dijkstra's shortest path tree, $\kappa_a = \kappa_b + s(b,a)$. Plugging into Equation~\eqref{eq:local-hung-1}, for each non-matching edge $(b,a)$, $y(b)-y(a)=\distance{a}{b}$, and the edge $(b,a)$ is admissible (i.e., all edges of $P$ are admissible).
    \begin{itemize}
        \item If $P$ is a path that ends at a free point $a\in A_\cell$ ($\kappa$ is determined by the first term in the RHS of Equation~\eqref{eq:local-kappa}), then $P$ is an admissible augmenting path, 
        \item otherwise, if $P$ is a path that ends at a point $b'\in B_\cell$ with $\kappa=\kappa_{b'}+s(b')$ ($\kappa$ is determined by the second term in the RHS of Equation~\eqref{eq:local-kappa}), then from Equation~\eqref{eq:local-hung-2}, $y(b') = y'(b')+s(b') = \distance{b'}{\mcC}$, and $P$ is an admissible augmenting path, and
        \item otherwise, $P$ is a path that ends at a point $b'\in B_\cell$ with $\kappa=\kappa_{b'}+\ymax-y'(b')$, and $P$ is an admissible alternating path.
    \end{itemize}
    Therefore, $P$ is an admissible alternating or augmenting path and the matching obtained by updating $M^\mcC$ along $P$ remains feasible (Lemma~\ref{lemma:augment}). 
    
    \myparagraph{Termination.} Let $P$ be a path from a free point $b\in B_\cell$ to a point $u\in \freeofcell{A}{\cell}\cup B_\cell$.
    \begin{itemize}
        \item If $u\in\freeofcell{A}{\cell}$, then $P$ is an admissible augmenting path and $P$ is in case (i).
        \item Otherwise, if $u\in B_\cell$ and $\kappa=\kappa_u+s(u)$, then $P$ is an admissible augmenting path and $P$ is in case (ii).
        \item Otherwise, $u\in B_\cell$ and $\kappa=\kappa_u+\ymax - y'(u)$. In this case, $P$ is an admissible alternating path from $b$ to a free point $b'\in B_\cell$ with $y'(b') < \ymax$ and $y(b')=\ymax$. 
    \end{itemize}
    In either case, the number of free points $b\in B_\cell$ with $y(b)<\ymax$ reduces by one, and the merge step terminates.

    \myparagraph{Dual weight of free points.} Note that the while-loop terminates when there are no free points $b_f\in B_\cell$ with $y(b_f)<\ymax$, whereas, as discussed above, all points will have a dual weight at most $\ymax$, i.e., the dual weight of each free point $b_f\in B_\cell$ after the termination of the while-loop is $\ymax$.

\myparagraph{Updated key.} Finally, we show that the updated key of $\cell$ denotes the net-cost of the minimum net-cost augmenting path inside $\cell$. Note that the extended matching $M^\mcC, y(\cdot)$ after the termination of the while-loop is feasible, $y(v)\le \ymax$ for all vertices $v\in A_\cell\cup B_\cell$, and $y(b_f)=\ymax$ for all free points $b_f\in B_\cell$. Let $\kappa_v$, for each $v\in A_\cell\cup B_\cell$, denote the distances computed in the update key step of the extended Hungarian search procedure. By Lemma~\ref{lem:slackcost}, for any augmenting path $P$ from a free point $b\in B_\cell$ to a free point $a\in A_\cell$, the net-cost of $P$ is 
\[\phi(P)=y(b) +\sum_{(a'',b'')\in P}s(a'',b'') \ge \kappa_a,\] where the last inequality holds by the construction of the residual graph and we get an equality of $P$ is the shortest path from $s$ to $a$. Similarly, for any augmenting path $P$ from a free point $b\in B_\cell$ to a point $b'\in B_\cell$, the net-cost of $P$ is 
\[\phi(P)=y(b) + s(b') +\sum_{(a'',b'')\in P}s(a'',b'') \ge \kappa_{b'} + s(b'),\]
and we get an equality if $P$ is a shortest path from $s$ to $b'$.
Therefore, the updated key of $\cell$, i.e., $\kappa_\cell=\min\{\min_{a\in A_\cell^F}\kappa_a, \min_{b\in B_\cell}\kappa_b+s(b)\}$ correctly computes the net-cost of the minimum net-cost augmenting path inside $\cell$. Finally, note that for any augmenting path $P$ from a free point $b\in B_\cell$, by Lemma~\ref{lem:slackcost}, $\phi(P)\ge y(b)=\ymax$ and the key of $\cell$ would be at least $\ymax$.

\subsection{Runtime Analysis of the Merge Step} 
For any non-leaf cell $\cell$ with $\cell'$ and $\cell''$ as children, the merge step at $\cell$ first increases the dual weights of the free points inside $\cell'$ and $\cell''$ in $\tilde{O}(n_\cell\Phi(n_\cell))$ time (Lemma~\ref{lemma:fresh_duals}). For the feasible extended matching $M^\mcC=(M, B^\mcC), y(\cdot)$ after this initial step, let $B^\mcC_\cell$ denote the subset of the boundary-matched points in $B^\mcC$ that are matched to $\Gamma_\cell$. As discussed in Section~\ref{subsec:merge_correctness}, each iteration of the while-loop in the merge step reduces the number of free points with a dual weight less than $\ymax$ by one; therefore, the total number of executions of the while-loop in the merge step is at most $|B^\mcC_\cell|$.

\begin{restatable}{lemma}{localProcessInactive}\label{lemma:local-inactive}
    For any cell $\cell$, the number of iterations of the merge step on $\cell$ is $O(|B^\mcC_\cell|)$.
\end{restatable}

Each iteration requires the computation of the distance $\kappa_v$ for each point $v\in A_\cell\cup B_\cell$, which takes $\tilde{O}(n_\cell^2)$ time. As shown in Section~\ref{sec:hungarian_nk}, the efficiency of this computation can be improved to $\tilde{O}(n_\cell\Phi(n_\cell))$ using a dynamic weighted nearest neighbor data structure with a query/update time of $\Phi(n_\cell)$. Furthermore, by Lemma~\ref{lemma:combination}, $|B_{\cell}^\mcC|=O(n^{4/5})$. Computing the updated key of $\cell$ also requires the computation of the distance of each point from the source in the residual graph, which also takes $\tilde{O}(n_\cell\Phi(n_\cell))$ time. Therefore, the total execution time of the merge step would be $\tilde{O}(n^{4/5}n_\cell\Phi(n_\cell))$.

\subsection{Runtime Analysis of the Extended Hungarian Search Procedure}\label{appendix:hung_time}
Recall that the extended Hungarian search procedure iteratively picks the cell with the minimum key from $\pq$ to be processed.
For each leaf cell $\cell$ in $\mcH$, since $\cell$ contains the points corresponding to a single request and contains only one point of $B$, the procedure picks $\cell$ at most once, at which it matches the point $b\in B_\cell$ to the boundaries of $\cell$. Therefore, the total time of the extended Hungarian search procedure on all leaf cells of $\mcH$ is $O(n)$. 

Next, we show that for any non-leaf cell $\cell$ of $\mcH$, our algorithm executes the search procedure on $\cell$ at most $O(n^{4/5})$ times. Since each execution of the procedure takes $\tilde{O}(n_{\cell}\Phi(n))$ time and each point participates in $O(\log (n\Delta))$ cells, the total running time of the extended Hungarian search procedure on all cells of $\mcH$ would be $O(n^{4/5}\sum_{\cell\in\mcH}n_\cell\Phi(n_\cell)) = O(n^{9/5}\Phi(n_\cell)\log(n\Delta))$, as claimed.

For any execution of the search procedure on $\cell$, since $\cell$ has the minimum key in $\pq$, the value $\ymax$ represents the net-cost of the minimum net-cost augmenting path inside $\cell$.
Recall that for any non-leaf cell $\cell$, we categorized the selections of $\cell$ by the search procedure as low-net-cost if the value of $\ymax$ in that iteration is at most $\ell_\cell n^{-1/5}$ and high-net-cost otherwise. 
To bound the high-net-cost selections, we show that as soon as the value $\ymax$ exceeds $\ell_\cell n^{-1/5}$, the number of free points inside $\cell$ cannot be more than $O(n^{4/5})$, and therefore, the number of high-net-cost selections of $\cell$ is $O(n^{4/5})$.

\highMax*
\begin{proof}
    Let $M'$ be the matching inside $\cell$ as constructed in Lemma~\ref{lemma:geometric-matching}, and let $M$ denote the matching of the extended matching $M^\mcC=(M, B^\mcC)$.
    Define $B^\mcC_\cell$ as the set of free points of $B_\cell$ with respect to $M^\mcC,y(\cdot)$. Let $M_\cell$ denote the subset of matching edges of $M$ that lie inside $\cell$, and let $\mcP_{\mathrm{aug}}$ (resp. $\mcP_{\mathrm{alt}}$) denote the set of (standard) augmenting (resp. alternating) paths in the symmetric difference $M_\cell\oplus M'$ with one endpoint in the set $B^\mcC_\cell$. Note that each path in $\mcP_{\mathrm{alt}}$ has one endpoint that is free in $M'$ and therefore, $|\mcP_{\mathrm{alt}}|\le |F(M')|=O(n^{4/5})$; here, $F(M')$ denotes the set of free points of $M'$. Next, we show that $|\mcP_{\mathrm{aug}}|=O(n^{4/5})$.

    From the definition of the net-cost of an augmenting path,
    \begin{align}
        \sum_{P\in \mcP_{\mathrm{aug}}} \phi(P) &= \sum_{P\in \mcP_{\mathrm{aug}}} \left(\sum_{(a,b)\in P\cap M'}\distance{a}{b} - \sum_{(a,b)\in P\cap M_\cell}\distance{a}{b} \right)\nonumber\\ &\le \sum_{P\in \mcP_{\mathrm{aug}}} \left(\sum_{(a,b)\in P\cap M'}\distance{a}{b}\right)\le w(M').\label{eq:net-cost-sum-1-1}
    \end{align}
    For each path $P\in \mcP_{\mathrm{aug}}$, let $b_P\in B_\cell$ and $a_P\in A_\cell$ denote the two end-points of $P$. Define $B_{\mathrm{aug}}:=\{b_P:P\in \mcP_{\mathrm{aug}}\}$ as the set of free endpoints of the paths in $\mcP_{\mathrm{aug}}$. Using Lemma~\ref{lem:slackcost} and Equation~\eqref{eq:net-cost-sum-1-1}, 
    \begin{align}
        \sum_{b\in B_{\mathrm{aug}}} y(b) = \sum_{P\in \mcP_{\mathrm{aug}}} y(b_P)\le \sum_{P\in \mcP_{\mathrm{aug}}} \phi(P) \le w(M') .\label{eq:net-cost-sum-2-1}
    \end{align}
    From invariant (I2), the dual weight of all free points of $B_\cell$ equals $\ymax>\ell_\cell n^{-1/5}$. Therefore,
    \begin{equation*}\label{eq:net-cost-sum-2-2}
        |B_{\mathrm{aug}}| = \frac{\sum_{b\in B_{\mathrm{aug}}} y(b)}{\ymax} \le \frac{w(M')}{\ell_\cell n^{-1/5}}=O(n_\cell^{3/5}n^{1/5}) = O(n^{4/5}).
    \end{equation*}
    Combining the two bounds, the total number of free points with respect to $M^\mcC, y(\cdot)$ is $|\mcP_{\mathrm{alt}}|+|\mcP_{\mathrm{aug}}|=O(n^{4/5})$.
\end{proof}

Define $\mcC_\cell$ as the set of all cells of $\mcH$, including $\cell$ itself, that are processed by the merge step of our algorithm while $\cell\in\mcC$ and $\ymax\le \ell_\cell n^{-1/5}$. For any cell $\cell'\in\mcC_\cell$, let $\mcB_{\cell'}$ denote the set of points of $B$ that are matched to the divider $\Gamma_{\cell'}$ at the beginning of the merge step at $\cell'$. During the execution of our algorithm, before processing $\cell'$ by the merge step, there are $k$ free points across all cells is $\mcC$. After combining the two children of $\cell'$ (and removing the divider of $\cell'$), the number of free points across all cells in the partitioning is now increased to $k + |\mcB_{\cell'}|$. Each iteration of the merge step either (i) finds an admissible augmenting path, which reduces the number of free points by one, or (ii) finds an admissible alternating path to a point $b\in B_{\cell'}$ with $y(b)=\ymax$, which does not change the number of free points. Therefore, after the merge step at $\cell'$, the number of free points across all cells in the partitioning is at most $k + |\mcB_{\cell'}|$. Our algorithm then iteratively executes the search procedure to reduce the number of free points across all cells to $k$. Hence, for each cell $\cell'\in\mcC_\cell$, the merge step at $\cell'$ might lead to the execution of a low-net-cost extended Hungarian search procedure on $\cell$ at most $|\mcB_{\cell'}|$ times.

\begin{restatable}{lemma}{iterationsOfGlobal}\label{lemma:global-iterations}
    For any cell $\cell$, the number of low-net-cost executions of the extended Hungarian search procedure on $\cell$ is at most $\sum_{\cell'\in\mcC_\cell} |\mcB_{\cell'}|$.
\end{restatable}

Recall that in a low-net-cost selection of $\cell$, the value $\ymax\le \ell_\cell n^{-1/5}$. Using invariant (I2), the dual weight of all points of $B$ is at most $\ymax$, and therefore, for any cell $\cell'\in\mcC_{\cell}$ and each point $b\in \mcB_{\cell'}$, $y(b)\le \ell_\cell n^{-1/5}$, i.e., $b$ is a boundary-matched point that is matched to the divider $\divider{\cell'}$ and has a dual weight at most $\ell_\cell n^{-1/5}$. Therefore, $\distance{b}{\divider{\cell'}} = y(b) \le \ell_\cell n^{-1/5}$ and all points in $\mcB_{\cell'}$ are at a distance at most $\ell_\cell n^{-1/5}$ from the divider $\divider{\cell'}$. From Lemma~\ref{lemma:ratio}, $\ell_{\cell}\le 3\ell_{\cell'}$. Therefore, using Lemma~\ref{lemma:margin}, the number of points of $B_{\cell'}$ at a distance at most $\lambda\ell_{\cell'}\ge \ell_\cell n^{-1/5}$ to the divider $\divider{\cell'}$ is $O(n_{\cell'}n^{-1/5})$; hence, $|\mcB_{\cell'}| = O(n_{\cell'}n^{-1/5})$.

\ratio*
\begin{proof}
    To prove this lemma, we first provide two useful relations between the sum of side-lengths $p_{\hat{\cell}}$ and the largest side-length $\ell_{\hat{\cell}}$ of any cell $\hat{\cell}$ of $\mcH$. Let $\ell_x$ (resp. $\ell_y$) denote the width (resp. height) of $\hat\cell$. Since the aspect ratio of $\hat{\cell}$ is at most $3$, i.e., $\min\{\ell_x, \ell_y\} \ge \frac{1}{3}\ell_{\hat{\cell}} = \frac{1}{3}\max\{\ell_x, \ell_y\}$, 
    \begin{equation*}\label{eq:ratio-edge}
        \frac{4}{3}\ell_{\hat{\cell}} \le \ell_x + \ell_y = p_{\hat{\cell}} =\ell_x + \ell_y \le 2\ell_{\hat{\cell}}.
    \end{equation*}
    Furthermore, for the smaller child $\hat{\cell}_1$ of $\hat{\cell}$, if the cell is divided on the $x$ axis, then the width of $\hat{\cell}_1$ is within  $\frac{1}{3}\ell_{\hat\cell}$ and $\frac{1}{2}\ell_{\hat\cell}$; therefore,
    \begin{equation*}
        \frac{4}{3}p_{\hat{\cell}_1}\le p_{\hat\cell} \le 2 p_{\hat{\cell}_1}.
    \end{equation*}
    For the cell $\cell$ (resp. $\cell'$), suppose $\cell_{\min}$ (resp. $\cell'_{\min}$) denote the smaller child of $\cell$ (resp. $\cell'$). Since our algorithm picked $\cell_{\min}$ for being merged as the smallest cell before $\cell'_{\min}$, 
    \begin{equation}\label{eq:ratio-1}
        \frac{4}{3}\ell_\cell\le p_\cell\le 2p_{\cell_{\min}}\le 2p_{\cell'_{\min}}\le \frac{3}{2}p_{\cell'}\le 3\ell_{\cell'}.
    \end{equation}
    Similarly, since our algorithm picked $\cell'_{\min}$ as the smallest cell to be processed rather than $\cell$,
    \begin{equation}\label{eq:ratio-2}
        \frac{4}{3}\ell_{\cell'}\le p_{\cell'}\le 2p_{\cell'_{\min}}\le 2p_{\cell}\le 4\ell_\cell.
    \end{equation}
    Combining Equations~\eqref{eq:ratio-1} and~\eqref{eq:ratio-2},
    \begin{equation*}
        \frac{1}{3} \ell_{\cell'} \le \ell_\cell \le \frac{9}{4} \ell_{\cell'}.
    \end{equation*}
\end{proof}


We next show that $\sum_{\cell'\in\mcC_{\cell}}n_{\cell'}=O(n)$ and conclude that $\sum_{\cell'\in\mcC_\cell}|\mcB_{\cell'}| = O(n^{4/5})$, which bounds the number of low-net-cost executions of the search procedure on $\cell$ by $O(n^{4/5})$, as desired.
By Lemma~\ref{lemma:ratio}, for any cell $\cell'\in\mcC_{\cell}$, $\ell_{\cell'}\in[\frac{1}{3}\ell_\cell, \frac{9}{4}\ell_\cell]$. By the construction of $\mcH$, for any cell $\cell'$ and its grandparent $\cell'_g$, $\ell_{\cell'}\le \frac{2}{3}\ell_{\cell'_g}$. Therefore, for $\cell'\in \mcC_\cell$, only the ancestor of $\cell'$ up to $2\log_{3/2} \frac{27}{4}$ levels higher can also be in $\mcC_\cell$; therefore, for any point $u\in A\cup B$, the point $u$ lies inside at most $O(1)$ cells of $\mcC_\cell$, and $\sum_{\cell'\in\mcC_{\cell}}n_{\cell'}=O(n)$, as claimed.

We conclude that the total number of executions of the extended Hungarian search for each cell $\cell$ of $\mcH$ is $O(n^{4/5})$. The next lemma follows since each execution takes $\tilde{O}(n_\cell\Phi(n))$ time.
\begin{restatable}{lemma}{GlobalTime}\label{lemma:global-time}
    For any cell $\cell$ of $\mcH$, the total execution time of the extended Hungarian search procedure on $\cell$ takes $\tilde{O}(n^{4/5}n_\cell\Phi(n))$ time.
\end{restatable}

\subsection{Extension to higher dimensions}\label{sec:appendix-high-d-ksp}
Given a set of requests $\requests$ in the $d$ dimensions, for any $d\ge 2$, let $\mcG_\requests$ on point sets $A$ and $B$ denote the graph constructed for the $k$-SP problem on $\requests$ under any $\ell_p$ norm for some $p\ge 1$. Let $\mcH_d$ denote the hierarchical partitioning constructed for the point set $A\cup B$ with a parameter $\lambda = 9n^{-\frac{1}{2d+1}}$. 
The hierarchical partitioning has a height $O(d\log(n\Delta))$. Using $\mcH_d$, we run our algorithm as described in Section~\ref{sec:k-seq}. 

We summarize the efficiency analysis of our algorithm for $d$-dimensional point sets next and show that our algorithm runs in $\tilde{O}(n^{2-\frac{1}{2d+1}}\Phi(n)\log \Delta)$ time. In particular, we show that for any cell $\cell$, the number of iterations of the merge step on $\cell$ is $O(n^{1-\frac{1}{2d+1}})$, where each iteration takes $\tilde{O}(n_\cell\Phi(n))$ time. We also show that our algorithm runs the extended Hungarian search procedure in $O(n^{1-\frac{1}{2d+1}})$ time on $\cell$, each in $\tilde{O}(n_\cell\Phi(n))$ time. Adding these bounds for all cells of $\mcH_d$, the total running time of our algorithm is $\tilde{O}(n^{2-\frac{1}{2d+1}}\Phi(n)\log\Delta)$. We summarize the details below.

Given a feasible extended matching $M^\mcC, y(\cdot)$, for any cell $\cell\in \mcC$, let $B^\mcC_\cell$ denote the set of all points of $B_\cell$ that are matched to the divider $\divider\cell$ in the matching before the execution of the merge step at $\cell$. By Lemma~\ref{lemma:local-inactive}, the number of iterations of the merge step on $\cell$ is $O(|B^\mcC_\cell|)$.
To bound the number of points in $B^\mcC_\cell$, we first show that there exists a partial matching $M'$ of high cardinality and low cost. 
\begin{lemma}\label{lemma:geometric-matching-d}
    For any cell $\cell$ of $\mcH_d$, there exists a matching $M'$ over $\mcG_\sigma$ inside $\cell$ that matches all except $O(n_\cell^{1-\frac{1}{2d+1}})$ points of $B_\cell$ and has a cost $O(\ell_\cell n_\cell^{1-\frac{2}{2d+1}})$.
\end{lemma}
\begin{proof}
    Similar to our construction for Lemma~\ref{lemma:geometric-matching}, to construct the matching $M'$ for Lemma~\ref{lemma:geometric-matching-d}, we place a grid of cell-side-length $\ell_\cell n_\cell^{-\frac{2}{2d+1}}$ and compute the matching corresponding to the $1$-partitioning of the requests inside each cell of this grid. It is easy to confirm that the matching $M'$ achieves the bounds claimed in Lemma~\ref{lemma:geometric-matching-d}.
\end{proof}

For the feasible extended matching $M^\mcC=(M, B^\mcC), y(\cdot)$, let $M_\cell$ denote the subset of matching edges of $M$ that lie inside $\cell$. 
Let $\mcP_{\mathrm{aug}}$ (resp. $\mcP_{\mathrm{alt}}$) denote the set of augmenting paths (resp. alternating paths) with an endpoint in $|B^\mcC_\cell|$ in the symmetric difference $M_\cell\oplus M'$.
As discussed above, $|\mcP_{\mathrm{alt}}|\le |F(M')|=O(n^{1-\frac{1}{2d+1}})$, where $F(M')$ denotes the set of free points of $B_\cell$ with respect to $M'$. 
We partition the free endpoints of $\mcP_{\mathrm{aug}}$ into the set $B^{\mathrm{close}}_{\mathrm{aug}}$ (resp. $B^{\mathrm{far}}_{\mathrm{aug}}$) that are at a distance closer than (resp. further than) $\lambda'_\cell=\ell_\cell n^{-\frac{1}{2d+1}}$ to the divider $\divider\cell$ of $\cell$. From Lemma~\ref{lemma:margin},  $|B^{\mathrm{close}}_{\mathrm{aug}}|=O(n^{1-\frac{1}{2d+1}})$. Finally, by Equation~\eqref{eq:net-cost-sum-2}, $\sum_{b\in B^{\mathrm{far}}_{\mathrm{aug}}}y(b) \le w(M')$. Since the dual weight of each free point in $B^{\mathrm{far}}_{\mathrm{aug}}$ is at least $\lambda'_\cell$, we get a bound of $O(n^{1-\frac{1}{2d+1}})$ on the number of such points, and a total execution time of $\tilde{O}(n^{1-\frac{1}{2d+1}}n_\cell\Phi(n))$ for the merge step on a cell $\cell$ of $\mcH$.

Next, we bound the number of executions of the extended Hungarian search procedure for each cell $\cell$.
A selection of $\cell$ by the extended Hungarian search procedure is low-net-cost if $\ymax\le \ell_\cell n^{-\frac{1}{2d+1}}$ and high-net-cost otherwise. For the high-net-cost selections, one can use the matching $M'$ from Lemma~\ref{lemma:geometric-matching-d} and a similar argument as in Lemma~\ref{lemma:high-max} to show that when $\ymax> \ell_\cell n^{-\frac{1}{2d+1}}$, the total number of free points inside $\cell$ cannot exceed $O(n^{1-\frac{1}{2d+1}})$ and conclude an upper bound of $O(n^{1-\frac{1}{2d+1}})$ on the number of high-net-cost selections of $\cell$. 
To bound the number of low-net-cost selections of $\cell$, define $\mcC_\cell$ as the set of all cells $\cell'\in\mcH$ that are processed by the merge step while $\cell$ is in $\mcC$ and $\ymax\le \ell_\cell n^{-\frac{1}{2d+1}}$. For any cell $\cell'\in\mcC_\cell$, let $\mcB_{\cell'}$ denote the set of boundary-matched points of $B_{\cell'}$ matched to the divider of $\cell'$ before the execution of the merge step on $\cell'$. By Lemma~\ref{lemma:global-iterations}, the number of low-net-cost selections of $\cell$ is at most $\sum_{\cell'\in\mcC_\cell}|\mcB_{\cell'}|$. Since $\ymax\le \ell_\cell n^{-\frac{1}{2d+1}}$, for each cell $\cell'\in\mcC_\cell$, all points in $\mcB_{\cell'}$ are at a distance at most $\ell_\cell n^{-\frac{1}{2d+1}}$ to the divider of $\cell'$. Therefore, from Lemma~\ref{lemma:margin} and~\ref{lemma:ratio} and using an identical discussion as above, $\sum_{\cell'\in\mcC_\cell}\mcB_{\cell'} = O(dn^{1-\frac{1}{2d+1}})$. Therefore, the total execution time of the extended Hungarian search step on any cell $\cell$ of $\mcH$ is $\tilde{O}(dn^{1-\frac{1}{2d+1}}n_\cell\Phi(n))$. 

Combining the total execution times of the merge step and the extended Hungarian search step, the running time of our algorithm would be $\tilde{O}(\sum_{\cell\in\mcH}dn^{1-\frac{1}{2d+1}}n_\cell\Phi(n)) = \tilde{O}(dn^{2-\frac{1}{2d+1}}\Phi(n)\log\Delta)$, leading to the following theorem.

\begin{theorem}\label{thm:kSP-d}
Given any sequence $\requests$ (resp. $\requests'=\servers\requests$) of $n$ requests (resp. $n$ requests and $k$ initial server locations) in $d$ dimensions with a spread of $\Delta$, and a value $1 \le k \le n$, there exists a deterministic algorithm that computes the optimal solution for the instance of $k$-SP (resp. $k$-SPI) problem under the $\ell_p$ norm in $\tilde{O}(\min\{nk, n^{2-\frac{1}{2d+1}}\log \Delta\}\cdot\Phi(n))$ time.
\end{theorem}



\section{Missing Proofs of Section~\ref{sec:nk-algod}}
\label{sec:nk-algo-appendix}

\partialoptimaltime*
\begin{proof}
Let $A_F$ (resp. $B_F$) denote the set of free points of $A$ (resp. $B$). Let $y_{\max}:=\max_{b\in B}y(b)$. 
Using Conditions~\eqref{eq:dualfeasibility-matching} and~\eqref{eq:dualfeasibility-free_a-nk}, we rewrite the cost of the matching $M$ as
\begin{align}
    w(M) &= \sum_{(a,b)\in M} \distance{a}{b} = \sum_{(a,b)\in M} y(b) - y(a)\nonumber \\ &= \left(\sum_{b\in B} y(b) - \sum_{a\in A} y(a)\right) - \sum_{b\in B_F} y(b) + \sum_{a\in A_F} y(a)\nonumber \\ &= \left(\sum_{b\in B} y(b) - \sum_{a\in A} y(a)\right) - |B_F|\cdot y_{\max}.
    \label{eq:dualOptimal-proof-1-1}
\end{align}
Let $M^*$ denote any minimum-cost $t$-matching on $\mcG_\requests$. Let $A^*_F$ (resp. $B^*_F$) denote the set of points of $A$ (resp. $B$) that are free in $M^*$. Since both $M$ and $M^*$ are $t$-matchings, $|B_F|=|B^*_F|$. Using Condition~\eqref{eq:dualfeasibility-non-matching},
\begin{align}
    w(M^*) &= \sum_{(a,b)\in M^*} \distance{a}{b} \ge \sum_{(a,b)\in M^*} y(b) - y(a)\nonumber \\ &= \left(\sum_{b\in B} y(b) - \sum_{a\in A} y(a)\right) - \sum_{b\in B^*_F} y(b) + \sum_{a\in A^*_F} y(a)\nonumber \\ &\ge \left(\sum_{b\in B} y(b) - \sum_{a\in A} y(a)\right) -  |B^*_F|\cdot y_{\max},
    \label{eq:dualOptimal-proof-2-1}
\end{align}
where the last inequality holds since $y(b)\le y_{\max}$ for each point $b\in B$ and $y(a)\ge 0$ for each point $a\in A$. Combining Equations~\eqref{eq:dualOptimal-proof-1-1} and~\eqref{eq:dualOptimal-proof-2-1},
\[w(M) = \sum_{b\in B} y(b) - \sum_{a\in A} y(a) - |B_F|\cdot y_{\max} \le w(M^*). \]
Since $M^*$ is a minimum-cost $t$-matching, $w(M)=w(M^*)$, and $M$ is also a minimum-cost $t$-matching.
\end{proof}

\nkAlgo*
\begin{proof}
For any $1\le t\le k$, let $M_t, y_t(\cdot)$ denote the $(n-t)$-matching computed by our algorithm, i.e., the matching $M_t, y_t(\cdot)$ is computed from $M_{t-1}, y_{t-1}(\cdot)$ by executing the \reverse\ procedure and reducing $M_{t-1}$ along the path $P$ returned by the \reverse\ procedure.
By our initial dual assignments, $M_1, y_1(\cdot)$ is a dual-optimal $(n-1)$-matching. Below, assuming that $M_{t-1}, y_{t-1}(\cdot)$ is dual-optimal, we show that $M_{t}, y_{t}(\cdot)$ is also dual-optimal. To do so, first, we show that $M_{t-1}, y_t(\cdot)$ is dual-optimal and the path $P$ is admissible with respect to the updated dual weights $y_t(\cdot)$. We then conclude Lemma~\ref{lemma:nkAlgo} by showing that the matching $M_t$ obtained by reducing $M_{t-1}$ along $P$ remains dual-optimal along with $y_t(\cdot)$.

For any edge $(a,b)\in E$, let $s_{t-1}(a,b)$ denote the slack of $(a,b)$ with respect to $y_{t-1}(\cdot)$. For any matching edge $(a,b)\in M_{t-1}$, $\kappa_b = \kappa_a$ since the only edge directed to $a$ in the reversed residual graph is the zero-slack edge $(a,b)$. Thus, Condition~\eqref{eq:dualfeasibility-matching} holds since
\begin{equation*}\label{eq:hung-0}
y_t(b) - y_t(a) = (y_{t-1}(b)-\kappa+\kappa_b) - (y_{t-1}(a)-\kappa+\kappa_a) = y_{t-1}(b)-y_{t-1}(a) = \distance{a}{b}.
\end{equation*}
Similarly, for any non-matching edge $(b,a)\in E$, $\kappa_b\le \kappa_a + s_{t-1}(b,a)$, and Condition~\eqref{eq:dualfeasibility-non-matching} holds since
\begin{equation*}
    y_t(b) - y_t(a) = (y_{t-1}(b)-\kappa+\kappa_b) - (y_{t-1}(a)-\kappa+\kappa_a) \le y_{t-1}(b)-y_{t-1}(a) + s_{t-1}(a,b) = \distance{a}{b}, 
\end{equation*}
and the equality holds for the edges on Dijkstra's shortest path tree; consequently, the path $P$ is admissible.
Finally, we show that the free points of $B$ have the highest dual weight among all points. Let $y_{\max}:=\max_{b'\in B}y_{t-1}(b')$. From the fact that $M_{t-1}, y_{t-1}(\cdot)$ is dual-optimal, for any free point $b\in B$, $y_{t-1}(b)=y_{\max}$; therefore, $\kappa_b=0$ since there is a zero-cost edge from source to $b$. Thus, 
\begin{equation}\label{eq:hung-2}
y_t(b) = y_{t-1}(b) - \kappa + \kappa_b = y_{\max} - \kappa.
\end{equation}
As a result, all free points of $B$ have the same dual weight of $y_{\max} - \kappa$ in $y_t(\cdot)$.
For any matched point $b\in B$, if $y_{t-1}(b)\ge y_{\max} - \kappa$, then the edge from the source to $b$ has a cost $y_{\max} - y_{t-1}(b)<\kappa$, and therefore, $\kappa_b\le y_{\max} - y_{t-1}(b)\le \kappa$; hence, 
\[y_t(b) = y_{t-1}(b)-\kappa+\kappa_b\le y_{t-1}(b)-\kappa+y_{\max}-y_{t-1}(b) = y_{\max}-\kappa.\]

Finally, note that there are no paths in the reversed residual graph from the source to the free points $a\in A$ with respect to $M_{t-1}$, and therefore, $y_t(a)=y_{t-1}(a)=0$. Additionally, for the path $P$ returned by the procedure, if $P$ ends at a point $a\in A$, then $\kappa = \kappa_a+y_{t-1}(a)$ and $y_t(a)=y_{t-1}(a)-\kappa+\kappa_a = 0$. Hence, $M_t, y_t(\cdot)$ is a dual-optimal matching. 
\end{proof}

\section{Bipartite Matching on Randomly Colored Points}

\subsection{Missing Proofs and Details of Section~\ref{sec:randomly-colored}}\label{appendix-randomly-colored}

\myparagraph{Constructing the Partial Matching.} In this part, for any cell $\cell$ of $\mcH$, we construct a matching $M'$ of $A_\cell\cup B_\cell$ that, in expectation, matches all except $\tilde{O}(n_\cell^{3/4})$ points of $B_\cell$ and has a cost $\tilde{O}(\ell_\cell n_\cell^{-1/4})$, proving Lemma~\ref{lemma:GRS-matching}.
We begin by introducing a set of notations. Let $\cell$ be any cell of $\mcH$, and let $G$ be a grid dividing $\cell$ into smaller squares.
For any square $\xi\in G$, let $A_\cell^\xi$ (resp. $B_\cell^\xi$) denote the subset of points of $A_\cell$ (resp. $B_\cell$) that lie inside $\xi$. Define the \emph{excess} of $\xi$ as $\excess(\xi):=\big||B_\cell^\xi| - |A_\cell^\xi|\big|$. Define $\excess(G):=\sum_{\xi\in\mcG}\excess(\xi)$ as the excess of $G$. We next show an important property of randomly colored point sets, which is critical in constructing the matching.

\begin{lemma}\label{lemma:convergence_excess}
    For any square $\cell$ of $\mcH$ and a grid $G$ inside $\cell$ with cell side length $O(\ell_\cell n_\cell^{-\alpha})$, $\mbE[\excess(G)]= \tilde{O}(n_\cell^{\alpha + \frac{1}{2}})$. 
\end{lemma}
\begin{proof}
    If $\alpha\ge \frac{1}{2}$, the lemma statement holds trivially since $n_\cell^{\alpha+\frac{1}{2}} \ge n_\cell$. Therefore, we assume $\alpha\le \frac{1}{2}$.
    Using the Hoeffding’s inequality~\cite{hoeffding1994probability}, for any hypercube $\xi$ of $G$,
    \begin{equation*}
        \prob{\big||B_\cell^\xi| - |A_\cell^\xi|\big|\ge c_1\sqrt{n_\cell log n}}\le n^{-c_2}\label{eq:excess}
    \end{equation*}
    for some constants $c_1, c_2 > 1$. Since $\excess(\xi)\le n_\cell$, 
    \begin{equation*}
        \mbE[\excess(\xi)]= O(\sqrt{n_\cell\log n}).\label{eq:excess_cell}
    \end{equation*}
    Summing over all squares of $G$,
    \begin{align*}
        \mbE[\excess(G)] = \sum_{\xi\in G} \mbE[\excess(\xi)] &=O\left(\sum_{\xi\in G} \sqrt{n_\cell\log n}\right) = O\left(\sqrt{n\log n\times |G|}\right) = \tilde{O}(n^{\alpha+\frac{1}{2}}).\label{eq:excess-dense}
    \end{align*}
\end{proof}

We next use Lemma~\ref{lemma:convergence_excess} to show that there exists a low-cost high-cardinality matching inside each sub-problem.

\GRSPartialMatching*
\begin{proof}
    Let $\langle G_1, \ldots, G_t\rangle$ denote a sequence of grids, where $t=\lceil \log\log n\rceil$ and each grid $G_i$ has a side-length $O(\ell_\cell n_\cell^{-\alpha_i})$ for \[\alpha_i:=\frac{1}{2} - \frac{2^t}{2^{t+2} - 3}\left(1-\frac{1}{2^i}\right).\]
    Using these grids, we construct a matching $M'$ as follows.
    Let $A_\cell^0:=A_\cell$ and $B_\cell^0:=B_\cell$. Starting from $i=1$, we compute a matching $M_i$ from $B_\cell^{i-1}$ to $A_\cell^{i-1}$ that for each square of $G_i$, matches as many points as possible inside the square arbitrarily. Define $A_\cell^i$ (resp. $B_\cell^i$) as the set of free points of $A_\cell^{i-1}$ (resp. $B_\cell^{i-1}$) and process the next grid $G_{i+1}$. We continue this procedure until the last grid $G_t$ is processed. Define $M':=\sum_{i=1}^t M_i$. This completes the construction of $M'$. In the following, we first show that the total number of free points with respect to $M'$ is $O(n_\cell^{3/4})$; we then show that $w(M')=O(\ell_\cell^2n_\cell^{1/4})$ and conclude the lemma statement.

    The matching $M$ matches as many points as possible inside each square of the grid $G_t$. Therefore, the total number of unmatched points is equal to the excess of $G_t$, which by Lemma~\ref{lemma:convergence_excess} is 
    \begin{align}
        \mbE[\excess(G_t)] &= \tilde{O}(n_\cell^{\alpha_t+\frac{1}{2}}) = O(n_\cell^{\frac{3}{4}+\frac{1}{4(2^{t+2}-3)}}) = O(n_\cell^{\frac{3}{4}+\frac{1}{16\log n-12}}) = O(n_\cell^{\frac{3}{4}}).
    \end{align}
    We next analyze the expected cost of $M'$. By the linearity of expectation,
    \begin{equation}
        \mbE[w(M')] = \sum_{i=1}^t \mbE[w(M_i)].
        \label{eq:expected_cost_1}
    \end{equation}
    For $i=1$, since all matching edges in $M_1$ have a squared Euclidean cost at most $O((\ell_\cell n_\cell^{-\alpha_1})^2)$, the cost of $M_1$ would be 
    \begin{equation}
        w(M_1) = O(n_\cell\times (\ell_\cell n_\cell^{-\alpha_1})^2) = O(\ell_\cell^2 n_\cell^{\frac{2^t}{2^{t+2}-3}}) = O(\ell_\cell^2 n_\cell^{\frac{1}{4}}).\label{eq:expected_cost_2}
    \end{equation}
    Finally, for each $1<i\le t$, the matching $M_i$ consists of matching edges with squared Euclidean cost of at most $O((\ell_\cell n_\cell^{-\alpha_i})^2)$. By Lemma~\ref{lemma:convergence_excess}, the expected number of matching edges in $M_i$ is at most $\tilde{O}(n_\cell^{\alpha_{i-1}+\frac{1}{2}})$; therefore,
    \begin{equation}
        \mbE[w(M_i)] = \tilde{O}(n_\cell^{\alpha_{i-1}+\frac{1}{2}}\times (\ell_\cell n_\cell^{-\alpha_i})^2) = \tilde{O}(\ell_\cell^2 n_\cell^{\frac{2^t}{2^{t+2}-3}}) = \tilde{O}(\ell_\cell^2 n_\cell^{\frac{1}{4}}).\label{eq:expected_cost_3}
    \end{equation}
    Combining Equations~\eqref{eq:expected_cost_1},~\eqref{eq:expected_cost_2}, and~\eqref{eq:expected_cost_3}, $\mbE[w(M)] = \tilde{O}(\ell_\cell^2 n_\cell^{\frac{1}{4}})$.
\end{proof}

\paragraph{Bounding the number of iterations.} We next show that the number of iterations of the merge step and the extended Hungarian search procedure on any cell $\cell$ of $\mcH$ is bounded by $O(n^{3/4})$.

\matchingIterations*
\begin{proof}
    From Lemma~\ref{lemma:local-inactive}, the total number of iterations of the merge step on $\cell$ is $O(|B_\cell^\mcC|)$.
    Next, note that the parameter $k$ is set to $0$; therefore, there are no free points inside $\cell$ when the merge step is executed on $\cell$. When erasing the divider of $\cell$, our algorithm creates $|B_\cell^\mcC|$ free points. Each iteration of the merge step either (i) finds an augmenting path, which reduces the number of free points by one, or (ii) finds an alternating path to a matched point $b\in B_{\cell}$ with $y(b)=\ymax$, which does not change the number of free points inside $\cell$. Therefore, after the merge step on $\cell$, the number of free points remaining inside $\cell$ is at most $|B_\cell^\mcC|$, and our algorithm executes one extended Hungarian search procedure for each remaining free point.
\end{proof}

\GRSIters*
\begin{proof}
Let $M'$ denote the matching constructed in Lemma~\ref{lemma:GRS-matching}, and let $M$ denote the matching of the extended matching $M^\mcC$ maintained by our algorithm. Let $M_\cell$ denote the matching edges of $M$ that lie inside $\cell$. Note that by lemma~\ref{lemma:no-cross}, no matching edges can cross the boundaries of $\mcC$.
Let $\mcP_{\mathrm{aug}}$ (resp. $\mcP_{\mathrm{alt}}$) denote the set of augmenting paths (resp. alternating paths) with an endpoint in $B_\cell^\mcC$ in the symmetric difference $M_\cell\oplus M'$. Clearly, $|B_\cell^\mcC|=|\mcP_{\mathrm{alt}}| + |\mcP_{\mathrm{aug}}|$. For the alternating paths, $|\mcP_{\mathrm{alt}}|=O(n^{3/4})$ since each alternating path has one free endpoint in $M'$. 
Next, we show that $|\mcP_{\mathrm{aug}}|=O(n^{3/4})$.

As shown in Equation~\eqref{eq:net-cost-sum-2},
\begin{align}
    \sum_{b\in B_{\mathrm{aug}}} y(b) = \sum_{P\in \mcP_{\mathrm{aug}}} y(b_P)\le \sum_{P\in \mcP_{\mathrm{aug}}} \phi(P) \le w(M').\label{eq:net-cost-sum-1-m}
\end{align}

Define $\alpha:=\ell_\cell n^{-1/4}$. Each free point $b\in B_{\mathrm{aug}}$ is called a \emph{close} (resp. \emph{far}) point if the Euclidean distance of $b$ to the divider of $\cell$ is at most (resp. more than) $\alpha$. Let $B_{\mathrm{aug}}^{\mathrm{close}}$ (resp. $B_{\mathrm{aug}}^{\mathrm{far}}$) denote the set of all close (resp. far) points of $B_{\mathrm{aug}}$. By Lemma~\ref{lemma:margin},
\begin{equation}\label{eq:close-points-m}
    |B_{\mathrm{aug}}^{\mathrm{close}}| = O(n_\cell n^{-1/4}) = O(n^{3/4}).
\end{equation}
For each far point $b\in B_{\mathrm{aug}}^{\mathrm{far}}$, $y(b)=\distance{b}{\divider\cell}\ge \alpha^2$ since $b$ is matched to the divider $\divider\cell$ (note that the distance function $\distance{\cdot}{\cdot}$ is the squared Euclidean distance). Therefore, 
\begin{equation}\label{eq:net-cost-3-m}
    \sum_{b\in B_{\mathrm{aug}}} y(b) \ge \sum_{b\in B_{\mathrm{aug}}^{\mathrm{far}}} y(b) \ge \alpha^2\times |B_{\mathrm{aug}}^{\mathrm{far}}|.
\end{equation}
Combining Equations~\eqref{eq:net-cost-sum-1-m} and~\eqref{eq:net-cost-3-m},
\begin{equation}\label{eq:far-points-m}
    |B_{\mathrm{aug}}^{\mathrm{far}}| \le \frac{\sum_{b\in B_{\mathrm{aug}}} y(b)}{\alpha^2} \le \frac{w(M')}{\alpha^2} = O(n^{3/4}).
\end{equation}
Combining with Equation~\eqref{eq:close-points-m},
\begin{align*}
    |B_\cell^\mcC| &\le |B_{\mathrm{aug}}| + |B_{\mathrm{alt}}| \le |B_{\mathrm{aug}}^{\mathrm{close}}| + |B_{\mathrm{aug}}^{\mathrm{far}}| + |B_{\mathrm{alt}}| = \tilde{O}(n^{3/4}).
\end{align*}
\end{proof}

\subsection{Analysis for General \texorpdfstring{$d$}{} and \texorpdfstring{$q$}{}}\label{sec:appendix-GRS}
For a point set $U$ of $2n$ points in $d$ dimensions and any $q\ge 1$, let $A$ denote a random subset of $n$ points of $U$ and let $B=U\setminus A$. 
To compute a minimum-cost perfect matching between $A$ and $B$ under $\ell_2^q$ distances, we construct our hierarchical partitioning $\mcH$ with a parameter $\lambda=9n^{-\frac{1}{d+2}}$ and execute our algorithm from Section~\ref{sec:k-seq} by setting $k=0$. In the following, we extend our analysis from Section~\ref{sec:randomly-colored} to any dimension $d\ge 2$ and any $q\ge 1$ and show the following result.

\begin{restatable}{lemma}{GRSgeneral}\label{lemma:GRS-matching-general}
    Suppose $U$ is a set of $2n$ points inside the unit $d$-dimensional hypercube, $d\ge 2$, and $A$ is a subset chosen uniformly at random from all subsets of size $n$. Let $B=U\setminus A$. Then, for any parameters $q\ge 1$, the expected running time of our algorithm for computing the minimum-cost matching on the complete bipartite graph on $A$ and $B$ under $\ell_2^q$ costs is
    \begin{equation*}
    \begin{cases}
        \tilde{O}(n^{2-\frac{q}{(q+1)d}}\Phi(n)\log \Delta),\qquad &q\le\frac{d}{2},\\
        \tilde{O}(n^{2-\frac{1}{d+2}}\Phi(n)\log \Delta),\qquad &q>\frac{d}{2}.
    \end{cases}
\end{equation*}
\end{restatable}

Note that when $q\le \frac{d}{2}$, the bound claimed in Lemma~\ref{lemma:GRS-matching-general} is identical to the runtime bound proven in~\cite[Theorem B.2]{gattani2023robust} and their analysis directly applies to our algorithm; hence, in the following, we analyze the running time of our algorithm assuming $q>\frac{d}{2}$.

We begin by showing that there exists a low-cost high-cardinality partial matching.

\begin{lemma}\label{lemma:GRS-matching-d}
    For any cell $\cell$ of $\mcH$, there exists a matching $M'$ that, in expectation, matches all except $O(n_\cell^{1-\frac{1}{d+2}})$ points of $B_\cell$ and has a cost $O((2\ell_\cell)^q n_\cell^{1-\frac{q+1}{d+2}})$.
\end{lemma}
\begin{proof}
Define $r:= \frac{2q}{d}>1$ and $t:=\lceil \log_r \log_2 n\rceil$. Let $\beta:=q(d+2)r^t - d(q+1)$. Define a sequence of grids $\langle G_1, \ldots, G_t\rangle$, where each grid $G_i$ has a side-length $O(\ell_\cell n^{-\alpha_i})$ for \[\alpha_i:=\frac{1}{d} - \frac{r^{t+1}}{\beta}\left(1-\frac{1}{r^i}\right).\]

We construct a matching $M$ as follows.
Let $A_\cell^0:=A_\cell$ and $B_\cell^0:=B_\cell$. Starting from $i=1$, let $M_i$ be a matching from $B_\cell^{i-1}$ to $A_\cell^{i-1}$ that matches as many points as possible inside each hypercube of $G_i$. Let $A_\cell^i$ and $B_\cell^i$ denote the set of free points of $A_\cell^{i-1}$ and $B_\cell^{i-1}$, respectively. Set $i\leftarrow i+1$ and continue this procedure until the last grid $G_t$ is processed. Define $M:=\sum_{i=1}^t M_i$. 
In the following, we show that the total number of free points with respect to $M$ is $\tilde{O}(n_\cell^{1-\frac{1}{d+2}})$ and the cost of $M$ is $\mbE[w(M)]=O((2\ell_\cell)^qn_\cell^{1-\frac{q+1}{d+2}})$.

The matching $M$ matches as many points as possible inside each square of the grid $G_t$. Therefore, the total number of unmatched points is equal to the excess of $G_t$, which by Lemma~\ref{lemma:convergence_excess-d} below is 
\begin{align}
    \mbE[\excess(G_t)] = \tilde{O}(n_\cell^{\frac{d}{2}\alpha_t+\frac{1}{2}}) = \tilde{O}(n_\cell^{1 - \frac{1}{d+2}}), \label{eq:expected_free_general}
\end{align}
where the second equality holds since
\begin{align*}
    \frac{d}{2}\alpha_t+\frac{1}{2} &= \frac{d}{2}\left(\frac{1}{d} - \frac{r^{t+1}}{\beta}(1-\frac{1}{r^t}) \right) + \frac{1}{2}= \left(\frac{1}{2} - \frac{qr^{t}-q}{q(d+2)r^t - d(q+1)} \right) + \frac{1}{2}\\ &= 1 - \frac{1}{d+2} + \frac{q-\frac{d(q+1)}{d+2}}{q(d+2)r^t - d(q+1)} \le 1 - \frac{1}{d+2} + \frac{1}{2\log n}. 
\end{align*}
We next analyze the expected cost of $M$. For $i=1$, since all matching edges in $M_1$ has a cost of at most $O((\ell_\cell n_\cell^{-\alpha_1})^q)$, the cost of $M_1$ would be 
\begin{equation}
    w(M_1) = O(n_\cell\times (\ell_\cell n_\cell^{-\alpha_1})^q) = O((2\ell_\cell)^q n_\cell^{1-\frac{q+1}{d+2}}),\label{eq:expected_cost_2_general}
\end{equation}
where the second equality holds since
\begin{align*}
    1-q\alpha_1 &= 1-\frac{q}{d}+\frac{q(r-1)r^t}{\beta}= 1-\frac{q}{d}+\frac{\frac{2q^2-qd}{d}r^t}{q(d+2)r^t - d(q+1)} \\ &= 1-\frac{q}{d}+\frac{2q-d}{d(d+2)} + \frac{\frac{(2q-d)(q+1)}{d+2}}{q(d+2)r^t - d(q+1)}\le 1-\frac{q+1}{d+2} + \frac{q}{d\log n}.
\end{align*}
Finally, for each $1<i\le t$, the matching $M_i$ consists of matching edges with a cost of $O((\ell_\cell n_\cell^{-\alpha_i})^q)$. By Lemma~\ref{lemma:convergence_excess-d}, the expected number of matching edges in $M_i$ is at most $\tilde{O}(n_\cell^{\frac{d}{2}\alpha_{i-1}+\frac{1}{2}})$; therefore,
\begin{equation}
    \mbE[w(M_i)] = \tilde{O}(n_\cell^{\frac{d}{2}\alpha_{i-1}+\frac{1}{2}}\times (\ell_\cell n_\cell^{-\alpha_i})^q) = \tilde{O}((2\ell_\cell)^q n_\cell^{1-\frac{q+1}{d+2}}),\label{eq:expected_cost_3_general}
\end{equation}
where the second equality is resulted from as follows.
\begin{align*}
    \frac{d}{2}\alpha_{i-1} + \frac{1}{2}-q\alpha_i &= \left(\frac{1}{2}-\frac{qr^t(1-\frac{1}{r^{i-1}})}{\beta}\right) + \frac{1}{2} - \left(\frac{q}{d}-\frac{qr^{t+1}(1-\frac{1}{r^{i}})}{\beta}\right)\\ &=  1-\frac{q}{d}+\frac{q(r-1)r^t}{\beta}\le 1-\frac{q+1}{d+2} + \frac{q}{d\log n}.
\end{align*}
Combining Equations~\eqref{eq:expected_cost_2_general} and~\eqref{eq:expected_cost_3_general}, the expected cost of $M$ is \[\mbE[w(M)]=\sum_{i=1}^t\mbE[w(M_i)]=O((2\ell_\cell)^qn_\cell^{1-\frac{q+1}{d+2}}).\]
\end{proof}

\begin{lemma}\label{lemma:convergence_excess-d}
    For any random partitioning of a set $U$ of $2n$ points inside the $d$-dimensional unit hypercube into two sets $A$ and $B$ of $n$ points each, a hypercube $\cell$ of $\mcH$, and a grid $G$ inside $\cell$ with cell side length $O(\ell_\cell n_\cell^{-\alpha})$, $\mbE[\excess(G)]= \tilde{O}(n_\cell^{\frac{d}{2}\alpha + \frac{1}{2}})$. 
\end{lemma}
\begin{proof}
    For $\alpha\ge \frac{1}{d}$, the lemma statement holds trivially since $n_\cell^{\frac{d}{2}\alpha+\frac{1}{2}} \ge n_\cell$. Therefore, we assume $\alpha\le \frac{1}{d}$.
    Using the Hoeffding’s inequality~\cite{hoeffding1994probability}, for any hypercube $\xi$ of $G$,
    \begin{equation*}
        \prob{|B_\cell^\xi| - |A_\cell^\xi|\ge c_1\sqrt{n_\cell log n}}\le n^{-c_2}\label{eq:excess-d}
    \end{equation*}
    for some constants $c_1, c_2 > 1$. Since $\excess(\xi)\le n_\cell$, 
    \begin{equation*}
        \mbE[\excess(\xi)]= O(\sqrt{n_\cell\log n}).
    \end{equation*}
    Summing over all squares of $G$,
    \begin{align*}
        \mcB[\excess(G)] = \sum_{\xi\in G} \mbE[\excess(\xi)] &=O\left(\sum_{\xi\in G} \sqrt{n_\cell\log n}\right) = O\left(\sqrt{n\log n\times |G|}\right) = \tilde{O}(n^{\frac{d}{2}\alpha+\frac{1}{2}}).\label{eq:excess-dense-d}
    \end{align*}
\end{proof}

We use Lemma~\ref{lemma:GRS-matching-d} to show that the number of points in $B_\cell^\mcC$ is $\tilde{O}(n^{1-\frac{1}{d+2}})$. Let $M'$ denote the matching constructed in Lemma~\ref{lemma:GRS-matching-d}, and let $M$ denote the matching of the extended matching $M^\mcC$ maintained by our algorithm. Let $M_\cell$ denote the matching edges of $M$ that lie inside $\cell$. 
Let $\mcP_{\mathrm{aug}}$ (resp. $\mcP_{\mathrm{alt}}$) denote the set of augmenting paths (resp. alternating paths) with an endpoint in $|B_\cell^\mcC|$ in the symmetric difference $M_\cell\oplus M'$.
As discussed above, $|\mcP_{\mathrm{alt}}|\le |F(M')|=O(n^{1-\frac{1}{d+2}})$. 
We partition the free endpoints of $\mcP_{\mathrm{aug}}$ into the set $B^{\mathrm{close}}_{\mathrm{aug}}$ that are at a Euclidean distance closer than $\lambda'_\cell=2\ell_\cell n^{-\frac{1}{d+2}}$ to the divider $\divider\cell$ of $\cell$ and the set $B^{\mathrm{far}}_{\mathrm{aug}}$ that are at a Euclidean distance further than $\lambda'_\cell$ from $\divider\cell$. From Lemma~\ref{lemma:margin},  $|B^{\mathrm{close}}_{\mathrm{aug}}|=\tilde{O}(n^{1-\frac{1}{d+2}})$. Finally, by Equation~\eqref{eq:net-cost-sum-1-m}, $\sum_{b\in B^{\mathrm{far}}_{\mathrm{aug}}}y(b) \le w(M')$. Since the dual weight of each free point in $B^{\mathrm{far}}_{\mathrm{aug}}$ is at least $(\lambda'_\cell)^q$, we get a bound of $\tilde{O}(n^{1-\frac{1}{d+2}})$ on the number of such points. 

Therefore, the total number of iterations of the merge and extended Hungarian search processes on each cell $\cell$ is $\tilde{O}(n^{1-\frac{1}{d+2}})$. Each iteration of the merge step and the search procedure takes $\tilde{O}(n_\cell\Phi(n))$ time; as a result, the total execution time of our algorithm on $\cell$ would be $\tilde{O}(n^{1-\frac{1}{d+2}}n_\cell\Phi(n))$. Since each point participates in $O(\log(n\Delta))$ cells in $\mcH$, the total execution time of our algorithm across all cells would be $\tilde{O}(n^{2-\frac{1}{d+2}}\Phi(n)\log\Delta)$, proving Lemma~\ref{lemma:GRS-matching-general}.

\end{document}